%% file: ZT_simplified_appendixreductions.tex
\documentclass[11pt]{article}
\usepackage[margin=1.1in]{geometry}
\usepackage{amsmath,amssymb,amsthm}
\usepackage{microtype}
\usepackage{xcolor}
\usepackage{tikz}
\usetikzlibrary{decorations.pathreplacing,calc,arrows.meta}
\usepackage{enumitem}
\usepackage{array}
\usepackage{hyperref}
\hypersetup{colorlinks=true,linkcolor=blue!50!black,citecolor=blue!50!black,urlcolor=blue!50!black}

\newtheorem{theorem}{Theorem}
\newtheorem{lemma}[theorem]{Lemma}
\newtheorem{corollary}[theorem]{Corollary}

\theoremstyle{definition}
\newtheorem{definition}[theorem]{Definition}

\theoremstyle{remark}
\newtheorem{remark}[theorem]{Remark}

\newcommand{\Z}{\mathbb{Z}}
\newcommand{\eps}{\varepsilon}
\newcommand{\val}{\mathrm{val}}
\newcommand{\Lop}{\mathrm{Lop}\text{-}\mathrm{AE}\text{-}\mathrm{SparseTri}}
\DeclareMathOperator{\Hh}{H}
\DeclareMathOperator{\polyop}{poly}
\DeclareMathOperator{\polylog}{polylog}
\makeatletter
\newcommand{\poly}{\@ifnextchar\log{\poly@log}{\polyop}}
\def\poly@log\log{\polylog}
\makeatother
\newcommand{\Leaves}{\mathrm{Leaves}}
\newcommand{\Cm}{M}
\newcommand{\Nz}{N_0}
\newcommand{\Mult}{\textnormal{\textsc{Mult}}}
\newcommand{\Pruned}{\textnormal{\textsc{Pruned}}}
\newcommand{\Full}{\textnormal{\textsc{Full}}}
\newcommand{\defn}[1]{\textbf{\emph{#1}}}

\title{Truly Subquadratic 3SUM and Truly Subcubic APSP \\via Triangles in Sparse Lopsided Graphs}
\author{Josh Alman\\ Columbia University \and Virginia Vassilevska Williams\\ MIT}
\date{}

\begin{document}
\maketitle

\begin{abstract}
We give the first polynomial improvements over the textbook algorithms for $3$SUM and All-Pairs Shortest Paths (APSP): we show how to deterministically solve $3$SUM on $n$ integers of polynomial size in  $O(n^{1.9992})$ time and APSP on directed $n$-vertex graphs with polynomially bounded integer weights in $O(n^{2.9995})$ time. This refutes the $3$SUM and APSP hypotheses of fine-grained complexity theory. Using known reductions, we also refute the real-valued versions of the $3$SUM and APSP hypotheses, the Exact Triangle hypothesis, the Zero-Weight $k$-Clique hypotheses, and the three rectangular hinted Online Matrix--Vector conjectures of van den Brand, Nanongkai, and Saranurak, and we give polynomial speedups for a variety of other problems. 

All of these results follow from a single new algorithm for thin matrix products. Let $X$ be an $N\times D$ integer matrix and $Y$ a $D\times N$ integer matrix with $D\le N^{1/18}$, and let $W$ be any set of at most $N^2/\sqrt D$ positions. We compute the entries $(XY)[I,J]$, $(I,J)\in W$, in $O(N^2/D^{0.063})$ operations, which is polynomially less than the time needed to write down $XY$ or to compute $N^2/\sqrt D$ inner products one by one. We design this algorithm by modifying a variant of Coppersmith's rectangular matrix multiplication algorithm, built from a ten-multiplication identity of Sch\"onhage, to perform only the operations needed for the entries in $W$, and show that few operations are needed. Interpreted as a graph algorithm, this solves the All-Edges Sparse Triangle problem in truly subquadratic time on sparse lopsided tripartite graphs where two parts have $n$ vertices but one part has $n^{\varepsilon}$ vertices for $\varepsilon<0.12$. By known reductions, Exact Triangle, and hence $3$SUM and APSP, reduce to this problem. We also give a data structure version that answers queries for single entries of $XY$, not known in advance.

We view this as a great success for fine-grained complexity, whose main focus has been reductions that uncover connections between problems. So far these reductions have mostly served to infer conditional hardness, since $3$SUM and APSP themselves resisted improvement for so long. But reductions are also algorithms, and they allow a single new technique to speed up a wide range of problems at once. This opens new directions in both algorithm design and complexity.

\vspace{12pt}

\textit{Claude, an AI model developed by Anthropic, discovered the algorithm that refutes the $3$SUM, APSP, and Exact Triangle hypotheses. The authors then worked to understand, simplify, strengthen, and extend the algorithm, derive additional consequences, and make the presentation accessible. See ``Acknowledgments and Methodology'' for how the result was found and shared with the authors. The authors take full responsibility for this paper.}

\vspace{12pt}

\textit{Claude also verified this paper's main results using the Lean 4 proof assistant with the Mathlib library.}
\end{abstract}

\newpage
\input{01-intro}
\input{02-matrix}
\input{03-reduction}
\input{04-general}
\input{05-other}

\input{06-conclusion}

\section*{Acknowledgments and Methodology}

\paragraph{How the result was originally found and shared with the authors.} An Anthropic employee used an internal research model to investigate open problems in the theory of cryptography. One of them was about cryptographic constructions based on the average-case hardness of Zero-$k$-Clique~\cite{LLV19,AHY25}. Claude was tasked with verifying and improving the constructions, but instead developed this algorithm, first for the average case, then for the worst case. The session used 16M output tokens with no human input.

Anthropic shared the algorithm with the authors in September 2026 under a confidentiality agreement, offered compensation, 
and provided access to the public version of Claude. 

\paragraph{Shared algorithm.}
The algorithm that was shared with the authors is essentially the algorithm in Section~\ref{sec:result}, although presented differently and with other numerical parameters. A different reduction from Exact Triangle to the problem of Section~\ref{sec:result} was also shared, although the authors did not include it here, and instead used known reductions to All-Edges Sparse Triangle, which suffice and strengthen the connection. The authors derived the data structure version in Section~\ref{sec:general} and its connection with Hinted OMv. 

\paragraph{Lean certification.} After the paper was completely written, Anthropic used an internal research model to certify this paper's main results using the Lean 4 proof assistant with the Mathlib library. The formalization is available at \url{https://github.com/anthropics/formal-math/tree/main/3sum-apsp}. The following are formalized there: Theorem~\ref{thm:zwt} (Exact Triangle in $O(n^{2.9983})$ time), Theorem~\ref{thm:3sumapsp} ($3$SUM in $O(n^{1.9992})$ time, and the $(\min,+)$-product and APSP in $O(n^{2.99942})$ time), and the Zero-Weight case of Corollary~\ref{cor:kclique} (Zero-Weight $k$-Clique in $O(n^{k-0.0017\lfloor k/3\rfloor})$ time for every $k\ge3$). All lemmas and prior results that these statements rely on are also formalized there.

\paragraph{Further acknowledgments.}
The authors used Claude to help with writing, figure creation, and checking mathematical details throughout this project. They have attempted to make the entire paper accessible and well-written, and take full responsibility for its contents.

The authors' contribution to this publication was done in their individual capacities, 
not in connection with their duties
or responsibilities to Columbia University or to MIT. Any opinions expressed in this work are opinions of the authors; they are not official positions or endorsements of Columbia University or of MIT, and do not represent the research of the authors at Columbia University or MIT.

\bibliographystyle{alpha}
\bibliography{references}

\end{document}

%% file: 01-intro.tex
\section{Introduction}\label{sec:intro}

Since its inception, algorithms research has sought the fastest possible algorithm for every problem it considers. Powerful techniques have been developed, from variants of dynamic programming to surprising algebraic tools such as the Fast Fourier Transform~\cite{CT65} and fast matrix multiplication~\cite{Str69}. Yet for many problems, these techniques have not sufficed. Some examples include:

\begin{itemize}
\item \textbf{$3$SUM}: Given $n$ numbers, decide whether three of them sum to $0$. 
This is a classical problem with a long history, and it is central to computational geometry (see~\cite{GO95}). Despite many decades of research, the $O(n^2)$-time algorithm taught in algorithms classes has only been sped up by polylogarithmic factors~\cite{BDP08,GP18,Cha20}.

\item \textbf{All-Pairs Shortest Paths} (APSP): Given an edge-weighted graph on $n$ vertices with no negative cycles, compute the shortest-path distance between every pair of vertices. 
Classical algorithms such as Johnson's~\cite{Joh77} and Floyd and Warshall's~\cite{Flo62,War62} solve APSP in $O(n^3)$ time. After a long line of work~\cite{Fre76,Tak92,Dob90,Han04,Tak04,Tak05,Zwi06,Cha08,Han08,Cha10,HT16}, the fastest known algorithm~\cite{Wil18} only shaves a subpolynomial, $2^{\Omega(\sqrt{\log n})}$ factor off the cubic running time.
 
\item \textbf{Directed Unweighted APSP}: The version of APSP in $n$-vertex directed unweighted graphs. The fastest known algorithm, by Zwick~\cite{Zwi02}, runs in $O(n^{2.5275})$ time with the current bounds for rectangular matrix multiplication~\cite{ADVXXZ25}, and beyond improvements to the matrix multiplication bounds, it has not been improved for more than two decades. If the exponent $\omega$ of matrix multiplication is $2$, the best known running time for this problem~\cite{AGM97,Zwi02} is still only $\tilde{O}(n^{2.5})$.

\item \textbf{Tree Edit Distance}: Given two rooted, ordered, labeled trees with $n$ vertices, and costs for deleting, inserting, and relabeling vertices, find the cheapest sequence of such operations that transforms one tree into the other. Tree Edit Distance is used, for instance, to compare RNA secondary structures and XML documents. The fastest known algorithms for this problem~\cite{DMRW09,NPSVXY25} run in cubic time, up to subpolynomial improvements. Only when all costs are equal to $1$ are truly subcubic algorithms known~\cite{Mao21,Dur23,NPSVXY25}.

\item \textbf{Online Set Disjointness}: Preprocess sets $S_1,\dots,S_N$ and $T_1,\dots,T_N$ over a universe of size $D$, so that given $I$ and $J$ one can quickly decide whether $S_I\cap T_J=\emptyset$, or, in the counting version, return $|S_I\cap T_J|$. Essentially only two straightforward solutions are known: (1) to compute all $N^2$ answers in advance, which fast rectangular matrix multiplication does in $N^{2+o(1)}$ time when $D\le N^{0.321}$~\cite{Cop82,VXXZ24}, and (2) to store the sets and compute each intersection at query time in $O(D)$ time. Set Disjointness data structures have been studied extensively~\cite{CP10,KPP16,GKLP17b,KV20,GLP19}, mostly for sets of small total size over a large universe. 

\item \textbf{CNF-SAT}: Given a Boolean formula $F$ in conjunctive normal form over $n$ variables and $O(n)$ clauses, is there an assignment to the variables that satisfies all clauses? The brute-force algorithm solves CNF-SAT in $2^n\poly(n)$ time. There are slightly faster algorithms when the width of the clauses is bounded by a constant independent of $n$~\cite{Sch99,PPSZ05,Her14}. However, no $O(1.99999^n)$-time algorithm is known.

\item \textbf{Orthogonal Vectors (OV)}: Given $n$ Boolean vectors in $d=\omega(\log n)$ dimensions, are two of them orthogonal? The brute-force algorithm solves OV in $O(n^2d)$ time, and only subpolynomial improvements are known: the fastest algorithms~\cite{AWY15,CW21} run in $n^{2-1/O(\log(d/\log n))}$ time.
\end{itemize}

The apparent lack of progress on problems such as these led to new theories of hardness: NP-completeness~\cite{Coo71,Kar72} for problems with no known polynomial-time algorithm, and fine-grained complexity (FGC)~\cite{Vas18,Bri19} for problems whose best known algorithms were stuck at essentially the straightforward running time. Both theories relate problems via reductions. FGC focuses on improvements in the exponent: a fine-grained reduction from problem $A$ to problem $B$ with respect to running times $a(n)$ and $b(n)$ implies that an $O(b(n)^{1-\eps})$-time algorithm for $B$ for some $\eps>0$ can be converted into an $O(a(n)^{1-\delta})$-time algorithm for $A$ for some $\delta>0$.

Fine-grained complexity took three of the problems from our list above and postulated that the state-of-the-art algorithms cannot be beaten: in the word RAM model of computation with $O(\log n)$-bit words, APSP with polynomially bounded integer weights requires $n^{3-o(1)}$ time (the ``APSP hypothesis''~\cite{RZ04,VW10,VW18,AV14}), $3$SUM on $n$ integers of polynomial size requires $n^{2-o(1)}$ time (the ``$3$SUM hypothesis''~\cite{GO95,Pat10}), and there is no $\varepsilon>0$ such that CNF-SAT on $n$ variables and $O(n)$ clauses can be solved in $(2-\varepsilon)^n$ time (the ``Strong Exponential Time Hypothesis,'' SETH~\cite{IP01,IPZ01,CIP10}).

Through fine-grained reductions, FGC built a web of relations among problems, and within it, classes of problems that are equivalent to each other (see the surveys~\cite{Vas18,Bri19}). For a large variety of problems for which no improved running times had been obtained in decades, FGC pointed to reasons: they are 3SUM-hard, APSP-hard, SETH-hard, or worse. In other words, to improve the known algorithms, one knew which obstacle needed to be overcome.

For instance, many problems are known to be equivalent to APSP under fine-grained reductions. Either all of the following problems are solvable in \emph{truly subcubic} time, meaning in $O(n^{3-\varepsilon})$ time for some constant $\varepsilon>0$ (\emph{truly subquadratic} is defined similarly), or none of them is: the $(\min,+)$-product $A\star B$ of two $n\times n$ matrices $A$ and $B$ (defined by $(A\star B)[i,j]=\min_k(A[i,k]+B[k,j])$)~\cite{FM71,AHU74}\footnote{The equivalence between APSP and $(\min,+)$-product is extremely tight: the running times are the same up to constant factors~\cite{FM71,AHU74}.}, Negative Triangle and Second Shortest Simple Path~\cite{VW10,VW18}, Graph Radius~\cite{AGV15,AGV23}, Tree Edit Distance~\cite{BGMW20,NPSVXY25}, and many other problems~\cite{VW10,VW18,RV11,AGV15,AGV23}, more of which are listed below, after Figure~\ref{fig:web}. Perhaps surprisingly, it is also known~\cite{CVX21} that solving APSP in truly subcubic time would give a polynomially faster algorithm for directed unweighted APSP as well, and the corresponding hypotheses are in fact equivalent~\cite{Fis26} (albeit conditionally\footnote{The equivalence in~\cite{Fis26} requires both $\omega=2$ and the resolution of a conjecture in combinatorics.}), even though the longstanding running times of the two problems are different: $\approx n^3$ for APSP and $\approx n^{2.5}$ for directed unweighted APSP (if $\omega=2$).

Similar reductions and equivalences are also known for $3$SUM: all nontrivial 3-linear degeneracy tests are equivalent to $3$SUM~\cite{DGS20}, and many other problems are 3SUM-hard (e.g., many problems in computational geometry~\cite{GO95,BH01,SEO03,AH08}, and some string problems~\cite{ACLL14,AVWW14,KPP16}). 

CNF-SAT reduces to Orthogonal Vectors~\cite{Wil05}, so that unless SETH is refuted, OV requires $n^{2-o(1)}$ time. SETH also implies fine-grained hardness for many exponential-time problems~\cite{CDL+16}.

Both $3$SUM and APSP reduce to the Exact Triangle problem (given a tripartite graph with integer edge weights, is there a triangle whose three weights sum to zero?)~\cite{VW13,VW10,VW18}. Through Exact Triangle, both reduce to online Set Disjointness (from the list above), and even to its offline version, in which the query pairs are given in advance together with the sets~\cite{VX20a,CX24}; for $3$SUM, direct reductions are also known~\cite{Pat10,KPP16,JX23,ABF23}. In graph terms, offline Set Disjointness is the All-Edges Sparse Triangle problem: given a graph with $m$ edges, decide for every edge whether it lies in a triangle.

All-Edges Sparse Triangle thus inherits the hardness of both $3$SUM and APSP, and it is in turn the source of the conditional lower bounds for Set Disjointness data structures, dynamic graph problems, and other problems~\cite{Pat10,AV14,KPP16,GLP19}. Many came to regard it as one of the most reliably hard problems in fine-grained complexity, and the evidence seemed good.

All-Edges Sparse Triangle can be solved in $O(m^{3/2})$ time just by listing all triangles~\cite{IR78,CN85}, and a high-degree/low-degree split combined with matrix multiplication~\cite{AYZ97} improves this to $O(m^{2\omega/(\omega+1)})$, where $2\le\omega<2.372$ is the matrix multiplication exponent.\footnote{Triangle detection in $n$-vertex graphs has long been known to be reducible to matrix multiplication~\cite{IR78}.} Even if $\omega=2$, this is $m^{4/3}$. On the graphs the reductions produce, which have $n$ vertices of degree about $\sqrt n$, this equals the brute-force running time $n^2$ conjectured to be optimal. Matrix multiplication only seemed useful for dealing with dense inputs, whereas these instances ask for a sparse set of outputs, and it seemed plausible that no technique could do better.

However, a reduction from $A$ to $B$, proved in order to show that $B$ is hard, is also an algorithm for $A$ whenever $B$ turns out to be easy. In this paper we give a new algorithm for a \emph{lopsided} version of All-Edges Sparse Triangle, in which the graph is tripartite and one of the three parts is significantly smaller than the other two. Most known reductions to All-Edges Sparse Triangle can be slightly modified to produce such instances.

As a consequence, all of the problems in the list above (and many more), except for CNF-SAT and Orthogonal Vectors, have polynomially faster algorithms than before:
\begin{itemize}

\item \textbf{$3$SUM}: deterministic $O(n^{1.9992})$ time for integers of polynomial size (Theorem~\ref{thm:3sumapsp}), and for real numbers, a Las Vegas algorithm with $O(n^{1.998})$ expected time whose only operations on real numbers are comparisons, additions, and subtractions (Theorem~\ref{thm:real}).

\item \textbf{APSP}: deterministic $O(n^{2.9995})$ time for polynomially bounded integer weights, and the same bound for the $(\min,+)$-product (Theorem~\ref{thm:3sumapsp}), and for real weights, Las Vegas algorithms with $O(n^{2.998})$ expected time, again using only comparisons, additions, and subtractions on real numbers (Theorem~\ref{thm:real}).

\item \textbf{Directed Unweighted APSP}: a deterministic algorithm running in $O(n^{2+\mu-\varepsilon''})$ time for some constant $\varepsilon''>0$, where $\mu<0.5275$ is defined by $\omega(1,\mu,1)=1+2\mu$, and $\omega(1,\mu,1)$ is the exponent of multiplying an $n\times n^\mu$ by an $n^\mu\times n$ matrix.
Note that $\tilde O(n^{2+\mu})\leq O(n^{2.5275})$ is the longstanding running time of Zwick's algorithm. This is the first improvement over Zwick's algorithm that does not come from faster rectangular matrix multiplication, and it refutes the Directed Unweighted APSP hypothesis~\cite{CVX21,CVX23,Fis26} (Theorem~\ref{thm:dirapsp}).

\item \textbf{Tree Edit Distance}: $O(n^{2.9995})$ time for polynomially bounded integer costs via the tight reduction of~\cite{NPSVXY25}.

\item \textbf{Online Set Disjointness}: for $D\le N^{0.12}$, a deterministic data structure whose preprocessing time is polynomially less than $N^2$ and whose query time is polynomially less than $\sqrt D$, and which returns $|S_I\cap T_J|$, so that it solves the counting version as well. For $D\le N^{1/18}$, the preprocessing takes $O(N^2/D^{0.063})$ time and the queries take $O(D^{0.437})$ time (Theorem~\ref{thm:intro3}). 
\end{itemize}

The APSP and $3$SUM hypotheses are therefore refuted, and so are several other hypotheses of fine-grained complexity. Figure~\ref{fig:web} summarizes the problems affected, and in Section~\ref{Sec:introknown} below we list many more new algorithmic results. All of these results follow from one algorithm, for computing a given sparse set of entries of the product of two thin matrices, which we describe in the next subsection.

We view this as a great success for fine-grained complexity. The hypotheses were made to be tested.
The reductions and connections between problems have always been the focus of fine-grained complexity, and these connections persist whether or not the problems turn out to be hard.
The web of reductions that was interpreted to imply hardness is exactly what turns a single new algorithm into faster algorithms for a myriad of problems at once.
This also opens a number of research directions in fine-grained complexity, algorithm design, and algebraic complexity; we discuss some in Section~\ref{sec:conclusion}.

\subsection{All-Edges Sparse Triangle}
In this paper we present a simple algebraic technique that solves All-Edges Sparse Triangle in $O(n^{2-\delta})$ time, for a constant $\delta>0$, on sparse tripartite graphs with two parts of $n$ vertices and a third part of at most $n^{0.12}$ vertices. In matrix language it computes a prescribed sparse set of entries of a dense product of thin matrices. We call a product of an $N\times D$ matrix by a $D\times N$ matrix \emph{thin} when $D$ is a small power of $N$, and, given a set $W$ of positions of the product, we call these positions, and the entries of the product at them, \emph{wanted}. Unlike in usual fast matrix multiplication algorithms, the wanted entries are a sparse subset of the full matrix product, and our new algorithm critically exploits this. A matrix multiplication algorithm computes all $N^2$ entries of the product whether they are wanted or not, and thus takes $\Omega(N^2)$ time, whereas ours is faster than this. Our theorem is as follows.

\begin{theorem}[see Corollary~\ref{cor:zt} and Theorem~\ref{thm:general}]\label{thm:intro}
Let $N\ge D^{18}$, let $X\in\Z^{N\times D}$ and $Y\in\Z^{D\times N}$ have entries of absolute value $N^{O(1)}$, and let $W$ be any set of $|W|\le N^2/\sqrt D$ positions. The entries $(XY)[I,J]$, $(I,J)\in W$, can be computed deterministically in $O(N^2/D^{0.063})$ operations on $O(\log N)$-bit integers. More generally, for every $\varepsilon<0.1204$ and every $\kappa>0$ there is a $\gamma>0$ such that, whenever $D\le N^{\varepsilon}$ and $|W|\le N^2/D^{\kappa}$, the task takes $O(N^2/D^{\gamma})$ operations.
\end{theorem}

For comparison, when $|W| = N^2/\sqrt D$, computing each wanted entry one by one as an inner product costs $N^2\sqrt D$, and computing the whole product with Coppersmith's rectangular matrix multiplication algorithm~\cite{Cop82} costs $N^{2+o(1)}$. Our algorithm is faster than either of these, and uses polynomially less than one operation per entry of the full product.

The algorithm builds on Coppersmith's algorithm~\cite{Cop82} for rectangular matrix multiplication, which makes use of an algebraic identity of Sch\"onhage~\cite{Sch81}. The high-level idea is to carefully analyze the inner workings of Coppersmith's algorithm to see which operations are no longer needed since they are only used to compute output entries outside of $W$. This approach would not work effectively for most known matrix multiplication algorithms,\footnote{In fact, prior to this work, the authors had tried approaches like this using the Coppersmith--Winograd identities~\cite{CW90} and more, without success.} but it turns out that Sch\"onhage's identity has sparsity and combinatorial properties we can take advantage of to make this work. See Section~\ref{sec:result} for a more detailed overview.

Coppersmith's paper~\cite{Cop82} contains two algorithms. The first algorithm achieves an $n^2\poly\log n$ running time for $n\times n^{0.172}\times n$ matrix multiplication (for matrices with $\poly\log n$ bit entries). We do not use this algorithm here, but it has been extensively used before, both in the previous fastest APSP algorithm by Williams~\cite{Wil18} and to design the fastest fine-grained algorithms for a number of problems using the polynomial method~\cite{Wil14,AWY15,AW15,CW21,ACW16,ACW20}. Our approach instead builds on the second algorithm in~\cite{Cop82}, which only achieves an $n^{2+o(1)}$ running time (and not $n^2\poly\log n$) for $n\times n^\alpha\times n$ matrix multiplication for some constant $\alpha>0$; these $n^{o(1)}$ factors become insignificant in our setting of polynomial savings. Subsequent work on rectangular matrix multiplication has mostly followed this second approach~\cite{Cop97,LeG12,LU18,VXXZ24} achieving the current bound $\alpha > 0.321$. In order to take advantage of the wanted set, we will modify Coppersmith's second algorithm quite differently from prior work, in a way which ultimately achieves the worse bound $\alpha>0.12$.

Computing the wanted entries of a thin matrix product is equivalent to the counting version of Lopsided All-Edges Sparse Triangle (if $X$ and $Y$ are the two biadjacency matrices, then $(XY)[a,b]$ is the number of middle vertices adjacent to both $a$ and $b$), so it solves the detection variant as well. In a tripartite graph with two parts of $n$ vertices and a middle part of $n^{\varepsilon}$ vertices, our algorithm counts the triangles through each of $|W|$ prescribed pairs of outer vertices in $O(|W|\,n^{0.437\varepsilon}+n^{2-0.063\varepsilon})$ time for $\varepsilon<1/18$, which is $O(n^{2-0.063\varepsilon})$ for $|W|\le n^{2-\varepsilon/2}$, and, whenever $|W|\le n^{2-\Omega(1)}$, in truly subquadratic time for every $\varepsilon<0.1204$ (Corollary~\ref{fact:lop} and Theorem~\ref{thm:general}).

\subsection{Known reductions} \label{Sec:introknown}
It is known that $3$SUM~\cite{Pat10,KPP16}, Exact Triangle and hence APSP~\cite{VX20a}, and even their real-valued versions~\cite{CVX22}, all reduce to Lopsided All-Edges Sparse Triangle. In matrix terms, this detection version is a Boolean matrix problem, and it is also known as offline Set Disjointness~\cite{KPP16,VX20a}. 

The published reductions are mostly randomized, since the hypotheses were stated for randomized algorithms, but some have been derandomized: the reductions from $3$SUM~\cite{CH20,FKP24}, including the reduction of~\cite{KPP16} to offline Set Disjointness~\cite[Theorem~1.4]{FKP24}, and the reduction from Exact Triangle to the standard balanced case of All-Edges Sparse Triangle, by Chan and Xu~\cite{CX24}. We show that the reduction from Exact Triangle to the lopsided case can also be made deterministic by hashing the weights modulo a prime chosen with fast matrix multiplication and building the instances as Chan and Xu~\cite{CX24} do (Section~\ref{sec:redproof}), so we refute both hypotheses deterministically.

With randomization, we refute the hypotheses in their most general form, where the inputs are real numbers. The reduction of Chan, Vassilevska Williams, and Xu~\cite{CVX22} from the real-valued problems to sparse triangle problems (which was designed as a hardness argument) uses only additions, subtractions, and comparisons (via Fredman's trick~\cite{Fre76}), so combined with our algorithm it gives Las Vegas algorithms for the real-valued versions of $3$SUM and APSP (Section~\ref{sec:real}). Unlike the reductions for integer inputs, which need only the detection version, the reduction of~\cite{CVX22} that we use is to the counting version of Lopsided All-Edges Sparse Triangle, which asks for the number of triangles through each prescribed pair, and which our Theorem~\ref{thm:main} also solves. Chan, Vassilevska Williams, and Xu~\cite{CVX22} also give a reduction to the Boolean version of All-Edges Sparse Triangle, but it is tight only from real APSP, not from real $3$SUM or Exact Triangle.

\begin{theorem}[Theorems~\ref{thm:zwt}, \ref{thm:3sumapsp}, and~\ref{thm:real}]\label{thm:intro2}
On a word RAM with $O(\log n)$-bit words, deterministic algorithms solve the following problems, where all numbers in the input are integers of absolute value $n^{O(1)}$:
\begin{itemize}
\item Exact Triangle on $n$-vertex graphs in $O(n^{2.9983})$ time,
\item APSP on directed $n$-vertex graphs with no negative cycles in $O(n^{2.9995})$ time,
\item the $(\min,+)$-product of two $n\times n$ matrices in $O(n^{2.9995})$ time, and
\item $3$SUM on $n$ numbers in $O(n^{1.9992})$ time.
\end{itemize}
In the real RAM in which the only operations allowed on real numbers are comparisons, additions, and subtractions, Las Vegas algorithms solve the following problems on real inputs:
\begin{itemize}
\item $3$SUM on $n$ numbers in $O(n^{1.998})$ expected time,
\item APSP in $O(n^{2.998})$ expected time,
\item the $(\min,+)$-product in $O(n^{2.998})$ expected time, and
\item Exact Triangle in $O(n^{2.998})$ expected time.
\end{itemize}
\end{theorem}

Lower bounds are known for the above problems in some restricted models of computation: over the $(\min,+)$ semiring, straight-line programs need $n^3$ operations~\cite{Ker70} and path-comparison algorithms need $\Omega(n^3)$ time~\cite{KKP93}, and 3-linear decision trees for $3$SUM need depth $\Omega(n^2)$~\cite{Eri99,AC05}.

Our algorithms are not in these restricted models, since the known reductions remove the weights (by hashing, or by Fredman's trick for real inputs), and we then count triangles in unweighted graphs using integer matrix multiplication. 

We remark that there already was some indication that some of these models were too restricted, as for instance recent work~\cite{KLM19} showed that $6$-linear decision trees of depth $O(n\log^2 n)$ suffice to solve $3$SUM, and similar strong bounds are known for Exact Triangle. (See the paragraph on unaffected conjectures below for more on this.)

\paragraph{Data structure version, and hinted Online Matrix--Vector.}
Our algorithm for computing the wanted entries of a thin matrix product can also be converted into a data structure. We strengthen Theorem~\ref{thm:intro} as follows (Section~\ref{sec:general}):

\begin{theorem}[see Theorem~\ref{thm:single} and Corollary~\ref{cor:zt}]\label{thm:intro3}
Let $N\ge D^{18}$, and let $X\in\Z^{N\times D}$, $Y\in\Z^{D\times N}$ have entries of absolute value $N^{O(1)}$. The pair $(X,Y)$ can be preprocessed deterministically in $O(N^2/D^{0.063})$ time, polynomially less than the size of the product. After this, any single entry $(XY)[I,J]$ can be computed deterministically in $O(D^{0.437})$ time, polynomially less than $\sqrt D$. More generally, for every $\varepsilon<0.1204$ and every $q>0$ there is a $\gamma>0$ such that, whenever $D\le N^{\varepsilon}$, the pair can be preprocessed in $O(N^2/D^{\gamma})$ time, after which any single entry can be computed in $O(D^{q})$ time.
\end{theorem}

The bounds of Theorem~\ref{thm:intro} follow by asking $|W|$ queries to the data structure.

This allows us to refute the hinted Online Matrix--Vector (OMv) conjectures of van den Brand, Nanongkai, and Saranurak~\cite{BNS19} in the regime of thin hints, i.e., for very rectangular matrices. In the OMv problem, one is given an $n\times n$ Boolean matrix $M$, and then $n$ vectors arrive one at a time, each of which must be multiplied by $M$ before the next one arrives. The OMv conjecture~\cite{HKNS15} asserts that this requires $n^{3-o(1)}$ total time, and it implies tight lower bounds for many dynamic problems. In the hinted versions, the algorithm gets a hint before the vector arrives. For instance, in $v$-hinted Mv, $M$ is an $n\times t$ matrix and the hint is a $t\times n$ matrix $V$, one of whose columns will be the vector. One can then either compute all of $MV$ by fast matrix multiplication as soon as the hint arrives, or wait and compute one matrix--vector product in $O(nt)$ time, and van den Brand et al.~\cite{BNS19} conjecture that no algorithm is polynomially faster than both. Our data structure is faster than both when the hint is thin:

\begin{theorem}[Corollary~\ref{cor:bns}]\label{thm:intro4}
The $v$-hinted Mv, Mv-hinted Mv, and uMv-hinted uMv conjectures of~\cite{BNS19} (Conjectures 5.2, 5.7, and 5.12 there) are refuted in the regime of thin hints. For hint dimension $t=n^{\tau}$ with $0<\tau<1/18$, the phase after the hint takes $O(n^{2-0.063\tau})$ time and the phase after the vector $O(n^{1+0.437\tau})$ time, against the conjectured $n^{2-o(1)}$ and $n^{1+\tau-o(1)}$. Correspondingly, the uMv version, whose two hints have dimensions $t_1=n^{\tau_1}$ and $t_2=n^{\tau_2}$, fails for $\tau_1<\tau_2/18$. All three fail, with smaller savings, for every $\tau<0.1204$ (for the uMv version, $\tau_1<0.1204\,\tau_2$).
\end{theorem}

The main conditional lower bounds of~\cite{BNS19} for dynamic matrix inverse use the conjectures at $\tau\approx0.53$, far from the thin regime, and are unaffected. The bounds that are affected are those for the tail end of the update--query trade-offs with the fastest queries. These
had established the conditional optimality of the dynamic matrix inverse algorithms of Sankowski~\cite{San04} and of~\cite{BNS19} when queries are fast, and these no longer have a believable conditional lower bound. Neither do other results whose lower bounds were based on this end of the trade-offs, e.g., for partially dynamic distance oracles~\cite{HLS24} and dynamic attention~\cite{BSZ24} (Section~\ref{sec:bns}).

\begin{figure}[p]
\centering
\resizebox{\textwidth}{!}{%
\begin{tikzpicture}[
  font=\footnotesize,
  >={Stealth[length=2mm]},
  box/.style={rounded corners=2pt,align=center,inner sep=3pt,text width=#1,execute at begin node={\hyphenpenalty=10000\relax}},
  new/.style={box=#1,draw=blue!60!black,fill=blue!7},
  new/.default=2.7cm,
  old/.style={box=#1,draw=black!45,fill=black!4,text=black!80},
  old/.default=2.7cm,
  cls/.style={box=#1,draw=blue!60!black,very thick,fill=blue!7},
  ours/.style={box=#1,draw=orange!85!black,very thick,fill=orange!12,inner sep=4pt},
  red/.style={->,thick,blue!55!black},
  hard/.style={->,semithick,black!45},
  lab/.style={font=\scriptsize,inner sep=1pt,text=black}
]
\node[cls=3.2cm,anchor=north] (et) at (7.78,0) {\textbf{Exact Triangle}\\ $\equiv$ \#Exact Triangle\\ \cite{CVX23}\\[1pt] \emph{here:} truly subcubic};
\node[cls=7.3cm,anchor=north west] (apsp) at (10.9,0) {\textbf{APSP class}~\cite{VW10,VW18,RV11,AGV15,AGV23,BGMW20,NPSVXY25,TT00,BDT16}\\[2pt]
APSP, $(\min,+)$-product, Negative / Min-Weight Triangle, Metricity, $(\min,+)$-Product Verification, Minimum Weight Cycle, Replacement Paths, Second Shortest Simple Path, Radius, Median, Betweenness Centrality (for unique shortest paths), Tree Edit Distance, Maximum Subarray\\[2pt] \emph{here:} truly subcubic};
\node[new=3.2cm,anchor=north east] (kcl) at (5.45,0) {\textbf{Zero-Weight $k$-Clique}\\ folding~\cite{NP85}\\ \emph{here:} $O(n^{k-\varepsilon})$};

\node[cls=6.3cm,anchor=south west] (sum) at (5.6,1.2) {\textbf{3SUM class}~\cite{GO95,VW10,VW18,FKP24,Pat10,CH20,DGS20,LPV20,CVX23,FJX25}\\[2pt] 3SUM, GeomBase, All-Numbers 3SUM, Convolution-3SUM, \mbox{3-Linear} Degeneracy Testing (e.g.\ 3AP), \#3SUM; MonoConvolution (baseline $n^{1.5}$)\\[2pt] \emph{here:} truly subquadratic, and $O(n^{1.5-\varepsilon})$ for MonoConvolution};
\node[old=3.6cm,anchor=east] (aes) at (5.1,0 |- sum.center) {\textbf{All-Edges Sparse Triangle} (balanced, $m^{4/3}$), triangle listing~\cite{Pat10,KPP16,VX20a,DKPW20,CX24}, and 4-cycle listing~\cite{ABKZ22,JX23}};
\coordinate (Tb) at ($(sum.north)+(0,0.35)$);
\node[old=10.0cm,anchor=south west] (hard3) at (0,0 |- Tb) {\textbf{3SUM-hard problems}\\[2pt]
computational geometry: three collinear points (hence minimum-area triangle), point covering, motion planning~\cite{GO95}, polygon containment~\cite{BH01}\\ strings: jumbled indexing~\cite{ACLL14}, local alignment~\cite{AVWW14}\\ set-intersection and set-disjointness data structures~\cite{Pat10,KPP16} (for universes larger than in Theorem~\ref{thm:intro3})\\ approximate distance oracles~\cite{ABKZ22,ABF23}};
\node[old=7.3cm,anchor=south east] (dyn) at (apsp.east |- Tb) {\textbf{Dynamic graph problems}~\cite{Pat10,AV14,KPP16}: reachability, shortest paths, subgraph connectivity, matching, \dots};

\node[new=3.6cm,anchor=south west] (conv) at ($(sum.south east)+(1.75,-0.65)$) {\textbf{\boldmath$(\min,+)$-Convolution class}~\cite{CMWW19,KPS17}: Superadditivity Testing, MaxConv Upper Bound, Unbounded and 0/1 Knapsack with capacity $t$ (time in terms of $n+t$)\\[1pt] \emph{here:} truly subquadratic (randomized for 0/1 Knapsack)};

\node[new=3.3cm,anchor=north west] (mink) at ($(apsp.south west)+(0,-0.5)$) {\textbf{Min-/Max-Weight $k$-Clique}\\ folding~\cite{NP85}\\[1pt] \emph{here:} $O(n^{k-\varepsilon})$};
\node[new=3.3cm,anchor=north east] (dir) at ($(apsp.south east)+(0,-0.5)$) {\textbf{Directed unweighted APSP}\\ (and weights $\le\nobreak\log^cn$)~\cite{CVX21}\\[1pt] \emph{here:} polynomially faster than Zwick's algorithm};

\node[ours=6.0cm,anchor=north east] (lop) at (10.35,-2.9) {\textbf{Lopsided All-Edges Sparse Triangle}\\ Definitions~\ref{def:lop} and~\ref{def:countlop}\\[2pt] \emph{here:} truly subquadratic\\ \mbox{for $D=n^{\varepsilon}$, $\varepsilon<0.12$}};
\node[new=2.6cm,anchor=north west] (real) at (0,-2.9) {\textbf{Real-valued problems} \cite{CVX22}: real 3SUM, real APSP and $(\min,+)$-product, real Exact Triangle\\[1pt] \emph{here:} truly subquadratic and truly subcubic (Las Vegas)};
\node[new=3.2cm,anchor=north east] (omv) at ($(lop.south east)+(0,-1.0)$) {\textbf{Hinted OMv} \cite{BNS19}, thin regime: $v$-hinted Mv, Mv-hinted Mv, uMv-hinted uMv\\[1pt] \emph{here:} polynomially faster than conjectured};
\node[new=2.2cm,anchor=north west] (xor) at ($(lop.south west)+(0,-1.0)$) {\textbf{3XOR}\\[1pt] \emph{here:} truly subquadratic};

\draw[red] (et |- sum.south) -- node[lab,right]{\cite{CH20,VW13}} (et);
\draw[red] (apsp.west|-et) -- node[lab,below,align=center,text width=1.0cm]{\cite{VW13,VW10,VW18}} (et.east);
\draw[red] (kcl.east|-et) -- (et.west);
\draw[red,very thick] (et) -- node[lab,right,align=left,text width=2.3cm]{\cite{VX20a,CX24,FKP24}, Sec.~\ref{sec:redproof}} (et|-lop.north);
\draw[red] (real.east|-lop.west) -- node[lab,above=1.5pt]{Sec.~\ref{sec:real}} (lop.west);
\draw[red,dashed] (omv.north) -- node[lab,right,align=left]{Thm.~\ref{thm:single},\\ Sec.~\ref{sec:bns}} (omv.north|-lop.south);
\draw[red] (xor.north) -- node[lab,right]{Remark~\ref{rem:3xor}} (xor.north|-lop.south);
\draw[red] (conv.south) -- node[lab,right=1.5pt]{\cite{BCD14}} (conv.south|-apsp.north);
\draw[red] (conv.west) -- node[lab,above=1.5pt,align=center,text width=1.5cm]{\cite{BIS17,CMWW19}} (conv.west-|sum.east);
\draw[red] (mink.north) -- (mink.north|-apsp.south);
\draw[red] (dir.north) -- (dir.north|-apsp.south);

\draw[hard] (sum.west) -- (aes.east|-sum.west);
\draw[hard] ([xshift=-12pt]sum.north) -- ([xshift=-12pt]sum.north|-hard3.south);
\draw[hard] ([xshift=-12pt]sum.north east) -- ([xshift=-12pt]sum.north east |- dyn.south);
\draw[hard] (et.north west) -- ([xshift=-14pt]aes.south east);
\draw[hard] ([xshift=-0.33cm]apsp.north east) -- ([xshift=-0.33cm]dyn.south east);
\draw[hard,rounded corners=3pt] (real.north-|0.55,0) |- (aes.west);

\node[box=3.4cm,draw=black!25,anchor=north west] (unaff) at ($(real.south west)+(0,-0.6)$) {\textbf{Unaffected:} SETH, Orthogonal Vectors, OMv without hints, $k$-SUM and $k$-XOR for $k\ge4$, 3SUM-Indexing};
\begin{scope}[shift={($(mink.west |- dir.south)+(0,-0.3)$)},every node/.style={anchor=west,inner sep=1pt}]
\draw[draw=blue!60!black,fill=blue!7,rounded corners=1pt] (0,-0.12) rectangle (0.7,0.12);
\node at (0.85,0) {new, polynomially faster algorithm};
\draw[draw=black!45,fill=black!4,rounded corners=1pt] (0,-0.57) rectangle (0.7,-0.33);
\node at (0.85,-0.45) {reductions no longer give evidence of hardness};
\draw[red] (0,-0.9) -- (0.7,-0.9);
\node[align=left,anchor=north west] at (0.85,-0.72) {$A\to B$: $A$ reduces to $B$, so our algorithm for $B$\\ gives one for $A$};
\draw[hard] (0,-1.75) -- (0.7,-1.75);
\node at (0.85,-1.75) {$A\to B$: $B$ is hard under $A$};
\draw[red,dashed] (0,-2.2) -- (0.7,-2.2);
\node[align=left,anchor=north west] at (0.85,-2.02) {$A\dashrightarrow B$: $A$ uses the data structure version\\ of our algorithm for $B$ (Theorem~\ref{thm:single})};
\end{scope}
\end{tikzpicture}}
\caption{Problems affected by our algorithm, which solves the problem in the orange box (Lopsided All-Edges Sparse Triangle) directly. The problems inside each of the three thick-bordered blue boxes (3SUM class, Exact Triangle, and APSP class) are pairwise equivalent under fine-grained reductions (for MonoConvolution, whose baseline is $n^{1.5}$, this means $O(n^{1.5-\varepsilon})$ time), and we omit the arrows between them. Problems in gray boxes are reached from 3SUM, APSP, or Exact Triangle only by one-way reductions, so our result gives no faster algorithms for them, but these reductions no longer give evidence of hardness. The new algorithms are deterministic except where marked. Citations are to the reductions; for the problems in the top rows, the cited papers contain many more problems of the same kind.}\label{fig:web}
\end{figure}
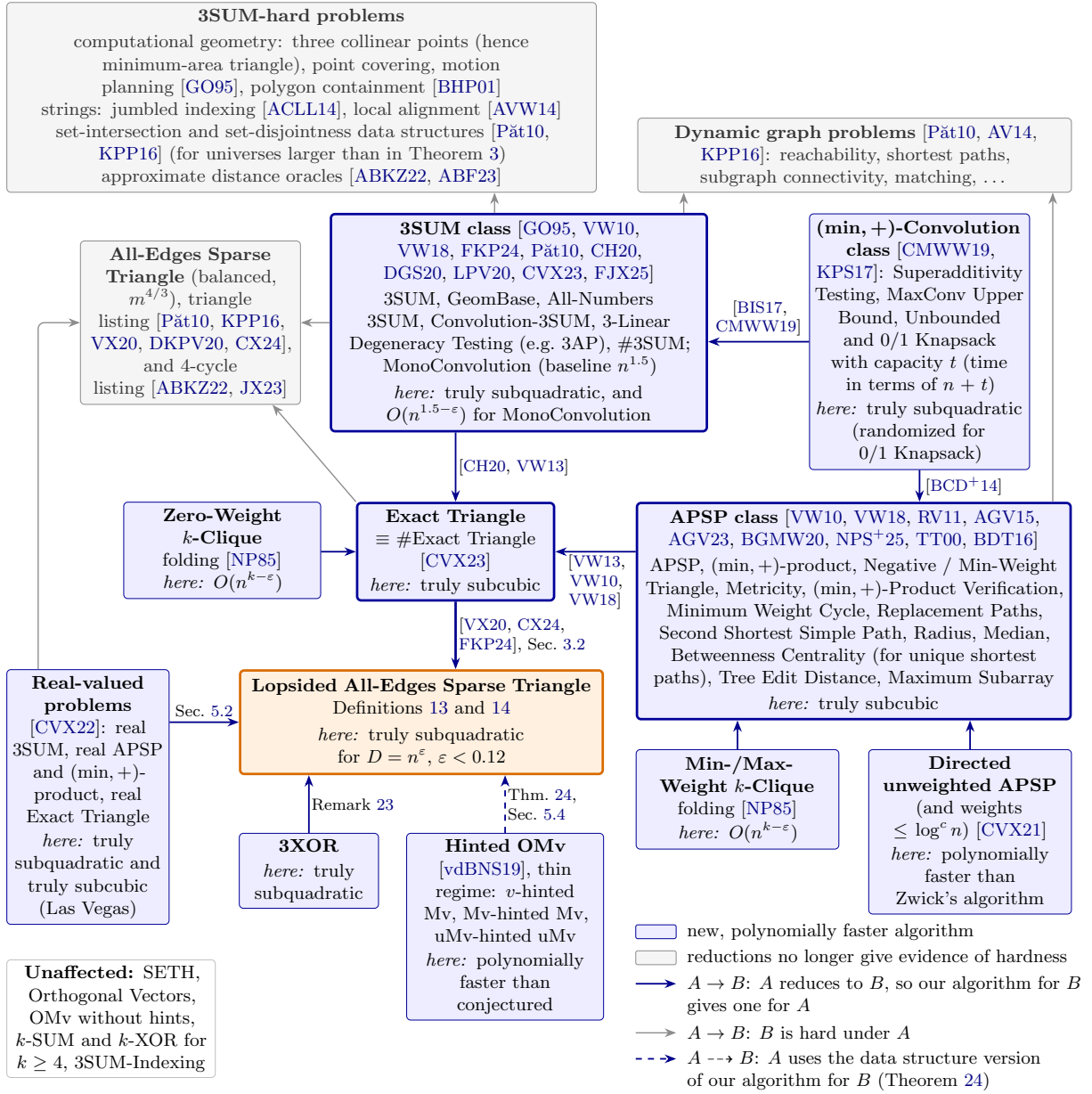

\paragraph{More new algorithms.}
Figure~\ref{fig:web} shows the reductions involved. In addition to the results listed above, composing known reductions with our algorithms for Exact Triangle, $3$SUM, APSP, and the $(\min,+)$-product (Theorem~\ref{thm:intro2}) gives polynomially faster algorithms for many other problems. We next list some examples. Throughout, the numbers in the input are integers of absolute value polynomial in the input size, $n$ is the number of vertices, matrix rows, or input elements, and our algorithms are deterministic unless stated otherwise.
\begin{itemize}
\item \textbf{The APSP class.} Negative Triangle, which asks whether an edge-weighted graph has a triangle of negative total weight, is arguably the simplest problem of the subcubic equivalence class of APSP~\cite{VW10,VW18}; the fact that it is equivalent to APSP made many other reductions possible. The APSP subcubic equivalence class also contains Minimum Weight Cycle (with nonnegative weights)~\cite{RV11,VW10,VW18}, Replacement Paths (for each edge of a shortest $s$--$t$ path, the length of a shortest $s$--$t$ path avoiding it) and Second Shortest Simple Path in directed graphs~\cite{VW10,VW18}, Radius, Median, and (for unique shortest paths) Betweenness Centrality~\cite{AGV15,AGV23}, Metricity (does a given matrix satisfy the triangle inequality?) and $(\min,+)$-Product Verification (is $C=A\star B$ for given matrices $A$, $B$, and $C$?)~\cite{VW10,VW18}, Tree Edit Distance~\cite{BGMW20,NPSVXY25}, Maximum Subarray (the contiguous submatrix of an $n\times n$ matrix with the largest sum of entries)~\cite{TT00,BDT16}, the Wiener Index (returning the sum of all distances in a graph)~\cite{LVW18}, computing the number of witnesses for every entry of the $(\min,+)$-product~\cite{CVX23} and a variety of parity versions of many of the mentioned problems~\cite{paritypaper}. \emph{We give the first truly subcubic algorithms for all of these problems}, and for Diameter and many other problems which reduce to APSP but may not be known to be equivalent.

\item \textbf{The $3$SUM class.} The subquadratic equivalence class of $3$SUM contains the variant with three input sets $A$, $B$, and $C$ (is there $a\in A$, $b\in B$, $c\in C$ with $a+b=c$?) and GeomBase (given $n$ points with integer coordinates on the three horizontal lines $y=0$, $y=1$, and $y=2$, is there a non-horizontal line through three of them?)~\cite{GO95}, All-Numbers $3$SUM (for every input number, decide whether it is part of a solution)~\cite{VW10,VW18} (the reduction becomes deterministic with the self-reductions of~\cite{LVWW16,GP18}, as noted in~\cite{FKP24}), Convolution-3SUM (do $x_0,\dots,x_{n-1}$ satisfy $x_i+x_j=x_{i+j}$ for some $i,j$?)~\cite{Pat10,CH20}, and every nontrivial variant of 3-Linear Degeneracy Testing (are there distinct input numbers $x_1,x_2,x_3$ with $\alpha_1x_1+\alpha_2x_2+\alpha_3x_3=t$, for fixed integers $\alpha_1,\alpha_2,\alpha_3,t$?), such as 3AP, the case $x_1+x_2=2x_3$ (do three of the numbers form an arithmetic progression?)~\cite{DGS20}. The counting version \#3SUM is also in the class~\cite{CVX23,FJX25}; see the last item below. \emph{We give the first truly subquadratic algorithms for all of these problems}.

The $3$SUM class also contains MonoConvolution, a ``colored'' version of Boolean convolution: given three integer sequences $a,b,c$ of length $n$, determine for every index $i$ whether there is a $j<i$ with $a[j]=b[i-j]=c[i]$. Lincoln, Polak, and Vassilevska Williams~\cite[Theorems~7 and~8]{LPV20} show that MonoConvolution is $(n^2, n^{1.5})$ fine-grained equivalent to $3$SUM:  an $O(n^{2-\varepsilon})$-time algorithm for $3$SUM gives an $\tilde O(n^{3/2-\varepsilon/(8-2\varepsilon)})$-time algorithm for MonoConvolution. Our approach gives $O(n^{1.4999})$ time for MonoConvolution, improving on the previous $\tilde O(n^{1.5})$ time of~\cite{LPV20}.

\item \textbf{The $(\min,+)$-convolution class.} The $(\min,+)$-convolution of two sequences $a$ and $b$ of $n$ integers is defined as the sequence $c$ with $c[k]=\min_{i+j=k}(a[i]+b[j])$ for all $k$. This convolution version of APSP is fine-grained reducible to both
 the $(\min,+)$-product~\cite{BCD14} and $3$SUM~\cite{BIS17,CMWW19}. Problems equivalent to $(\min,+)$-convolution include Superadditivity Testing (is $a[i+j]\ge a[i]+a[j]$ for all $i,j$?), Maximum Consecutive Subsums (for every $k$, the largest sum of $k$ consecutive elements), and Tree Sparsity (a subtree with $k$ vertices of maximum weight in a vertex-weighted tree)~\cite{CMWW19}. \emph{We give the first truly subquadratic algorithms for all of these problems}, refuting the $(\min,+)$-convolution hypothesis~\cite{KPS17,CMWW19}. (This would have been possible by refuting either the APSP or the $3$SUM hypothesis, and we refute both.)

\item \textbf{Knapsack and Approximate Subset Sum.} Both 0/1 and Unbounded Knapsack, with $n$ items and capacity $t$, reduce to $(\min,+)$-convolution~\cite{CMWW19}. So does approximating Subset Sum to within a factor $1-\varepsilon$, which is in fact equivalent to it: an algorithm for $(\min,+)$-convolution in time $T(n)$ gives a randomized approximation scheme in $\tilde O(n+T(1/\varepsilon))$ time~\cite{BN21}. For 0/1 Knapsack, the best known bound in terms of $n$ and the largest item weight $w_{\max}\le t$ is $\tilde O(n+w_{\max}^2)$~\cite{Bri24,Jin24}. \emph{We give the first algorithms running in $\tilde O(n+t^{2-\delta})$ time, and the first approximation scheme running in $\tilde O(n+(1/\varepsilon)^{2-\delta})$ time, for a constant $\delta>0$.} Only the algorithm for Unbounded Knapsack is deterministic.

\item \textbf{Zero-, Min-, and Max-Weight $k$-Clique.} For a constant $k\ge3$, find $k$ vertices of an edge-weighted graph whose $\binom k2$ edge weights sum to zero, or to the minimum or maximum possible value. Due to the now classical reduction of Ne\v{s}et\v{r}il and Poljak~\cite{NP85} from $k$-Clique to triangle detection, our new Exact Triangle algorithm immediately implies \emph{the first algorithms running in $O(n^{k-\varepsilon})$ time for an $\varepsilon>0$} (Section~\ref{sec:kclique}), refuting the weighted $k$-Clique hypotheses~\cite{AL13,LVW18,BDT16}.

\item \textbf{3XOR.} Given three lists of $n$ vectors in $\mathbb F_2^{d}$, where $d=O(\log n)$, decide whether there are vectors $a$, $b$, and $c$, one from each list, with $a\oplus b\oplus c=0$, where $\oplus$ denotes the bitwise XOR. This is a well-studied variant of $3$SUM~\cite{JV16,DSW18}. It is not known to reduce to $3$SUM, but the same approach works for it as well, so \emph{we give the first truly subquadratic algorithm} (Remark~\ref{rem:3xor}).

\item \textbf{Counting versions.} Counting the zero-weight triangles of an Exact Triangle instance, or the solutions of a $3$SUM instance, is equivalent to the decision problem~\cite{CVX23}, and for $3$SUM the equivalence also holds under deterministic reductions~\cite{FJX25}. Similarly, counting the witnesses of each entry of the $(\min,+)$-product is equivalent to computing the $(\min,+)$-product~\cite{CVX23}. \emph{We give the first truly subcubic algorithms for \#Exact Triangle and for counting the witnesses of the $(\min,+)$-product, and the first truly subquadratic algorithm for \#3SUM}. See~\cite{CVX23,FJX25} for applications.
\end{itemize}

For problems that are only known to be 3SUM-hard or APSP-hard, such as many of the geometric problems of~\cite{GO95} and the dynamic problems of~\cite{Pat10,AV14,KPP16}, we get no faster algorithms, since the reductions go the other way. For instance, it is now open whether one can find three collinear points among $n$ points in the plane in truly subquadratic time.

\paragraph{Relationships to some unaffected conjectures.}
While we refute various hypotheses, their underlying problems are all similar in nature. We outline how a few other important problems with fine-grained hypotheses differ from $3$SUM, Exact Triangle, and APSP, signaling that refuting them might require other techniques.

We begin with CNF-SAT and Orthogonal Vectors (OV). Neither of these problems is known to reduce to All-Edges Sparse Triangle or similar problems. CNF-SAT and Exact Triangle are known to both reduce to certain triangle problems (Matching Triangles and Triangle Collection~\cite{AVY15}), but these triangle problems concern the triangles in each of about $n^3$ color triples in dense node-colored graphs, and our techniques do not appear to work there.

Unlike CNF-SAT and Orthogonal Vectors, $3$SUM, APSP, and Exact Triangle have long been known to have fast algorithms in certain more powerful models.

$3$SUM (and more generally $k$-SUM) and also Exact Triangle have very efficient {\em linear decision trees}. Following a long line of work~\cite{MadH84,Eri99,AC05,GS17,ES17,CIO16,GP18,KLM19}, it is now known~\cite{KLM19} that for every $k\geq 3$, $k$-SUM has $2k$-linear decision trees of depth $O(n\log^2 n)$, and Exact Triangle in $m$-edge graphs has $6$-linear decision trees of depth $O(m\log^2 m)$. APSP has decision trees of depth $\tilde O(n^{2.5})$~\cite{Fre76}.  No such results are known for CNF-SAT and OV.

$3$SUM, Exact Triangle, and APSP also all have fast
{\em nondeterministic and co-nondeterministic algorithms}.
A nondeterministic algorithm for a decision problem verifies a proof for YES, and a co-nondeterministic algorithm verifies a given proof for NO. Carmosino et al.~\cite{CGI+16}
showed that $3$SUM can be solved nondeterministically and co-nondeterministically in $\tilde O(n^{1.5})$ time, and Exact Triangle and APSP in $\tilde{O}(n^{(3+\omega)/2})$ time. The best known co-nondeterministic algorithms for CNF-SAT and OV are no better than the known deterministic ones, and~\cite{CGI+16} postulate NSETH, the hypothesis that CNF-SAT requires $2^{n-o(n)}$ time even for co-nondeterministic algorithms. They then prove that if NSETH is true, then there cannot be a deterministic (or zero-error randomized) fine-grained reduction from CNF-SAT or OV to $3$SUM, APSP, or Exact Triangle. In fact, our algorithm  shows that if there were a deterministic fine-grained reduction from CNF-SAT or OV to $3$SUM, APSP, or Exact Triangle, then SETH (and not merely NSETH) is false.

Interestingly, the issues posed by co-nondeterminism disappear in the Merlin--Arthur world where randomization is allowed. CNF-SAT~\cite{Wil16}, $3$SUM, and APSP~\cite{ACJRW22} all have vastly improved Merlin--Arthur protocols: the verification time for CNF-SAT (indeed for counting its satisfying assignments) is $2^{n/2}\poly(n)$, for $3$SUM it is $\tilde O(n)$, and for APSP it is $\tilde O(n^2)$.

While we obtain a polynomial improvement for $3$SUM, we are unable to do so for $k$-SUM for $k\geq 4$. $3$SUM reduces to $k$-SUM, but no reduction in the other direction is known, and although $k$-SUM reduces to exact-weight subgraph problems~\cite{AL13,ALW14}, whose $n^k$ baseline we improve (Section~\ref{sec:kclique}), those reductions can only give $k$-SUM algorithms far slower than the $n^{\lceil k/2\rceil}$ baseline. The same holds for $k$-XOR for $k\ge 4$ and for the 3SUM-Indexing conjecture~\cite{GKLP17b,KP19,GGHPV20}. Ultimately, we believe that one good reason for the lack of progress on $k\geq 4$ is that $k$-SUM for $k\geq 4$ is not known to have fine-grained self-reductions. This has been a major obstacle in designing reductions from $k$-SUM and is also a major obstacle in designing algorithms for it. (The same goes for $k$-XOR for $k\geq 4$.)

\paragraph{Combinatorial algorithms.}
The new algorithms are algebraic and potentially impractical in their current form: the constants hidden in the $O(\cdot)$ are enormous, and the exponents can likely be improved. One might ask whether there are practical or ``combinatorial'' algorithms for these problems, similar to how one talks about combinatorial Boolean Matrix Multiplication (BMM). The known algorithms for combinatorial BMM, including the Four Russians~\cite{ADKF70}, the $n^3/\poly\log n$ bounds of~\cite{BW09,Cha15,Yu18}, and the recent $n^3/2^{\Omega((\log n)^{1/7})}$ bound of Abboud, Fischer, Kelley, Lovett, and Meka~\cite{AFK+24}, save only a subpolynomial factor, and no truly subcubic combinatorial algorithm is known. We do not know whether such algorithms exist. Most FGC reductions are combinatorial, so one could refocus the hypotheses to be about combinatorial algorithms~\cite{VW10,VW18,AV14}. 

\paragraph{Other parameter regimes.}
Even without considering combinatorial algorithms, one might also ask whether All-Edges Sparse Triangle can be solved faster in every regime of how large the three parts are. Our technique (due to the identity we use) needs the middle part to have at most $n^{0.12}$ vertices, whereas prior reductions have focused on different regimes.

For instance, 
the balanced sparse case of All-Edges Sparse Triangle introduced by P{\u a}tra{\c s}cu \cite{Pat10} has as input a tripartite graph with roughly equal parts and $m$ edges overall. Most known fine-grained reductions produce such balanced instances and give conditional lower bounds of $m^{4/3-o(1)}$.

Another case of interest is where the tripartite graph has two parts of size $n$ and one part of size $\sqrt n$ and $O(n^{3/2})$ edges overall. Solving this version in $O(n^{2-\eps})$ time, even only for triangle detection (and not even for the all-edge version), would have consequences for girth approximation in undirected unweighted graphs~\cite{RV12}.

\paragraph{Even faster algorithms.}
Our focus in this paper is to give the simplest presentation of our new algorithm for thin matrix products, and show how known reductions to All-Edges Sparse Triangle combine with this to give all the polynomial speedups mentioned above. In particular, over the course of this project, with Claude we generated some algorithms whose exponents are slightly better than the ones we present here, but which are significantly more complicated. We intentionally did not include them in this paper. See also the discussion in Section~\ref{sec:conclusion} and footnote~\ref{footnote:improvements}.

\paragraph{Organization.}
In Section~\ref{sec:result} we give the main matrix algorithm. In Section~\ref{sec:reduction} we derive the algorithms for Exact Triangle, $3$SUM, and APSP from it, by using or slightly modifying known reductions. In Section~\ref{sec:general} we give the data structure version of our matrix algorithm, proving Theorems~\ref{thm:intro} and~\ref{thm:intro3}. In Section~\ref{sec:other} we present reductions to achieve the other algorithmic results discussed above. Finally, in Section~\ref{sec:conclusion} we discuss future directions.

%% file: 02-matrix.tex
\section{Quickly computing certain entries of a thin matrix product}\label{sec:result}

We now present our main new algorithm, for computing a subset of the output entries of a thin matrix product. We call such a matrix product \emph{thin} since the input matrices are highly rectangular. We work in the standard word RAM model with $O(\log N)$-bit words, and thus count operations on $O(\log N)$-bit integers.

\begin{theorem}\label{thm:main}
Let $D\ge4$ be a power of four and $N\ge D^{18}$. Given as input matrices $X\in\Z^{N\times D}$ and $Y\in\Z^{D\times N}$, whose entries are integers of absolute value at most $N^{O(1)}$, as well as a set $W$ of at most $N^2/\sqrt D$ positions of an $N \times N$ matrix, the entries $(XY)[I,J]$, $(I,J)\in W$, can be computed deterministically in time $O(N^2\log^2\!D/D^{1/18})$.
\end{theorem}

We call the positions in $W$, and the entries of $XY$ at them, \emph{wanted}. We emphasize that $X, Y$, and $XY$ may be arbitrary dense matrices, and no structure is assumed about the set $W$ other than its size $|W| \leq N^2/\sqrt D$. 

We have aimed to present this section in a way that is accessible to nonexperts. Some details in Theorem~\ref{thm:main}, such as the exponent $18$, are chosen to simplify the presentation in this section while still proving a version sufficient to refute the $3$SUM and APSP hypotheses. Later, in Section~\ref{sec:general}, we present a stronger data structure version that has improved parameter trade-offs.

Notably, while the language of tensors and bilinear complexity naturally describes much of our algorithm and the prior work it builds on, we deliberately avoid that language here. This is in part so that a reader unfamiliar with that background may still follow, but also because our arguments here will focus on sparsity and combinatorial properties of our algorithms and algebraic identities that are often intentionally abstracted away in tensor language.

Sections~\ref{sec:background}--\ref{sec:full} present an exposition of the relevant prior work that we directly build on, especially Strassen~\cite{Str69},  Sch\"onhage~\cite{Sch81}, and Coppersmith~\cite{Cop82}, and no prior understanding of those works is assumed. The main new idea of this paper is then presented in Section~\ref{sec:sparse}. 

\paragraph{Technique overview.}
Strassen's algorithm multiplies matrices by applying a small identity (seven products for multiplying $2\times2$ matrices) recursively. In that algorithm, the multiplications happen at the leaves of a recursion tree, and the additions that prepare the leaves (the ``encoding phase'') and assemble the result (the ``decoding phase'') are comparatively cheap. We do the same with an identity of Sch\"onhage that computes, simultaneously, two small rectangular matrix products. Applied at $L$ levels of recursion, the identity computes $2^L$ matrix products at once. In order to multiply $X$ and $Y$ using this, we cut them into smaller blocks, and then simultaneously multiply many pairs of blocks using Sch\"onhage's identity (Section~\ref{sec:full}).

The new idea (Section~\ref{sec:sparse}) is to visit only the leaves of the recursion tree that the wanted entries from $W$ need. A priori, it is not clear that this reduces the number of leaves by much, or that the encoding and decoding phases stay cheap. Both nonetheless turn out to hold, for two reasons.

First, we observe that the additions in the encoding phase that prepare the leaves depend on $X$ alone or on $Y$ alone (as is the case in all tensor rank decompositions). Since each $X$ block is multiplied with every $Y$ block, we only need to compute these additions once for each $X$ block and may share them among many calls to our algorithm (and similarly for $Y$). One reason for the requirement $N\ge D^{18}$ is to ensure that this total cost will be negligible (Section~\ref{sec:encoders}).

Second, we prune the recursion tree, and make a recursive call only when one of its outputs is wanted by $W$. This way, the running time is governed by the union of the sets of leaves that the wanted entries need. Most of the analysis goes into bounding this union (Section~\ref{sec:counts}). We will see that every output entry has a private leaf that no other entry needs, but that all its other leaves are shared by many entries. We will characterize each leaf by how often one needs to pick each of the two matrix products from Sch\"onhage's identity to get to it in the recursion tree, and find a trade-off: choosing one of the matrix products leads to a relatively small number of leaves, while choosing the other increases how many output entries share each leaf. Balancing the two, we will find that a set of $N^2/\sqrt D$ wanted entries only needs about $N^2/D^{1/18}$ leaves in all, and paying $O(\log^2\!D)$ operations per leaf gives Theorem~\ref{thm:main}.

\subsection{Strassen's recursive algorithm}\label{sec:background}

We use the recursive approach for multiplying matrices introduced by Strassen~\cite{Str69}. We begin by recalling his algorithm for intuition. The key behind his algorithm is an identity that shows how to multiply two $2\times2$ matrices $A$ and $B$ using only seven multiplications instead of eight. It forms seven products, each of a linear combination of the entries of $A$ with a linear combination of the entries of $B$, and then obtains each entry of $AB$ as a linear combination of those seven products. Recall that, writing $A=(a_{ij})$ and $B=(b_{ij})$, the seven products are
\begin{equation}\label{eq:strassen}
\begin{alignedat}{2}
S_1&=(a_{11}+a_{22})(b_{11}+b_{22}), &\qquad S_5&=(a_{11}+a_{12})\,b_{22},\\
S_2&=(a_{21}+a_{22})\,b_{11}, & S_6&=(a_{21}-a_{11})(b_{11}+b_{12}),\\
S_3&=a_{11}\,(b_{12}-b_{22}), & S_7&=(a_{12}-a_{22})(b_{21}+b_{22}),\\
S_4&=a_{22}\,(b_{21}-b_{11}),
\end{alignedat}
\end{equation}
and the product is
\[
AB=\begin{pmatrix} S_1+S_4-S_5+S_7 & S_3+S_5\\ S_2+S_4 & S_1-S_2+S_3+S_6\end{pmatrix}.
\]
To use this identity to multiply $N \times N$ matrices $A, B$ for larger $N$, we interpret $A$ and $B$ as $2\times2$ block matrices (whose entries are $N/2 \times N/2$ matrices), then multiply them as in the identity, i.e., we form seven pairs of linear combinations of the blocks of $A$ and of $B$, recursively multiply each pair, and then assemble the four blocks of $AB$ from the seven results. When $N=2^L$, the recursion tree (Figure~\ref{fig:strassen}) has $L$ levels, and $7^L$ leaves, one for each choice of one of the seven products at each level. The multiplications take place at the leaves, meaning there are $7^L=N^{\log_27} < N^{2.81}$ multiplications in total. We call the additions on the way down, which prepare the leaves and are computed from $A$ alone and from $B$ alone, the \defn{encoding}, and those on the way up, which assemble the product, the \defn{decoding}.

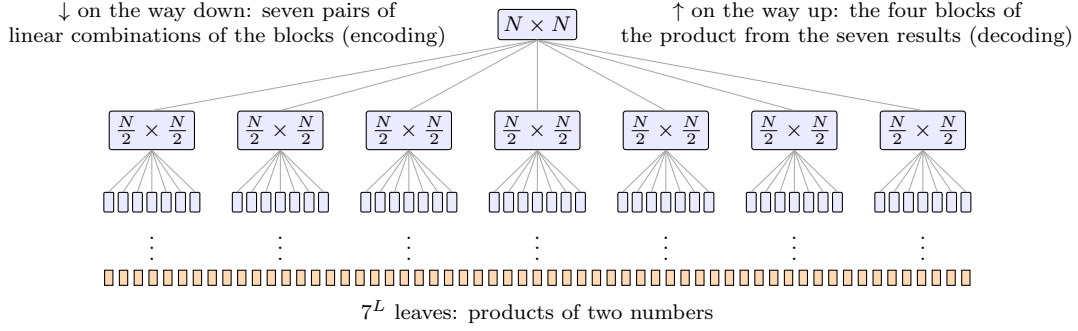
\begin{figure}[tb]
\centering
\begin{tikzpicture}[x=1cm,y=1cm,
  call/.style={draw,rounded corners=1pt,inner sep=2pt,font=\scriptsize,minimum height=4mm,fill=blue!8},
  small/.style={draw,rounded corners=0.5pt,inner sep=0pt,minimum height=2.6mm,minimum width=1.3mm,fill=blue!8},
  leafn/.style={draw,inner sep=0pt,minimum height=2mm,minimum width=1.1mm,fill=orange!30},
  wire/.style={gray!70,line width=0.35pt},
  lab/.style={font=\scriptsize,align=center}]
\node[call] (root) at (6,4.3) {$N\times N$};
\foreach \i in {0,...,6} {
  \node[call] (c\i) at (0.9+1.7*\i,2.9) {$\tfrac N2\times\tfrac N2$};
  \draw[wire] (root.south) -- (c\i.north);
  \foreach \j in {0,...,6} {
    \node[small] (d\i\j) at (0.9+1.7*\i+0.19*\j-0.57,1.95) {};
    \draw[wire] (c\i.south) -- (d\i\j.north);
  }
  \node[font=\scriptsize] at (0.9+1.7*\i,1.45) {$\vdots$};
}
\foreach \k in {0,...,58} { \node[leafn] at (0.33+0.1955*\k,0.95) {}; }
\node[lab] at (6,0.5) {$7^L$ leaves: products of two numbers};
\node[lab] at (1.9,4.3) {$\downarrow$ on the way down: seven pairs of\\ linear combinations of the blocks (encoding)};
\node[lab] at (10.1,4.3) {$\uparrow$ on the way up: the four blocks of\\ the product from the seven results (decoding)};
\end{tikzpicture}
\caption{Strassen's recursion on $N\times N$ matrices, $N=2^L$, showing the first two levels of calls below the root.}
\label{fig:strassen}
\end{figure}

To calculate the total cost of Strassen's algorithm, we also need to count the additions and subtractions at the internal vertices of the recursion tree. At depth $k$ of the recursion tree, there are $7^k$ internal vertices. Both on the way down and on the way up the tree, the algorithm must compute linear combinations of $2^{L-k}\times2^{L-k}$ matrices at each of these vertices, which uses $O(4^{L-k})$ operations. The total across the whole tree is therefore $\sum_{k=0}^{L-1}7^k4^{L-k}=O(7^L)$, a geometric series dominated by its last term, so the additions only increase the running time by a constant factor. This phenomenon is prevalent in recursive matrix multiplication algorithms, although in the new algorithm we develop in Section~\ref{sec:sparse} below, this same argument will no longer apply, and we will need to count the additions and subtractions more carefully.

\subsection{Sch\"onhage's identity for an inner product and an outer product}\label{sec:identity}

While we will still use Strassen's recursive approach here, we will replace his identity with another one due to Sch\"onhage~\cite{Sch81}. In this subsection we give the notation to introduce this identity and explain its usefulness.

It will be convenient to write such identities as trilinear polynomials.\footnote{These are often called tensors, or order-3 tensors, in the literature.} As an example, we first show how to write Strassen's identity from Section~\ref{sec:background} in this way. Introduce four \defn{output variables} $c_{11},c_{12},c_{21},c_{22}$, one for each entry of the product $AB$, and multiply each of the seven products $S_1,\dots,S_7$ of Strassen's identity above by the signed sum of the output variables of the entries where it appeared in the decoding phase. For instance, $S_1$ appears in the formulas for $(AB)_{11}$ and $(AB)_{22}$, so it is multiplied by $c_{11}+c_{22}$, and $S_5$ appears in the formula for $(AB)_{11}$ with a minus sign and in the formula for $(AB)_{12}$, so it is multiplied by $c_{12}-c_{11}$. Strassen's identity then reads
\begin{multline*}
S_1(c_{11}+c_{22})+S_2(c_{21}-c_{22})+S_3(c_{12}+c_{22})+S_4(c_{11}+c_{21})\\
+S_5(c_{12}-c_{11})+S_6\,c_{22}+S_7\,c_{11}\;=\;\sum_{i,j,k=1}^{2}a_{ij}\,b_{jk}\,c_{ik},
\end{multline*}
where $S_1,\dots,S_7$ are the products of two linear forms from~\eqref{eq:strassen}. Now the coefficient of $c_{ik}$ on the right-hand side is the entry $(AB)_{ik}=a_{i1}b_{1k}+a_{i2}b_{2k}$ that we want, and its coefficient on the left-hand side is the combination of the seven products that computes it, such as $S_1+S_4-S_5+S_7$ for $c_{11}$. This polynomial identity thus records the linear combinations and products to be formed, as well as the combinations that assemble the outputs from them.

We will next present Sch\"onhage's identity in the same way. 
Recall that the \defn{outer product} of two vectors $(x_1,x_2,x_3)$ and $(y_1,y_2,y_3)$ consists of the nine numbers $x_iy_j$, and the \defn{inner product} of two vectors $(p_{11},p_{12},p_{21},p_{22})$ and $(q_{11},q_{12},q_{21},q_{22})$ is the single number $\sum_{i,j=1}^{2}p_{ij}q_{ij}$. These are special cases of matrix products: the outer product is multiplying a $3\times1$ matrix times a $1\times3$ matrix, and the inner product is multiplying a $1\times4$ matrix times a $4\times1$ matrix. (We index the $p$ and $q$ variables with $\{1,2\} \times \{1,2\}$ rather than $\{1,2,3,4\}$ to set up nicer notation for $\hat{p}$ and $\hat{q}$ in Lemma~\ref{lem:identity} below.)

Separately computing the inner product and outer product would use $9 + 4 = 13$ multiplications. However, Sch\"onhage~\cite[\S6, eq.~(6.3), with $k=n=3$]{Sch81} showed that they can be computed (along with some error terms that will turn out to be harmless) using only 10 multiplications. The fact that this is possible even when the inner and outer products are on disjoint sets of variables is surprising, and refuted a border rank version of the ``direct sum conjecture for tensor rank''~\cite{Str73}.

We call $x_1,x_2,x_3,\allowbreak p_{11},p_{12},p_{21},p_{22}$ the seven \defn{left variables} and $y_1,y_2,y_3,\allowbreak q_{11},q_{12},q_{21},q_{22}$ the seven \defn{right variables}. We will also use ten output variables $z_{ij}$ (for $1 \leq i,j \leq 3$) and $z_0$, one for each of the ten quantities we aim to compute, namely, the coefficients of the output variables in the polynomial
\[
G:=\sum_{i,j=1}^3 x_iy_j\,z_{ij}+\left(\sum_{i,j=1}^2 p_{ij}q_{ij}\right)z_0 .
\]

As a notational helper to state the identity, we introduce two $3\times3$ matrices $\hat p$ and $\hat q$ of linear forms. These are formed by starting with the $2\times2$ matrices of the $p$'s and of the $q$'s, then extending them to $3 \times 3$ matrices so that every column of $\hat p$ and every row of $\hat q$ sums to zero:
\[
\hat p=\begin{pmatrix} p_{11}&p_{12}&0\\ p_{21}&p_{22}&0\\ -p_{11}-p_{21}&-p_{12}-p_{22}&0\end{pmatrix},\qquad
\hat q=\begin{pmatrix} q_{11}&q_{12}&-q_{11}-q_{12}\\ q_{21}&q_{22}&-q_{21}-q_{22}\\ 0&0&0\end{pmatrix}.
\]
Notice that $\sum_{i,j=1}^{3}\hat p_{ij}\hat q_{ij}=\sum_{i,j=1}^{2}p_{ij}q_{ij}$ is the desired inner product.

Sch\"onhage's identity has ten \defn{terms}, which we name $P_{ij}$ (for $1\le i,j\le3$, corresponding to the entries of $\hat p$ and $\hat q$) and $P_0$. Each term $\lambda$ is the product of a linear form $\varphi_\lambda$ in the left variables, a linear form $\psi_\lambda$ in the right variables, and a linear form $\chi_\lambda$ in the output variables. 
These linear forms are as follows:
\begin{alignat*}{3}
\varphi_{P_{ij}}&=x_i+\hat p_{ij}, &\qquad \psi_{P_{ij}}&=y_j+\hat q_{ij}, &\qquad \chi_{P_{ij}}&=z_{ij}+z_0,\\
\varphi_{P_0}&=-(x_1+x_2+x_3), & \psi_{P_0}&=y_1+y_2+y_3, & \chi_{P_0}&=z_0.
\end{alignat*}

\begin{lemma}\label{lem:identity}
Summing over the ten terms,
\[
\sum_{\lambda}\varphi_\lambda\psi_\lambda\chi_\lambda \;=\; G+E,\qquad\text{where}\qquad
E:=\sum_{i,j=1}^{3}\big(x_i\hat q_{ij}+\hat p_{ij}y_j+\hat p_{ij}\hat q_{ij}\big)\,z_{ij}.
\]
\end{lemma}

\begin{proof}
We compute the coefficient of each output variable on the left-hand side. The coefficient of $z_{ij}$ comes only from the term $P_{ij}$, and it is $\varphi_{P_{ij}} \psi_{P_{ij}} = (x_i+\hat p_{ij})(y_j+\hat q_{ij})=x_iy_j+(x_i\hat q_{ij}+\hat p_{ij}y_j+\hat p_{ij}\hat q_{ij})$. The coefficient of $z_0$ is $\sum_{i,j=1}^{3}(x_i+\hat p_{ij})(y_j+\hat q_{ij})-(x_1+x_2+x_3)(y_1+y_2+y_3)$. In this expression, the products $x_iy_j$ cancel, and the cross terms sum to $\sum_{i=1}^{3}x_i\sum_{j=1}^{3}\hat q_{ij}+\sum_{j=1}^{3}y_j\sum_{i=1}^{3}\hat p_{ij}=0$ since the rows of $\hat q$ and the columns of $\hat p$ sum to zero. What remains is $\sum_{i,j=1}^{3}\hat p_{ij}\hat q_{ij}=\sum_{i,j=1}^{2}p_{ij}q_{ij}$, as desired.
\end{proof}

Lemma~\ref{lem:identity} is the key identity we will use. It computes our desired $G$ plus an error polynomial $E$. Before moving on to our algorithm, we make two structural observations about the identity that we will use below.

Our first observation will be used to handle the error $E$. We call the variables of the inner product, $p$, $q$, and $z_0$, \defn{inner}, and those of the outer product, $x$, $y$, and $z_{ij}$, \defn{outer}. Notice that every monomial of $G$ consists only of inner variables ($p_{ij}q_{ij}z_0$) or only of outer variables ($x_iy_jz_{ij}$). Meanwhile, every monomial of $E$ has an outer output variable together with an inner left or right variable. We will use this observation to show that $E$ is harmless in our algorithm in Lemma~\ref{lem:correct} below.\footnote{Sch\"onhage stated the identity approximately, as a ``border rank'' expression: multiplying every $p$ and $q$ by an indeterminate $\varepsilon$ and every $z_{ij}$ by $\varepsilon^2$ turns it into $\sum_\lambda\varphi_\lambda\psi_\lambda\chi_\lambda=\varepsilon^2G+O(\varepsilon^3)$. One could view this as approximately computing $G$ as $\varepsilon \to 0$, although Bini~\cite{Bin80} showed that when a border rank expression is applied recursively, it may be converted to an exact identity with only small additional cost.  Rather than do this, we instead substitute $\varepsilon = 1$ into the identity, and we will see that $E$ is harmless in our algorithm.}

Our second observation concerns the pattern of output variables in our identity, and will be used to take advantage of the sparsity of the set $W$ of wanted entries in our final algorithm. 
We say that a term $\lambda$ \defn{contributes to} an output variable $z$ if $z$ appears in $\chi_\lambda$. We see that every term contributes to $z_0$, and only $P_{ij}$ contributes to $z_{ij}$. (The decoding layer of Figure~\ref{fig:level} below will also illustrate this.) The fact that only one term contributes to each $z_{ij}$ will eventually help us to reduce how many products we need to compute when we only need certain output entries. Since the $\chi_\lambda$ forms are all sums of output variables (without any minus signs or other coefficients), it follows that for each output variable $z$, its coefficient on the left-hand side of the identity is $\sum_\lambda \varphi_\lambda\psi_\lambda$, where the sum is over the terms $\lambda$ contributing to $z$.

\subsection{Rectangular matrix multiplication using Sch\"onhage's identity}\label{sec:full}

In this subsection, we apply Sch\"onhage's identity recursively to derive an algorithm for rectangular matrix multiplication (computing all the output entries, not just a small subset). Our approach is similar to Coppersmith's algorithm~\cite[``Proof version 2'']{Cop82}, although we will make modifications that worsen some parameters of the algorithm in exchange for preserving some of the combinatorial and structural properties of Sch\"onhage's identity that we will use below. 

\subsubsection{The recursion}\label{sec:recursion}

Sch\"onhage's identity computes two matrix products, the outer product and the inner product, on disjoint sets of variables. We will see that applying it to $L$ levels of recursion therefore involves $2^L$ matrix products, one for each way of choosing the outer or the inner product at each level. A choice with the inner product at $m$ of the levels, and the outer product at the other $L-m$ levels, corresponds to the product of a $3^{L-m}\times4^m$ matrix and a $4^m\times3^{L-m}$ matrix. (We will focus on a setting of $m$ much smaller than $L$ so that these are thin matrix products.) In this subsection, we set up notation for the arrays that the recursion works on and state the recursion, then in Section~\ref{sec:whatit} we will describe precisely what it computes.

The inputs and outputs of our algorithm will be arrays that contain the entries of many different matrices concatenated together. (Later in Section~\ref{sec:tiling} we will use this algorithm as a subroutine to multiply a single pair of matrices.) We will index the entries of these arrays by strings of $L$ variables of Sch\"onhage's identity. Whether each variable is an inner or outer variable will tell us which of the $2^L$ matrix products the array entry corresponds to, and the specific variables will tell us which entry of the matrix it corresponds to.

A \defn{left string} of length $L$ is a string $u=u_1u_2\cdots u_L$ of $L$ left variables, and we call $u_\ell$ its variable at \defn{level}~$\ell$. An \defn{array $a$ on the left strings of length $L$} assigns an integer $a[u]$ to each $u$ that is a left string of length $L$. The array $a$ has $7^L$ entries in total (since there are $7$ left variables). We define right strings, output strings, and arrays on them in the same way. For instance, an array on the output strings of length $L$ has $10^L$ entries (since there are $10$ output variables).

For a term $\lambda$ and a left variable $s$, let $\varphi_\lambda(s)\in\{0,\pm1\}$ denote the coefficient of $s$ in $\varphi_\lambda$, and define $\psi_\lambda(t)$ similarly for right variables $t$. Recall from Sch\"onhage's identity that, for each term $\lambda$, at most three of the seven values $\varphi_\lambda(s)$ are nonzero since each $\varphi_\lambda$ has at most three monomials, and similarly for the seven values $\psi_\lambda(t)$.

Just as Strassen's algorithm cuts matrices into blocks, we cut an array $a$ on strings of length $L\ge1$ into \defn{slices} according to the variable at level~$1$: the slice of $a$ at a left variable $s$ is the array $a_s$ on the strings of length $L-1$ given by $a_s[u']:=a[su']$. Thus, $a$ consists of its seven slices, just as a block matrix consists of its blocks. We similarly define the slices $b_t$ of an array $b$ on right strings, and $c_z$ of an array $c$ on output strings.

\smallskip
\noindent\begin{minipage}[t]{\linewidth}
The recursion $\Full(a,b)$, for arrays $a,b$ on the left and right strings of length $L$, returns an array on the output strings of length $L$:
\begin{enumerate}[label=(\arabic*),nosep]
\item If $L=0$, return the number $ab$.
\item For each term $\lambda$ of Sch\"onhage's identity, form $A_\lambda:=\sum_s\varphi_\lambda(s)\,a_s$ and $B_\lambda:=\sum_t\psi_\lambda(t)\,b_t$.
\item For each term $\lambda$, recursively compute $C_\lambda:=\Full(A_\lambda,B_\lambda)$.
\item Return the array $c$ whose slices are $c_{z_{ij}}:=C_{P_{ij}}$ and $c_{z_0}:=\sum_\lambda C_\lambda$.
\end{enumerate}
\end{minipage}
\medskip

For $L=1$, $\Full$ is one application of Sch\"onhage's identity: step~(2) evaluates the linear forms $\varphi_\lambda$ and $\psi_\lambda$ at the seven numbers $a[s]$ and the seven numbers $b[t]$, step~(3) multiplies the two values for each $\lambda$, and step~(4) sums the products into the ten outputs as the forms $\chi_\lambda$ prescribe. In the terms of Section~\ref{sec:background}, step~(2) is the encoding and step~(4) the decoding.

The calls of $\Full$ form a recursion tree. We refer to the call reached by choosing the terms $\tau_1,\dots,\tau_k$ at levels $1,\dots,k$ as the \defn{vertex} $\tau_1\cdots\tau_k$, a string of $k$ terms. It is at depth $k$, and its input arrays are on strings of length $L-k$. The root is the vertex with $k=0$, and the \defn{leaves}, the vertices $\tau=\tau_1\cdots\tau_L$ at depth $L$, are the $10^L$ multiplications performed in step~(1).

\begin{figure}[tb]
\centering
\begin{tikzpicture}[x=1cm,y=1cm,
  slice/.style={draw,minimum width=1.0cm,minimum height=5mm,inner sep=1pt,font=\scriptsize,fill=blue!6},
  comb/.style={draw,minimum width=1.05cm,minimum height=5mm,inner sep=1pt,font=\scriptsize,fill=blue!14},
  call/.style={draw,minimum width=1.05cm,minimum height=5mm,inner sep=1pt,font=\scriptsize,fill=orange!18},
  outc/.style={draw,minimum width=1.05cm,minimum height=5mm,inner sep=1pt,font=\scriptsize,fill=green!12},
  plus/.style={gray!75,line width=0.4pt},
  minus/.style={red!75!black,line width=0.4pt,dashed},
  lab/.style={font=\small,align=center}]
\foreach \i/\n in {0/{x_1},1/{x_2},2/{x_3},3/{p_{11}},4/{p_{12}},5/{p_{21}},6/{p_{22}}} {
  \node[slice] (s\i) at (5.625+1.05*\i-3.15,4.2) {$a_{\n}$};
}
\draw[decorate,decoration={brace,amplitude=4pt}] (5.625-3.15-0.5,4.55) -- (5.625+3.15+0.5,4.55) node[midway,above=5pt,lab] {the array $a$, cut into its seven slices $a_s$ according to the variable $s$ at level 1};
\foreach \j/\n in {0/{P_{11}},1/{P_{12}},2/{P_{13}},3/{P_{21}},4/{P_{22}},5/{P_{23}},6/{P_{31}},7/{P_{32}},8/{P_{33}},9/{P_0}} {
  \node[comb] (A\j) at (1.25*\j,2.6) {$A_{\n}$};
  \node[call] (C\j) at (1.25*\j,1.2) {$C_{\n}$};
  \draw[plus,-{Stealth[length=3pt]}] (A\j) -- (C\j);
}
\foreach \j/\i in {0/0,0/3, 1/0,1/4, 2/0, 3/1,3/5, 4/1,4/6, 5/1, 6/2, 7/2, 8/2} { \draw[plus] (s\i.south) -- (A\j.north); }
\foreach \j/\i in {6/3,6/5, 7/4,7/6, 9/0,9/1,9/2} { \draw[minus] (s\i.south) -- (A\j.north); }
\foreach \j/\n in {0/{z_{11}},1/{z_{12}},2/{z_{13}},3/{z_{21}},4/{z_{22}},5/{z_{23}},6/{z_{31}},7/{z_{32}},8/{z_{33}},9/{z_0}} {
  \node[outc] (Z\j) at (1.25*\j,-0.4) {$c_{\n}$};
}
\foreach \j in {0,...,8} { \draw[plus,-{Stealth[length=3pt]}] (C\j) -- (Z\j); }
\foreach \j in {0,...,8} { \draw[plus] (C\j.south) -- (Z9.north); }
\draw[plus,-{Stealth[length=3pt]}] (C9) -- (Z9);
\draw[decorate,decoration={brace,amplitude=4pt,mirror}] (-0.55,-0.75) -- (11.25+0.55,-0.75) node[midway,below=5pt,lab] {the array $c$, cut into its ten slices $c_z$ according to the output variable $z$ at level 1};
\node[lab,anchor=east,font=\scriptsize] at (-0.8,3.4) {step (2):\\ encode};
\node[lab,anchor=east,font=\scriptsize] at (-0.8,1.9) {step (3):\\ recurse};
\node[lab,anchor=east,font=\scriptsize] at (-0.8,0.4) {step (4):\\ decode};
\end{tikzpicture}
\caption{One level of the recursion, the analog of one step of Strassen's algorithm. In step~(2), each of the ten terms $\lambda$ combines at most three of the seven slices $a_s$ into $A_\lambda=\sum_s\varphi_\lambda(s)\,a_s$ (solid lines: $+$, dashed: $-$). For instance, $A_{P_{31}}=a_{x_3}-a_{p_{11}}-a_{p_{21}}$ and $A_{P_0}=-a_{x_1}-a_{x_2}-a_{x_3}$. The same is done to $b$ with the forms $\psi_\lambda$. In step~(3), the ten pairs are multiplied by recursive calls, $C_\lambda=\Full(A_\lambda,B_\lambda)$. In step~(4), the slice $c_{z_{ij}}$ of the output is $C_{P_{ij}}$, and the slice $c_{z_0}$ is the sum of all ten, because every term contributes to $z_0$.}
\label{fig:level}
\end{figure}
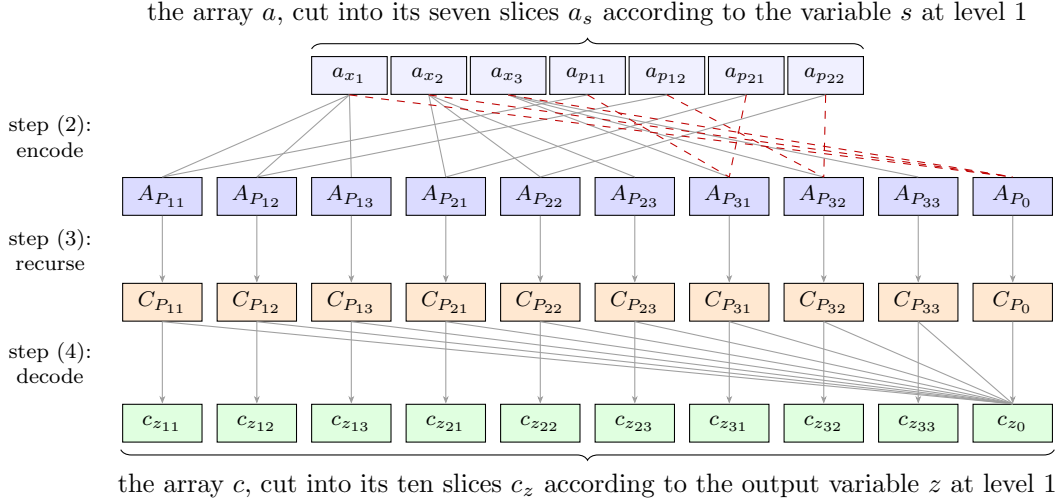

Let us now count the operations performed by this algorithm. At depth $k$ there are $10^k$ vertices (since Sch\"onhage's identity has $10$ terms), each with two input arrays of $7^{L-k}$ entries and an output array of $10^{L-k}$ entries. Step~(2) takes $O(7^{L-k})$ operations at a vertex, and step~(4) takes $O(10^{L-k})$ operations. Hence, $\Full$ uses $\sum_{k=0}^{L}10^k\cdot O(7^{L-k}+10^{L-k})=O(L\cdot10^L)$ operations in total, which is not much more than the $10^L$ multiplications. Of this, the encodings, step~(2), account for only $\sum_{k=0}^{L}10^k\cdot O(7^{L-k})=O(10^L)$ operations, and the decodings, step~(4), for the rest.

\subsubsection{Unraveling the recursion computation}\label{sec:whatit}

In this subsection, we carefully work out which two numbers are multiplied at the leaf $\tau$, and what is returned by $\Full(a,b)$. For a leaf $\tau$ let
\[
\Phi_\tau(a):=\sum_u a[u]\prod_{\ell=1}^{L}\varphi_{\tau_\ell}(u_\ell),\qquad
\Psi_\tau(b):=\sum_v b[v]\prod_{\ell=1}^{L}\psi_{\tau_\ell}(v_\ell),
\]
where the sums are over all left strings $u$ and all right strings $v$ of length $L$. We extend the relation ``contributes to'' of Section~\ref{sec:identity} from terms to vertices: a vertex $\tau_1\cdots\tau_k$ \defn{contributes to} an output string $w$ if $\tau_\ell$ contributes to $w_\ell$ at every level $\ell\le k$. Thus the leaf $\tau$ contributes to $w$ if $\tau_\ell=P_{ij}$ at every level $\ell$ where $w_\ell = z_{ij}$ (but $\tau_\ell$ may be arbitrary at the levels where $w_\ell = z_0$). Let $\Mult(a,b)$ be the array on the output strings of length $L$ given by
\begin{equation}\label{eq:multdef}
\Mult(a,b)[w]\;:=\;\sum_{\tau\text{ contributing to }w}\Phi_\tau(a)\,\Psi_\tau(b).
\end{equation}

\begin{lemma}\label{lem:recursion}
At the leaf $\tau$, $\Full(a,b)$ multiplies $\Phi_\tau(a)$ by $\Psi_\tau(b)$, and it returns $\Mult(a,b)$.
\end{lemma}

\begin{proof}
By induction on $L$ (see Figure~\ref{fig:level}). For $L=0$, the only leaf is the empty string $\tau$, and $\Phi_\tau(a)=a$ and $\Psi_\tau(b)=b$, as the products over the levels are empty, so step~(1) multiplies them and returns $ab=\Mult(a,b)$. For $L\ge1$, step~(2) is the encoding: since $a_s[u']=a[su']$, for every leaf $\lambda\tau'$ we have
\[
\Phi_{\lambda\tau'}(a)=\sum_s\varphi_\lambda(s)\,\Phi_{\tau'}(a_s)=\Phi_{\tau'}(A_\lambda),
\]
and similarly $\Psi_{\lambda\tau'}(b)=\Psi_{\tau'}(B_\lambda)$. Thus, by induction, the leaf $\lambda\tau'$ multiplies $\Phi_{\lambda\tau'}(a)$ by $\Psi_{\lambda\tau'}(b)$, and $C_\lambda[w']=\sum_{\tau'}\Phi_{\lambda\tau'}(a)\Psi_{\lambda\tau'}(b)$, summed over the $\tau'$ contributing to $w'$. Step~(4) is the decoding: the leaves contributing to $zw'$ are the $\lambda\tau'$ with $\lambda$ contributing to $z$ and $\tau'$ contributing to $w'$, so $c[zw']=\sum_{\lambda\text{ contributing to }z}C_\lambda[w']$ agrees with~\eqref{eq:multdef}.
\end{proof}

Since one application of the identity computes the coefficients of $G+E$, we expect that $L$ recursive applications should compute products of $L$ such coefficients. For a left variable $s$, a right variable $t$, and an output variable $z$, let
\begin{equation}\label{eq:c}
\gamma(s,t,z)\;:=\;\sum_{\lambda\text{ contributing to }z}\varphi_\lambda(s)\,\psi_\lambda(t),
\end{equation}
the coefficient of the monomial $stz$ on the left-hand side of the identity. By Lemma~\ref{lem:identity}, it is also the coefficient of $stz$ in $G+E$, and we know that $\gamma(s,t,z)\in\{0,\pm1\}$.

\begin{lemma}\label{lem:mult}
For all arrays $a,b$ and every output string $w$,
\[
\Mult(a,b)[w]\;=\;\sum_{u,v}\gamma(u_1,v_1,w_1)\,\gamma(u_2,v_2,w_2)\cdots\gamma(u_L,v_L,w_L)\;a[u]\,b[v],
\]
the sum over all left strings $u$ and right strings $v$.
\end{lemma}

\begin{proof}
By the definitions of $\Phi_\tau$ and $\Psi_\tau$, the coefficient of $a[u]\,b[v]$ in~\eqref{eq:multdef} is $\sum_\tau\prod_{\ell=1}^{L}\varphi_{\tau_\ell}(u_\ell)\psi_{\tau_\ell}(v_\ell)$, summed over the leaves $\tau$ contributing to $w$. These are the leaves whose term at level $\ell$ contributes to $w_\ell$, independently for each $\ell$, so this sum of products is the product of sums
\[
\prod_{\ell=1}^{L}\ \sum_{\tau_\ell\text{ contributing to }w_\ell}\varphi_{\tau_\ell}(u_\ell)\,\psi_{\tau_\ell}(v_\ell)\;=\;\prod_{\ell=1}^{L}\gamma(u_\ell,v_\ell,w_\ell)
\]
by~\eqref{eq:c}.
\end{proof}

We will use both descriptions of $\Mult$: Lemma~\ref{lem:mult} to read matrix products off it (Section~\ref{sec:defs}), and the sum over leaves in~\eqref{eq:multdef} to skip the leaves that do not contribute to the wanted entries (Section~\ref{sec:sparse}).

\subsubsection{Batch computation of multiple matrix products}\label{sec:defs}

We now explain the claim from the beginning of Section~\ref{sec:recursion}, that our recursive algorithm is simultaneously computing many matrix products. We will also show that the error $E$ is harmless, using the fact that each of its monomials uses an outer output variable, but at least one inner input variable.

The \defn{inner set} of a left string is the set of levels at which it has a $p$ (as opposed to an $x$). It will tell us which of the $2^L$ matrix products the string belongs to. Similarly, the inner set of a right string is the set of levels at which it has a $q$ (as opposed to a $y$), and the inner set of an output string is the set of levels at which it has $z_0$ (as opposed to a $z_{ij}$). The \defn{inner part} of a string is the string of its variables at the levels of its inner set, in the order of the levels, and its \defn{outer part} is the string of its variables at the other levels. A string is determined by its inner set, its outer part, and its inner part.

Fix $m\ge1$, and let $\Nz:=3^{L-m}$ and $D:=4^m$. Of the $2^L$ matrix products, one for each set $Q\subseteq\{1,\dots,L\}$ of levels at which the inner product is chosen, we will use only the $\binom Lm$ products with $|Q|=m$, and set the inputs of the others to zero. As we now make precise, the product for such a $Q$ multiplies an $\Nz\times D$ matrix by a $D\times\Nz$ matrix: the $L-m$ outer levels, with three left and three right outer variables each, give its $\Nz=3^{L-m}$ rows and $\Nz$ columns, and the $m$ inner levels, with four left and four right inner variables each, give the $D=4^m$ terms of the sum in each of its entries. (To prove Theorem~\ref{thm:main}, we will ultimately set $L=19m$, a choice we will explain in Section~\ref{sec:counts} below.) 

We will index the rows and columns of the matrices in these products by strings of variables. For an $\Nz\times D$ matrix, we index its $\Nz$ rows by the $\Nz$ strings of $L-m$ outer left variables, and its $D$ columns by the $D$ strings of $m$ inner left variables. Similarly, for a $D\times\Nz$ matrix, we index its $D$ rows by the strings of $m$ inner right variables, and its $\Nz$ columns by the strings of $L-m$ outer right variables. To multiply an $\Nz\times D$ matrix $A$ by a $D\times\Nz$ matrix $B$, we match the column of $A$ indexed by the string $\pi=p_{i_1j_1}p_{i_2j_2}\cdots p_{i_mj_m}$ with the row of $B$ indexed by the string $\pi'=q_{i_1j_1}q_{i_2j_2}\cdots q_{i_mj_m}$ with the same indices, so that the entry of $AB$ in row $r$ and column $c$ is
\[
(AB)[r,c]\;:=\;\sum_{\pi}A[r,\pi]\,B[\pi',c],
\]
where the sum is over all the $D$ strings $\pi$ of $m$ inner left variables. Thus, the rows of the product are indexed by the strings of outer left variables and its columns by the strings of outer right variables.

Let $Q$ be a subset of $\{1,\dots,L\}$ of size $m$. As above, the outer part of a left string $u$ with inner set $Q$ is a row of an $\Nz\times D$ matrix and its inner part is a column, so the left strings with inner set $Q$ index the entries of such a matrix. For an $\Nz\times D$ matrix $X_Q$ we write $X_Q[u]$ for its entry at the row and the column of $u$. Similarly, the right strings $v$ with inner set $Q$ index the entries of a $D\times\Nz$ matrix $Y_Q$, the row given by the inner part of $v$ and the column by the outer part, and we write $Y_Q[v]$ for this entry. Finally, the output strings with inner set $Q$ index the entries of the product $X_QY_Q$: if the variables of $w$ at the levels outside $Q$ are $z_{i_1j_1},\dots,z_{i_{L-m}j_{L-m}}$, in the order of the levels, then the \defn{row} of $w$ is the string $x_{i_1}\cdots x_{i_{L-m}}$ and its \defn{column} is the string $y_{j_1}\cdots y_{j_{L-m}}$, and we write $(X_QY_Q)[w]$ for the entry there (Figure~\ref{fig:strings}).

We can now check that $\Full$, which returns $\Mult(a,b)$, really computes these products. For each of the $K:=\binom Lm$ subsets $Q$ of $\{1,\dots,L\}$ of size $m$, let $X_Q$ be any $\Nz\times D$ matrix and $Y_Q$ be any $D\times\Nz$ matrix. Let the input arrays $a$ and $b$ be these matrices laid side by side, by setting
\[
a[u]:=X_Q[u],\qquad b[v]:=Y_Q[v]
\]
for all the left strings $u$ and the right strings $v$ whose inner set $Q$ has exactly $m$ elements, and $a[u]:=0$, $b[v]:=0$ at every other string.

\begin{figure}[tb]
\centering
\begin{tikzpicture}[x=1cm,y=1cm,
  cell/.style={draw,minimum width=7.5mm,minimum height=6mm,inner sep=1pt,font=\small},
  inner/.style={cell,fill=blue!12},
  lab/.style={font=\scriptsize,anchor=west,align=left}]
\newcommand{\strrow}[3]{%
  \foreach \txt/\sty [count=\k] in {#2} { \node[\sty] at (\k*0.8,#1) {$\txt$}; }
  \node[lab] at (6*0.8+0.5,#1) {#3};
}
\node[anchor=east,font=\small] at (0.35,0) {left string $u$:};
\strrow{0}{x_1/cell,p_{11}/inner,x_2/cell,x_3/cell,p_{22}/inner,x_1/cell}{outer part $x_1x_2x_3x_1$: the row of $u$ in $X_Q$,\\ inner part $p_{11}p_{22}$: its column,\\ so $u$ indexes an entry of $X_Q$}
\node[anchor=east,font=\small] at (0.35,-1.3) {right string $v$:};
\strrow{-1.3}{y_3/cell,q_{11}/inner,y_1/cell,y_3/cell,q_{22}/inner,y_2/cell}{inner part $q_{11}q_{22}$: the row of $v$ in $Y_Q$,\\ outer part $y_3y_1y_3y_2$: its column,\\ so $v$ indexes an entry of $Y_Q$}
\node[anchor=east,font=\small] at (0.35,-2.6) {output string $w$:};
\strrow{-2.6}{z_{13}/cell,z_0/inner,z_{21}/cell,z_{33}/cell,z_0/inner,z_{12}/cell}{row $x_1x_2x_3x_1$ and column $y_3y_1y_3y_2$, read off\\ from its $z_{ij}$, so $w$ indexes an entry of $X_QY_Q$}
\draw[decorate,decoration={brace,mirror,amplitude=4pt}] (0.8-0.38,-3.0) -- (0.8*6+0.38,-3.0) node[midway,below=5pt,font=\small] {levels $\ell=1,\dots,6$ (shaded: the inner set $Q=\{2,5\}$)};
\begin{scope}[shift={(1.6,-8.0)}]
\fill[blue!30] (0,3.2-0.6-0.08) rectangle (0.8,3.2-0.6+0.08);
\draw (0,0) rectangle (0.8,3.2); \node[font=\small] at (0.4,3.55) {$X_Q$: $81\times16$};
\node[font=\scriptsize,anchor=east] at (-0.1,3.2-0.6) {row $x_1x_2x_3x_1$};
\node[font=\Large] at (1.35,1.6) {$\cdot$};
\fill[blue!30] (1.9+2.41-0.08,2.4) rectangle (1.9+2.41+0.08,3.2);
\draw (1.9,2.4) rectangle (5.1,3.2); \node[font=\small] at (3.5,3.55) {$Y_Q$: $16\times81$};
\node[font=\scriptsize,anchor=north] at (1.9+2.41,2.35) {column $y_3y_1y_3y_2$};
\node[font=\small] at (5.6,1.6) {$=$};
\draw (6.1,0) rectangle (9.3,3.2); \node[font=\small] at (7.7,3.55) {$X_QY_Q$: $81\times81$};
\fill[red!70!black] (6.1+2.41-0.08,3.2-0.6-0.08) rectangle (6.1+2.41+0.08,3.2-0.6+0.08);
\draw[gray!70] (6.1+2.41+0.1,3.2-0.6) -- (9.45,3.2-0.6);
\node[font=\scriptsize,anchor=west] at (9.45,3.2-0.6) {$(X_QY_Q)[w]$};
\end{scope}
\end{tikzpicture}
\caption{From strings to matrix entries, for $L=6$ and $m=2$. A string whose inner set is $Q=\{2,5\}$ has an inner variable ($p$, $q$, or $z_0$) at the two levels of $Q$ and an outer one at the other four. One run of $\Full$ computes all $\binom62=15$ products $X_QY_Q$, one for each subset $Q$ of size $2$.}
\label{fig:strings}
\end{figure}

\begin{lemma}\label{lem:correct}
For the input arrays $a,b$ above and every output string $w$ whose inner set $Q$ has exactly $m$ elements,
\begin{equation}\label{eq:sum}
\Mult(a,b)[w]\;=\;(X_QY_Q)[w] .
\end{equation}
\end{lemma}

\begin{proof}
By Lemma~\ref{lem:mult}, $\Mult(a,b)[w]$ is the sum, over all left strings $u$ and right strings $v$, of $\gamma(u_1,v_1,w_1)\cdots\gamma(u_L,v_L,w_L)\,a[u]b[v]$, where $\gamma(s,t,z)$ is the coefficient of $stz$ in $G+E$. We determine which summands can be nonzero. First, $a[u]b[v]\ne0$ requires the inner sets of $u$ and $v$ to have exactly $m$ elements. Next, at a level $\ell\in Q$ we have $w_\ell=z_0$, and the only monomials of $G+E$ that contain $z_0$ are the $p_{ij}q_{ij}z_0$, with coefficient $1$. So $(u_\ell,v_\ell)=(p_{ij},q_{ij})$ for some $i,j$. Thus $u$ and $v$ have an inner variable at every level of $Q$, and their inner sets are therefore exactly $Q$. Finally, at a level $\ell\notin Q$ we have $w_\ell=z_{ij}$ for some $i,j$, and $u_\ell$ and $v_\ell$ are outer. The only monomial of $G+E$ that contains $z_{ij}$ and no inner variable is $x_iy_jz_{ij}$, with coefficient $1$, so $(u_\ell,v_\ell)=(x_i,y_j)$.

Hence, the summands that remain are those where $u$ has inner set $Q$, the row $r$ of $w$ as its outer part, and some inner part $\pi$, and where $v$ has inner set $Q$, the column $c$ of $w$ as its outer part, and the inner part $\pi'$ with the same indices as $\pi$. There is one such summand for each of the $D$ strings $\pi$, and its coefficient is $1$. Their sum is $\sum_\pi X_Q[r,\pi]\,Y_Q[\pi',c]=(X_QY_Q)[w]$.
\end{proof}

To summarize: one run of $\Full$, which uses $10^L$ multiplications, computes all $K$ of the products $X_QY_Q$. The $K$ products can be completely independent: we are free to plug in any choices for $X_Q$ and $Y_Q$ and the algorithm will give us all $K$ products. In the next subsection we will end up plugging in the same matrix for different copies $X_Q$ (and $Y_Q$).

\subsubsection{\texorpdfstring{Tiling the $N\times D\times N$ product by products of shape $N_0\times D\times N_0$}{Tiling the N x D x N product by products of shape N0 x D x N0}}\label{sec:tiling}

We now use our algorithm, which performs $K$ simultaneous multiplications of an $\Nz \times D$ matrix with a $D \times \Nz$ matrix, in order to do a single multiplication of an $N \times D$ matrix $X$ with a $D \times N$ matrix $Y$. For our parameter setting, $N \gg \Nz$, so we can do this with a simple tiling approach. Tiling enlarges only the outer dimensions, so it handles a smaller inner dimension than Coppersmith's algorithm from the same identity (at best $D\approx N^{0.1204}$ rather than $N^{0.1402}$~\cite{Cop82}). We use it because each entry of $XY$ is then a single output of a single run of $\Full$, rather than a linear combination of many outputs as in Coppersmith's algorithm. We will take advantage of this in our algorithm below (Section~\ref{sec:sparse}).

We cut $X$ into \defn{row blocks} of $\Nz$ consecutive rows and $Y$ into \defn{column blocks} of $\Nz$ consecutive columns. Thus, $XY$ consists of all the products of a row block by a column block, each an $\Nz\times D$ by $D\times\Nz$ product, and one run of $\Full$ computes $K$ such products at once. We use each run on a $K_0\times K_0$ grid of block products, where $K_0:=\lfloor\sqrt K\rfloor$: we group the row blocks into \defn{bands} of $K_0$ consecutive blocks, and similarly the column blocks, and we call a row band together with a column band a \defn{tile} (Figure~\ref{fig:tiling}). Each tile is one run of $\Full$. We fix $K_0^2\le K$ distinct subsets of $\{1,\dots,L\}$ of size $m$, one for each block product of the grid, the same in every tile. In a tile, let $X_Q$ and $Y_Q$ be the row block and the column block of the block product with subset $Q$, and let $X_Q=Y_Q=0$ for the other subsets $Q$. Then one run of $\Full$ computes all $K_0^2$ block products of the tile (Lemma~\ref{lem:correct}), and running it on every tile computes all of $XY$.

We pair the blocks in a grid like this, rather than arbitrarily, so that the left input array of a tile depends only on its row band and the right one only on its column band. This will let us share their encodings among tiles in Section~\ref{sec:encoders} below. (To make the bands fit, we first pad $N$ to a multiple of $K_0\Nz$ with zero rows of $X$ and zero columns of $Y$, which at most doubles $N$ since $N\ge K_0\Nz$, as we will see in Section~\ref{sec:proof}.)

Before moving on, let us compute the number of operations performed by this algorithm. Let $\Cm:=K\Nz^2$ denote the number of output entries among all $K$ products of a tile. 
Since $K_0\ge\sqrt K/2$, there are $(N/K_0\Nz)^2\le4N^2/\Cm$ tiles. Thus, recalling the operation count of a single tile from Section~\ref{sec:recursion}, the whole product takes $O\big(\tfrac{N^2}{\Cm}\,L\cdot10^L\big)$ operations, that is, $O(L\cdot10^L/\Cm)$ per entry of $XY$. One can calculate that this is $N^2\,\mathrm{poly}(L)$ in total when $L=10m$, but about $N^2D^{0.2}$ for our choice $L=19m$. We nonetheless choose $L=19m$, and in Section~\ref{sec:counts} below we explain why this is needed for our sparsity arguments.

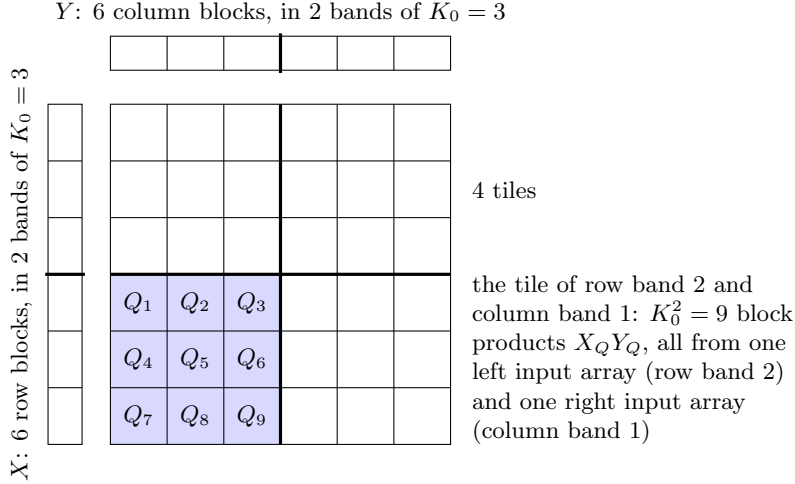
\begin{figure}[tb]
\centering
\begin{tikzpicture}[x=0.75cm,y=0.75cm,font=\footnotesize]
\draw (1.5,4.3) rectangle (7.5,4.9);
\foreach \i in {1,...,5} \draw (1.5+\i,4.3) -- (1.5+\i,4.9);
\draw[very thick] (4.5,4.25) -- (4.5,4.95);
\node[anchor=south] at (4.5,4.95) {$Y$: $6$ column blocks, in $2$ bands of $K_0=3$};
\draw (0.4,-2.3) rectangle (1.0,3.7);
\foreach \i in {1,...,5} \draw (0.4,3.7-\i) -- (1.0,3.7-\i);
\draw[very thick] (0.35,0.7) -- (1.05,0.7);
\node[rotate=90,anchor=south] at (0.3,0.7) {$X$: $6$ row blocks, in $2$ bands of $K_0=3$};
\fill[blue!15] (1.5,-2.3) rectangle (4.5,0.7);
\draw (1.5,-2.3) rectangle (7.5,3.7);
\foreach \i in {1,...,5} { \draw (1.5+\i,-2.3) -- (1.5+\i,3.7); \draw (1.5,3.7-\i) -- (7.5,3.7-\i); }
\draw[very thick] (1.5,0.7) -- (7.5,0.7);
\draw[very thick] (4.5,-2.3) -- (4.5,3.7);
\foreach \i/\j/\t in {0/0/1,1/0/2,2/0/3,0/1/4,1/1/5,2/1/6,0/2/7,1/2/8,2/2/9}
  \node at (2+\i,0.2-\j) {$Q_{\t}$};
\node[anchor=west,align=left] at (7.7,-0.8) {the tile of row band $2$ and\\ column band $1$: $K_0^2=9$ block\\ products $X_QY_Q$, all from one\\ left input array (row band $2$)\\ and one right input array\\ (column band $1$)};
\node[anchor=west,align=left] at (7.7,2.2) {$4$ tiles};
\end{tikzpicture}
\caption{The tiling, drawn for $N=6\Nz$ and $K_0=3$. A block product, a row block of $X$ times a column block of $Y$, holds $\Nz^2$ entries of $XY$. A tile, a band of $K_0$ row blocks together with a band of $K_0$ column blocks, is one run of $\Full$, which computes its $K_0^2\le K$ block products, each with its own subset $Q$ ($Q_1,\dots,Q_9$ shown in the highlighted tile).}\label{fig:tiling}
\end{figure}

\subsection{Extracting only the desired entries}\label{sec:sparse}

In this subsection, we present the main new idea of this paper. Our goal is now to compute only the wanted entries of $XY$, i.e., those at the positions in $W$, which is a set of size $|W| \leq N^2/\sqrt D$. We will modify the algorithm of Section~\ref{sec:full} in two ways. First, we will observe that the encoding (step~(2)) of each input array depends only on that array and not the other, and that each input array is shared by many tiles (defined in Section~\ref{sec:tiling}), so we can compute each encoding only once and share it among those tiles (Section~\ref{sec:encoders}). Second, we will prune the recursion tree for the multiplication and decoding steps (steps~(3)~and~(4)), by skipping a recursive call whenever none of its outputs leads to an output entry from $W$ (Section~\ref{sec:prune}). Finally we will analyze the algorithm by counting the leaves that are reached (Section~\ref{sec:counts}), and adding up the costs (Section~\ref{sec:proof}).

\subsubsection{Sharing the encoding}\label{sec:encoders}

The two numbers $\Phi_\tau(a)$ and $\Psi_\tau(b)$ multiplied at a leaf $\tau$ depend on $a$ alone and on $b$ alone, respectively, so we will compute them once per input array rather than once per tile. The \defn{encoding} of $a$ is the array of the $10^L$ numbers $\Phi_\tau(a)$, indexed by the leaves $\tau$. This is the array that the encoding steps of $\Full$ compute from $a$. To compute it, we run $\Full$ with $b$ left out and without step~(4): by Lemma~\ref{lem:recursion}, the left number at the leaf $\tau$ is $\Phi_\tau(a)$, and the leaf stores it. This takes $O(10^L)$ operations, the cost of the encoding in Section~\ref{sec:recursion}. The encoding of $b$, the array of the numbers $\Psi_\tau(b)$, is computed in the same way.

We compute the encodings of the input arrays of all row bands and all column bands (defined in Section~\ref{sec:tiling}), $2N/K_0\Nz$ arrays of $10^L$ numbers in all, and every tile reads its two encodings from these. It is important that these encodings are shared, since an encoding has $10^L=(9+1)^L$ numbers, more than the $\Cm=\binom Lm9^{L-m}$ entries of the $K$ products of a tile, so encoding every tile separately would cost more than one operation per entry of the entire matrix product (even before focusing on $W$). Since each encoding is computed only once per band, the $2N/K_0\Nz$ encodings are used by all $(N/K_0\Nz)^2$ tiles. Each one costs $O(10^L)$, independently of $W$. The assumption $N\ge D^{18}$ makes this cost negligible (as we will compute in Section~\ref{sec:proof}). %

\subsubsection{Skipping the calls that are not needed}\label{sec:prune}

Once the encodings are computed as above, we can skip that phase of $\Full$. Instead, each leaf reads its two numbers from the encodings, and each call only performs the decoding of step~(4). We now modify $\Full$ so that each call computes only the outputs that are wanted from it, i.e., the outputs that lead to entries in $W$. We will implement this by passing to each call the set $S$ of the outputs that are wanted from it.

More precisely, the outputs of the vertex $\tau_1\cdots\tau_k$ (defined in Section~\ref{sec:recursion}) are indexed by output strings of length $L-k$, and $S$ is a set of such strings. If $U$ is the set passed to the root (in the final algorithm, this will be the set of output strings of the wanted entries of the tile), we will see in Lemma~\ref{lem:prune} that $S$ will consist of the suffixes $w_{k+1}\cdots w_L$ of the strings $w\in U$ to which the vertex contributes (as defined in Section~\ref{sec:whatit}). 
Sets of output strings have slices like arrays (see Section~\ref{sec:recursion}): for a set $S$ of output strings and an output variable $z$, the slice $S_z$ is the set of output strings $w'$ with $zw'\in S$, and an array on $S$ assigns an integer to each string of $S$.

Our modified version of the recursive algorithm $\Full$, which we call $\Pruned$, is as follows.

\medskip
\noindent\begin{minipage}[t]{\linewidth}
$\Pruned_{\tau_1\cdots\tau_k}(S)$, for a set $S$ of output strings of length $L-k$, returns the array returned by the vertex $\tau_1\cdots\tau_k$ of $\Full$, restricted to the strings in $S$:
\begin{enumerate}[label=(\arabic*),nosep]
\item If $k=L$, then look up $\Phi_\tau(a)$ and $\Psi_\tau(b)$ from the precomputed encodings (Section~\ref{sec:encoders}) and return $\Phi_\tau(a)\,\Psi_\tau(b)$ for the leaf $\tau=\tau_1\cdots\tau_L$.
\item For each term $\lambda$ of Sch\"onhage's identity, let $S_\lambda$ be the union of the slices $S_z$ over the $z$ to which $\lambda$ contributes, i.e., $S_{P_{ij}}:=S_{z_{ij}}\cup S_{z_0}$ and $S_{P_0}:=S_{z_0}$.
\item For each term $\lambda$ of Sch\"onhage's identity with $S_\lambda\ne\emptyset$, compute $C_\lambda:=\Pruned_{\tau_1\cdots\tau_k\lambda}(S_\lambda)$, an array on $S_\lambda$.
\item Return the array $c$ on $S$ whose slices are $c_{z_{ij}}:=C_{P_{ij}}$ and $c_{z_0}:=\sum_\lambda C_\lambda$, restricted to $S_{z_{ij}}$ and to $S_{z_0}$. (These restrictions make sense, since $S_{z_{ij}}\subseteq S_{P_{ij}}$ and $S_{z_0}\subseteq S_\lambda$ for every $\lambda$, by step~(2).)
\end{enumerate}
\end{minipage}
\medskip

Let us confirm which leaves $\Pruned$ visits. Recall from~\eqref{eq:multdef} that $\Mult(a,b)[w]$ is the sum of $\Phi_\tau(a)\Psi_\tau(b)$ over the leaves contributing to $w$ (as defined in Section~\ref{sec:whatit}), i.e., those that choose $P_{ij}$ at every level where $w$ has $z_{ij}$, and may choose any term at levels where $w$ has $z_0$. For a set $U$ of output strings of length $L$, we write $\Leaves(U)$ for the set of leaves that contribute to some output string of $U$.

\begin{lemma}\label{lem:prune}
Let $a$ and $b$ be input arrays, and let $U$ be a set of output strings of length $L$. Called at the root, $\Pruned(U)$ returns $\Mult(a,b)$ restricted to $U$. The leaves it visits are exactly those of $\Leaves(U)$, and the sets passed to the calls it makes have total size at most $(L+1)\,|\Leaves(U)|$.
\end{lemma}

\begin{proof}
The values are correct by induction on $L-k$: each $C_\lambda$ is the array of the corresponding child of $\Full$ restricted to $S_\lambda$, so step~(4) is step~(4) of $\Full$ restricted to $S$.

Recall that $S_\lambda$ consists of the last $L-k-1$ variables of each string in $S$ such that $\lambda$ contributes to the first variable of that string. Thus, an induction on $k$ shows that the set passed to the vertex $\tau_1\cdots\tau_k$ consists of the suffixes $w_{k+1}\cdots w_L$ of the strings $w$ of $U$ to which the vertex contributes, and furthermore that the vertex is called if and only if this set is nonempty. This means a leaf is visited if and only if it contributes to some string of $U$.

For the total size, map each suffix in the set passed to a vertex $\tau_1\cdots\tau_k$ to the leaf that extends $\tau_1\cdots\tau_k$ by picking $P_{ij}$ at the levels where the suffix has $z_{ij}$ and by picking $P_0$ where it has $z_0$. By the previous paragraph, the vertex contributes to some string $w$ of $U$ with this suffix, and then so does this leaf, so the leaf is in $\Leaves(U)$. The leaf determines the vertex and the suffix. This means the sets passed to the vertices at each of the $L+1$ depths have total size at most $|\Leaves(U)|$.
\end{proof}

In other words, each call is made at most once, no matter how many strings of $U$ depend on it, and the running time depends on the size of the union of the sets of leaves contributing to the strings of $U$, rather than the sum of their sizes. It remains to bound this union.

\subsubsection{Few leaves contribute to a sparse set of entries}\label{sec:counts}

In our final algorithm, which we present below in Section~\ref{sec:proof}, the set $U$ passed to $\Pruned$ will consist only of the output strings we have been focusing on, i.e., output strings of length $L$ whose inner sets (defined in Section~\ref{sec:defs}) have exactly $m$ elements, so that they index the entries of the products $X_QY_Q$ (Lemma~\ref{lem:correct}). There are $\Cm=\binom Lm9^{L-m}=K\Nz^2$ such output strings (with $\Cm$ as defined in Section~\ref{sec:tiling}), and in this subsection all output strings are of this kind.

Our focus in this subsection is on counting how many leaves the $\Pruned$ algorithm visits. A single output string has $10^m=D^{\log_410}\approx D^{1.66}$ leaves contributing to it. This may appear problematic, since it is more than the $D$ multiplications that computing that output entry directly, as an inner product, would take. Because of this, it is important that we do not simply count the leaves of each output string of $U$ separately. Instead, we will find that most leaves contributing to an output string also contribute to many other output strings, and there are few such leaves in total.

Let us begin by classifying leaves by how many levels they choose $P_0$ at (as opposed to one of the $P_{ij}$s). As we mentioned above, the leaves contributing to an output string with inner set $Q$ are the $10^m$ leaves that choose any term at each of the $m$ levels of $Q$, but $P_{ij}$ at every level outside $Q$ where the output string has $z_{ij}$. We call the one that chooses $P_0$ at every level of $Q$ the \defn{private leaf} of the output string. The private leaf contributes to no other output string, since an output string with a different inner set requires a term other than $P_0$ at some level of $Q$, and an output string with the same inner set but a different outer part (a different $z_{ij}$ at some level outside $Q$) requires a different term there. 

To describe how close a leaf is to being a private leaf, we define the \defn{order} of a leaf as $m$ minus the number of levels at which it chooses $P_0$. A leaf contributing to an output string chooses $P_0$ only at levels of the output string's inner set, so its order is at least $0$. Since $P_0$ is the only term that serves the inner product alone, this is the classification described in the technique overview: the order of a leaf contributing to an output string is the number of levels of the string's inner set at which the leaf takes one of the terms $P_{ij}$ shared with the outer product. The only leaf of order $0$ contributing to an output string is its private leaf. Beyond this, exactly
\[
\alpha_d:=\binom md9^d
\]
leaves of order $d$ contribute to an output string, namely its private leaf with $P_0$ replaced by one of the nine other terms at $d$ of the levels of $Q$ (Figure~\ref{fig:terms}).

The total number of leaves of order $d$ in the entire recursion tree is
\[
\beta_d:=\binom{L}{m-d}9^{L-m+d}
\]
since we can choose the $m-d$ levels with $P_0$ and one of the nine other terms at every other level. In particular, $\beta_0=\binom Lm9^{L-m}=\Cm$ since the leaves of order $0$ are exactly the private leaves, and there is one per output string. A leaf of order $d$ contributes to $\binom{L-m+d}{d}$ output strings, one for each inner set of size $m$ that contains the $m-d$ levels at which it chooses $P_0$. Indeed, $\Cm\alpha_d=\beta_d\binom{L-m+d}{d}$, as both sides count the pairs of an output string and a leaf of order $d$ contributing to it. Finally, if $L=19m$, then for $1\le d\le m$ we have
\begin{equation}\label{eq:decay}
\frac{\beta_d}{\beta_{d-1}}=\frac{9(m-d+1)}{L-m+d}\le\frac{9m}{18m+1}<\frac12,\qquad\text{so}\qquad \beta_d\le2^{-d}\Cm .
\end{equation}

As in the technique overview, leaves of higher order are more numerous per output string, but each of them is shared by more output strings, and by~\eqref{eq:decay} there are fewer of them in all. Below, we count all the leaves of order $d>m/9$, whether or not they contribute to a string of $U$, and this count will be small enough since the $\beta_d$ decrease geometrically. By the ratio in~\eqref{eq:decay}, this requires $m$ to be well below $L/10$: at $m\approx L/10$, the choice that computes the whole product cheaply (Section~\ref{sec:tiling}), the ratio would be close to $1$ for small $d$, so the $\beta_d$ would stay about $\Cm$ instead of decreasing, and the count would give no saving. A small $m$ also makes $D=4^m$ small compared with $\Nz=3^{18m}$ (about $\Nz^{0.07}$), and since a tile must fit into the product, $N$ must be a large power of $D$. Together with the cost of the encodings (Section~\ref{sec:proof}), this is why Theorem~\ref{thm:main} assumes $N\ge D^{18}$.

From now on, we fix $L=19m$, so that $\Nz=3^{18m}$ and $K=\binom{19m}m$. Recall that $D=4^m$, so $2^m=\sqrt D$ and $2^{-m/9}=D^{-1/18}$.

\begin{figure}[t]
\centering
\begin{tikzpicture}[x=1cm,y=1cm,
  cell/.style={draw,minimum width=8mm,minimum height=6mm,inner sep=1pt,font=\small},
  inner/.style={cell,fill=blue!12},
  lab/.style={font=\small,anchor=west,align=left}]
\newcommand{\wordrow}[3]{%
  \foreach \txt/\sty [count=\k] in {#2} { \node[\sty] at (\k*0.85,#1) {$\txt$}; }
  \node[lab] at (6*0.85+0.5,#1) {#3};
}
\node[anchor=east,font=\small] at (0.3,0) {order 0:};
\wordrow{0}{P_{13}/cell,P_0/inner,P_{21}/cell,P_{33}/cell,P_0/inner,P_{12}/cell}{the private leaf of $w$,\\ $\alpha_0=1$ leaf, contributing to $1$ output string}
\node[anchor=east,font=\small] at (0.3,-1.15) {order 1:};
\wordrow{-0.8}{P_{13}/cell,P_{ij}/inner,P_{21}/cell,P_{33}/cell,P_0/inner,P_{12}/cell}{}
\wordrow{-1.5}{P_{13}/cell,P_0/inner,P_{21}/cell,P_{33}/cell,P_{ij}/inner,P_{12}/cell}{}
\node[lab] at (6*0.85+0.5,-1.15) {$\alpha_1=\binom21\cdot9=18$ leaves,\\ each contributing to $\binom51=5$ output strings};
\node[anchor=east,font=\small] at (0.3,-2.3) {order 2:};
\wordrow{-2.3}{P_{13}/cell,P_{ij}/inner,P_{21}/cell,P_{33}/cell,P_{ij}/inner,P_{12}/cell}{$\alpha_2=\binom22\cdot81=81$ leaves,\\ each contributing to $\binom62=15$ output strings}
\draw[decorate,decoration={brace,mirror,amplitude=4pt}] (0.85*1-0.4,-2.7) -- (0.85*6+0.4,-2.7) node[midway,below=5pt,font=\small] {levels $\ell=1,\dots,6$ (shaded: the levels of $Q=\{2,5\}$)};
\end{tikzpicture}
\caption{The leaves contributing to one output string, drawn for $L=6$, $m=2$, and the output string $w=z_{13}z_0z_{21}z_{33}z_0z_{12}$ of Figure~\ref{fig:strings}. A leaf is written as the string of terms it chose, level by level. A leaf of order $d$ is the private leaf with $P_0$ replaced by one of the nine terms $P_{ij}$ (shown as a generic $P_{ij}$ in the shaded cells) at exactly $d$ levels. Leaves of higher order are more numerous per output string, but each is shared by more output strings: a leaf that chose $P_0$ at $m-d$ levels contributes to one output string for each set $Q$ of size $m$ containing those levels, $\binom{L-m+d}{d}$ of them. At $L=19m$ this makes them fewer in all, by~\eqref{eq:decay}.}
\label{fig:terms}
\end{figure}

\begin{lemma}\label{lem:count}
For every set $U$ of output strings whose inner sets have exactly $m$ elements,
\[
|\Leaves(U)|\;\le\;\sum_{d=0}^{m}\min\{|U|\alpha_d,\ \beta_d\},
\]
and if $L=19m$, then $|\Leaves(U)|\le2^{-m/9}\big(2^m|U|+2\Cm\big)$.
\end{lemma}

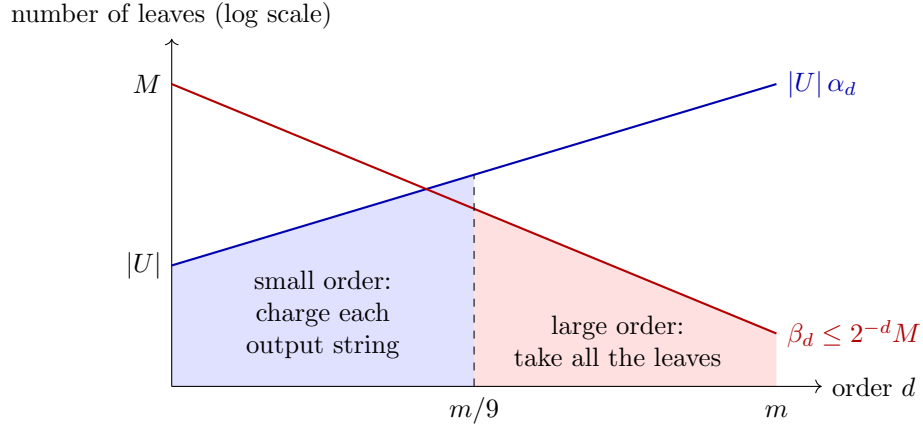
\begin{figure}[t]
\centering
\begin{tikzpicture}[x=1cm,y=1cm]
\coordinate (A0) at (0,1.6);
\coordinate (A1) at (8,4.0);
\coordinate (B0) at (0,4.0);
\coordinate (B1) at (8,0.7);
\coordinate (XA) at (4,2.8);
\coordinate (XB) at (4,2.35);
\fill[blue!12] (A0) -- (XA) -- (4,0) -- (0,0) -- cycle;
\fill[red!12] (XB) -- (B1) -- (8,0) -- (4,0) -- cycle;
\draw[->] (0,0) -- (8.6,0) node[right,font=\small] {order $d$};
\draw[->] (0,0) -- (0,4.6) node[above,font=\small] {number of leaves (log scale)};
\draw[blue!70!black,thick] (A0) -- (A1) node[right,font=\small] {$|U|\,\alpha_d$};
\draw[red!70!black,thick] (B0) -- (B1) node[right,font=\small] {$\beta_d\le2^{-d}\Cm$};
\draw[dashed] (4,0) -- (XA);
\node[below,font=\small] at (4,0) {$m/9$};
\node[below,font=\small] at (8,0) {$m\vphantom{/9}$};
\node[left,font=\small] at (0,4.0) {$\Cm$};
\node[left,font=\small] at (0,1.6) {$|U|$};
\node[font=\small,align=center] at (2.0,0.95) {small order:\\ charge each\\ output string};
\node[font=\small,align=center] at (5.9,0.6) {large order:\\ take all the leaves};
\end{tikzpicture}
\caption{The count of Lemma~\ref{lem:count}, schematically. The leaves of order $d$ in $\Leaves(U)$ number at most $\min\{|U|\alpha_d,\,\beta_d\}$: the bound $|U|\alpha_d$ grows with $d$ up to $d\approx0.9m$ (more leaves of higher order contribute to each output string), and the bound $\beta_d$ shrinks geometrically (there are fewer leaves of higher order). The proof splits at $d=m/9$, which is not necessarily where the two bounds cross, and it bounds the total by the shaded area.}
\label{fig:count}
\end{figure}

\begin{proof}
We bound the number of leaves of order $d$ in $\Leaves(U)$ in two ways: by $|U|\alpha_d$, charging each output string of $U$ for its leaves of order $d$, and by $\beta_d$, the number of all leaves of order $d$. Since every leaf of $\Leaves(U)$ has an order $0\le d\le m$, summing the smaller of the two over $d$ gives the first inequality. (Summed on their own, $|U|\alpha_d$ gives $|U|\cdot10^m$, the cost of computing each output string of $U$ by itself, and $\beta_d$ gives less than $2\Cm$ by~\eqref{eq:decay}, the cost of computing all $\Cm$ output strings.) Neither $|U|\alpha_d$ nor $\beta_d$ depends on which output strings are in $U$.

For the second inequality, let $L=19m$. We use $|U|\alpha_d$ for the orders $d\le m/9$ and $\beta_d$ for the orders $d>m/9$ (Figure~\ref{fig:count}). The orders $d>m/9$ contribute at most $\sum_{d>m/9}\beta_d\le\Cm\sum_{d>m/9}2^{-d}<2\cdot2^{-m/9}\Cm$ by~\eqref{eq:decay}. The orders $d\le m/9$ contribute at most $|U|\sum_{d\le m/9}\binom md9^d$, and since $9^d=72^d8^{-d}\le72^{m/9}8^{-d}$ for $d\le m/9$,
\[
\sum_{d\le m/9}\binom md9^d\le72^{m/9}\sum_{d=0}^{m}\binom md8^{-d}=72^{m/9}\Big(\frac98\Big)^m=\Big(72\cdot\big(\tfrac98\big)^9\Big)^{m/9}<256^{m/9}=2^m\cdot2^{-m/9},
\]
since $72\cdot(9/8)^9<208$. Adding the two parts gives the second inequality.
\end{proof}

Since $2^m=\sqrt D$ and $2^{-m/9}=D^{-1/18}$, the second inequality of Lemma~\ref{lem:count} says that
\[
|\Leaves(U)|\;\le\;D^{-1/18}\big(\sqrt D\,|U|+2\Cm\big).
\]
Consider a tile in which $|U|=\Cm/\sqrt D$ entries are wanted. (This is the density of Theorem~\ref{thm:main}, which allows $N^2/\sqrt D$ wanted entries out of $N^2$.) Then $\sqrt D\,|U|=\Cm$, and the bound says that only $O(\Cm/D^{1/18})$ leaves contribute to the wanted entries of the tile. This is a factor $D^{1/18}$ below the number $\Cm$ of output entries of the tile, and far below the $|U|\cdot D=\Cm\sqrt D$ multiplications that computing each wanted entry directly, as an inner product, would take. Next in Section~\ref{sec:proof} we sum the number of leaves over all the tiles, and we will use the fact that the bound is linear in $|U|$, meaning the total is the same no matter how the wanted entries are spread over the tiles.

\subsubsection{Proof of Theorem~\ref{thm:main}}\label{sec:proof}

We finally put everything together to state and analyze our algorithm. As before, set $D=4^m$ and $L=19m$.

\emph{Two consequences of $N\ge D^{18}$.} We use the assumption $N\ge D^{18}$ twice, each time to check that $N$ is large enough compared with a quantity that grows exponentially in $m$. Since $D^{18}=4^{18m}$, we compare $m$-th powers by comparing their bases. 

First, the tiling of Section~\ref{sec:tiling} requires that $N\ge K_0\Nz$ so that it does not lose too much when padding, since it pads $N$ to a multiple of $K_0\Nz$. This holds: using the general bound $\binom Lm\le(eL/m)^m$, we have $K=\binom{19m}m\le(19e)^m$, so $K\Nz\le(19e\cdot3^{18})^m\le4^{18m}$, as $19e\cdot3^{18}<4^{18}$, and hence $K_0\Nz\le K\Nz\le D^{18}\le N$. This means the padding at most doubles $N$, and we continue to write $N$ for the padded size. 

Second, the encoding phase of the algorithm (Section~\ref{sec:encoders}) has a cost that does not depend on $W$ at all, and we need to ensure it is not too large. Recall that we perform $2N/K_0\Nz$ encodings, which use $O(10^L)$ operations each. To afford them within our budget of $O(2^{-m/9}N^2)$ operations, we need \begin{equation}\label{eq:enchyp}
10^L\le2^{-m/9}\,N\Nz.
\end{equation} This also holds: $10^{19}\cdot2^{1/9}/3^{18}<4^{18}$, so $10^L\cdot2^{m/9}/\Nz\le D^{18}\le N$.

\emph{The algorithm.} We compute the encodings of all row bands and column bands (Section~\ref{sec:encoders}). Each wanted position $(I,J)\in W$ lies in one tile $T$ and is indexed by one output string of $T$, the string whose inner set is the subset $Q$ of the block product containing $(I,J)$ and whose row and column are those of $(I,J)$ within that block product (Sections~\ref{sec:defs} and~\ref{sec:tiling}). For each tile $T$, let $W_T$ be the set of the output strings of the wanted positions in $T$. Then, for each tile $T$ with $W_T\ne\emptyset$, we run $\Pruned(W_T)$, and report, for each $(I,J)\in W$, the value at its output string. This is correct by Lemmas~\ref{lem:correct} and~\ref{lem:prune}.

\emph{The cost.} By~\eqref{eq:enchyp}, the encodings cost $(2N/K_0\Nz)\cdot O(10^L)=O(2^{-m/9}N^2)$ operations. Each position of $W$ lands in one set $W_T$, so $\sum_T|W_T|\le|W|\le N^2/\sqrt D=2^{-m}N^2$, and there are at most $4N^2/\Cm$ tiles, so $\sum_T\Cm\le4N^2$. A call of $\Pruned$ with a set $S$ takes $O(L|S|)$ operations, since each of its steps handles each string of $S$ a constant number of times and a string has length $L$ (we store the sets as sorted lists, so that a slice is a segment and a union of slices is a merge). By Lemma~\ref{lem:prune}, the total number of operations per leaf of $\Leaves(W_T)$ is thus $O(L^2)=O(\log^2\!D)$. Thus, by Lemma~\ref{lem:count} summed over the tiles, the pruned recursions use
\[
O\Big(L^2\sum_T|\Leaves(W_T)|\Big)\;\le\;O\Big(L^2\,2^{-m/9}\Big(2^m|W|+2\sum_T\Cm\Big)\Big)\;=\;O(L^2\,2^{-m/9}N^2)
\]
operations, no matter how $W$ is distributed among the tiles. The remaining minutiae, which list the $K_0^2$ subsets, locate the output strings, and sort the sets $W_T$, take $O(L)$ operations per subset, per tile, and per position of $W$, which is $O(L\,2^{-m/9}N^2)$ in all (as $K\le N$, there are at most $4N^2/\Nz$ tiles, $|W|\le2^{-m}N^2$, and $N,\Nz\ge2^{m/9}$). Finally, since $L=19\log_4D$ and $2^{-m/9}=D^{-1/18}$, the total operation count is $O(N^2\log^2\!D/D^{1/18})$, as desired.

\emph{Word size.} Every number the algorithm handles is an integer of magnitude $N^{O(1)}$. The input entries are of this size, by assumption. Every number in an encoding, or formed on the way to one, is a $\pm1$ combination of at most $7^L$ input entries, as step~(2) only adds and subtracts slices, and every value computed by $\Pruned$ is a sum of at most $10^L$ products of two encoded numbers (Lemma~\ref{lem:recursion}), where $7^L\le10^L\le N^2$ by~\eqref{eq:enchyp}. Indices and list positions are at most the number of words in use, also $N^{O(1)}$. Hence every integer has $O(\log N)$ bits, as desired. \qed

\begin{remark}[comparison with related techniques]
Prior work on matrix multiplication typically counts only multiplications, since they are usually a dominant cost for the algorithm, whereas in our algorithm, the additions need to be dealt with carefully. Our encodings (Section~\ref{sec:encoders}) are Yates' algorithm for Kronecker powers~\cite{Yat37,Goo58}, as is standard in recursive matrix multiplication algorithms. The pruned recursion of Section~\ref{sec:prune} is an analog of FFT pruning~\cite{Mar71,SB93}, of trimmed M\"obius inversion~\cite{BHKK10}, and, transposed, of sparse evaluation of Kronecker-power circuits~\cite{AL25}, which builds on~\cite{Wil24}. The new part along these lines in this paper is the count of Section~\ref{sec:counts}, which shows that the pruned recursion visits few leaves.
\end{remark}

%% file: 03-reduction.tex
\section{Exact Triangle reduces to computing certain entries of a thin matrix product}\label{sec:reduction}

This section derives the faster deterministic algorithms for Exact Triangle (given a tripartite graph with integer weights on its edges, is there a triangle whose three weights sum to zero?), $3$SUM, and APSP from our new matrix theorem. Theorem~\ref{thm:main} already suffices for this, and its stronger form in Section~\ref{sec:general}, Corollary~\ref{cor:zt}, gives the running times stated in the introduction. Later, in Section~\ref{sec:other}, we will similarly derive our other algorithmic results.

Almost nothing in this section is new. We will apply known reductions from fine-grained complexity theory, due to P\u{a}tra\c{s}cu~\cite{Pat10}, Kopelowitz, Pettie, and Porat~\cite{KPP16}, Vassilevska Williams and Williams~\cite{VW13,VW10,VW18}, Vassilevska Williams and Xu~\cite{VX20a}, Chan and He~\cite{CH20}, Chan and Xu~\cite{CX24}, and Fischer, Kaliciak, and Polak~\cite{FKP24}. We go into some detail rather than just citing these known results for two main reasons. First, the theorem statements in the literature do not quite match the parameter regime of the matrix theorem, and obtaining exactly the matrix shapes that our algorithm needs requires slight modifications of the known constructions. Second, for one step, the reduction from Exact Triangle to the lopsided triangle problem, the only known reduction that achieves the parameters we need is randomized~\cite{VX20a}, and the known deterministic one~\cite{CX24} targets the balanced problem. We will make a small modification to achieve a deterministic reduction, using a known tool~\cite{FKP24}. The reductions from $3$SUM and APSP to Exact Triangle are deterministic already, and we cite them as they are. Thus, all the steps are deterministic, and so we ultimately design deterministic algorithms for all these problems. In Figure~\ref{fig:chain} we give an overview of the relevant reductions.

\suppressfloats[t]
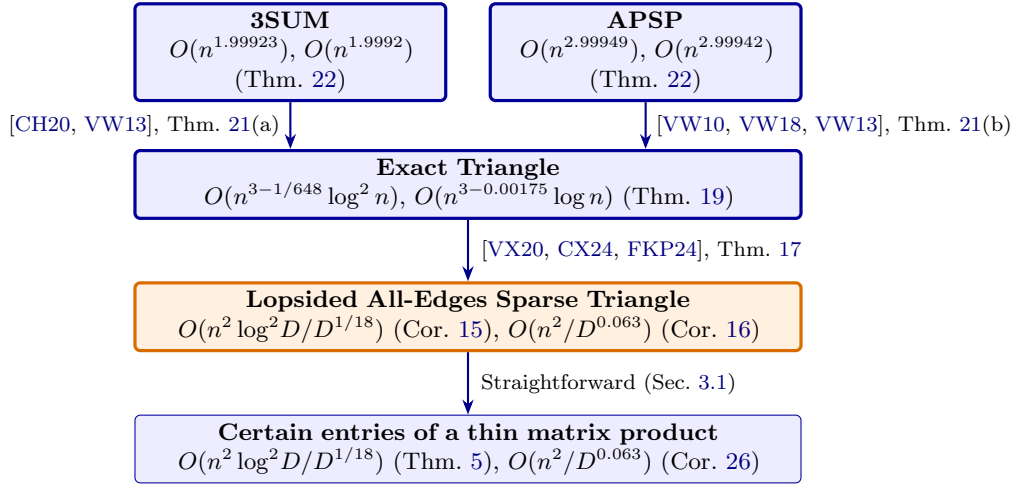
\begin{figure}[t]
\centering
\begin{tikzpicture}[
  font=\footnotesize,
  >={Stealth[length=2mm]},
  box/.style={rounded corners=2pt,align=center,inner sep=3pt,text width=#1},
  cls/.style={box=#1,draw=blue!60!black,very thick,fill=blue!7},
  ours/.style={box=#1,draw=orange!85!black,very thick,fill=orange!12},
  thm/.style={box=#1,draw=blue!60!black,fill=blue!7},
  red/.style={->,thick,blue!55!black},
  lab/.style={font=\scriptsize,fill=white,inner sep=1.5pt,text=black}
]
\node[cls=3.9cm] (sum) at (-2.35,0) {\textbf{3SUM}\\ $O(n^{1.99923})$, $O(n^{1.9992})$\\ (Thm.~\ref{thm:3sumapsp})};
\node[cls=3.9cm] (apsp) at (2.35,0) {\textbf{APSP}\\ $O(n^{2.99949})$, $O(n^{2.99942})$\\ (Thm.~\ref{thm:3sumapsp})};
\node[cls=8.6cm] (et) at (0,-1.75) {\textbf{Exact Triangle}\\ $O(n^{3-1/648}\log^2n)$, $O(n^{3-0.00175}\log n)$ (Thm.~\ref{thm:zwt})};
\node[ours=8.6cm] (lop) at (0,-3.5) {\textbf{Lopsided All-Edges Sparse Triangle}\\ $O(n^2\log^2\!D/D^{1/18})$ (Cor.~\ref{fact:lopmain}), $O(n^2/D^{0.063})$ (Cor.~\ref{fact:lop})};
\node[thm=8.6cm] (mat) at (0,-5.25) {\textbf{Certain entries of a thin matrix product}\\ $O(n^2\log^2\!D/D^{1/18})$ (Thm.~\ref{thm:main}), $O(n^2/D^{0.063})$ (Cor.~\ref{cor:zt})};
\draw[red] (sum.south) -- node[lab,left,pos=0.5,xshift=-2pt]{\cite{CH20,VW13}, Thm.~\ref{thm:knownred}(a)} (sum.south |- et.north);
\draw[red] (apsp.south) -- node[lab,right,pos=0.5,xshift=2pt]{\cite{VW10,VW18,VW13}, Thm.~\ref{thm:knownred}(b)} (apsp.south |- et.north);
\draw[red] (et.south) -- node[lab,right,pos=0.5,xshift=3pt]{\cite{VX20a,CX24,FKP24}, Thm.~\ref{thm:red}} (lop.north);
\draw[red] (lop.south) -- node[lab,right,pos=0.5,xshift=3pt]{Straightforward (Sec.~\ref{sec:tri})} (mat.north);
\path let \p1=(current bounding box.east) in (-\x1,0);
\end{tikzpicture}
\caption{The reductions of this section. Each arrow is a reduction, labeled with the prior work it follows and with the statement here that proves it. Each box gives two running times: the first follows from Theorem~\ref{thm:main}, and the second from its stronger form, Corollary~\ref{cor:zt}. We reprove only the reduction from Exact Triangle to Lopsided All-Edges Sparse Triangle, deterministically and for the parameters of the matrix theorem; the reductions from 3SUM and APSP are cited as they are. The running times in the two bottom boxes assume $n\ge D^{18}$ and $|W|\le n^2/\sqrt D$.}
\label{fig:chain}
\end{figure}

We remind the reader of some basics from fine-grained complexity that we will use here. 
A \defn{reduction} from a problem $A$ to a problem $B$ is an algorithm for $A$ that builds instances of $B$, solves them using an \defn{oracle} (an algorithm for $B$ that it treats as a constant time black box), and does some extra work. We will typically state three things about it: how many instances it builds, of what size, and how much extra time it uses. Composed with an algorithm for $B$, it is an algorithm for $A$, whose running time is the extra time plus the total time to solve $B$ on each instance created. Such a reduction can be interpreted as conferring hardness (if $A$ is hard then so is $B$) or as an algorithmic tool (a faster algorithm for $B$ gives one for $A$). Here we will use the second interpretation to design our algorithms.

The particular types of reductions we focus on are {\em fine-grained reductions}. Here, for two problems $A$ and $B$ and two time bounds $a(n),b(n)$, an $(a(n),b(n))$-fine-grained reduction from $A$ to $B$ is an algorithm $R$ that solves any $n$-sized instance of $A$ by making oracle calls to instances of $B$ of some sizes $n_1,\ldots,n_t$ where for every $\delta>0$ there is an $\eps>0$ s.t. $\sum_{i=1}^t (b(n_i))^{1-\delta}\leq a(n)^{1-\eps}$ and so that $R$ solves $A$ in $O(a(n)^{1-\eps})$ time. In other words, if $B$ had an $O(b(N)^{1-\delta})$ time algorithm, then replacing the oracle calls by calls to the algorithm, due to the inequality above, we would get an $O(a(n)^{1-\eps})$ time algorithm for $A$.

Section~\ref{sec:tri} restates Theorem~\ref{thm:main} and Corollary~\ref{cor:zt} in the language of graphs. Section~\ref{sec:redproof} reproves the reduction from Exact Triangle to the lopsided problem in a deterministic way. It hashes the weights modulo a small prime (in place of the random hash functions of earlier reductions: almost-linear hashing in the reductions from $3$SUM~\cite{Pat10,KPP16}, and a random map into a large prime field in~\cite{VX20a}), chooses the prime deterministically by counting false positives, as in the $3$SUM reduction of Fischer, Kaliciak, and Polak~\cite{FKP24}, and builds the instances and searches for witnesses as in Chan and Xu~\cite{CX24}. 
The false positive counting in~\cite{FKP24} is accomplished using the FFT; here, similar to prior work (e.g.,~\cite{AGM97}), for Exact Triangle we compute the count by multiplying two matrices whose entries are polynomials of low degree which can be done efficiently using fast matrix multiplication.
Section~\ref{sec:zwt} combines this reduction with Theorem~\ref{thm:main} and Corollary~\ref{cor:zt} to solve Exact Triangle in truly subcubic time. Section~\ref{sec:3sumapsp} restates the reductions from $3$SUM and APSP to Exact Triangle~\cite{VW13,VW10,VW18,CH20} with their overheads, and applies them.

\subsection{The Lopsided All-Edges Sparse Triangle problem}\label{sec:tri}

We define a lopsided (rectangular) version of the well-studied problem All-Edges Sparse Triangle.
When viewed in graph terms, Theorem~\ref{thm:main} counts the triangles through prescribed pairs of vertices in a tripartite graph whose middle part is small, and thus in particular can be used to solve All-Edges Sparse Triangle in such graphs.

\begin{definition}[Lopsided All-Edges Sparse Triangle, $\Lop(n,D)$]\label{def:lop}
We are given an unweighted undirected tripartite graph with two parts $A$ and $B$ of $n$ vertices each and a middle part $M$ of at most $D$ vertices, with arbitrary edges in $M\times A$ and $M\times B$. Let $W\subseteq A\times B$ be the set of edges between $A$ and $B$. Decide, for every edge $(a,b)\in W$, whether it lies in a triangle with some vertex of $M$, that is, whether $a$ and $b$ have a common neighbor in $M$.
\end{definition}

In this notation, the usual All-Edges Sparse Triangle problem is the balanced case $D=n$ (after the standard reduction to tripartite graphs), except that one asks about every edge of the graph, not only those between $A$ and $B$, and that the running time is measured in terms of the number $m$ of edges. We think of $W$ as being polynomially smaller than $n^2$ and $D$ as polynomially smaller than $n$, so the instance is sparse in this sense. The brute-force algorithm solves $\Lop(n,D)$ in $O(|W|D)$ time. If $D\leq n^{0.321}$, then, since $0.321\le\alpha$~\cite{VXXZ24},\footnote{Here $\alpha$ denotes the supremum of the constants $\alpha\in(0,1]$ such that $n\times n^\alpha$ by $n^\alpha\times n$ matrices can be multiplied using $n^{2+o(1)}$ operations.} fast rectangular matrix multiplication can solve the problem in $n^{2+o(1)}$ time even when $|W|=n^2$. We will beat both $|W|D$ and $n^2$ for $D=n^{\varepsilon}$ and $|W|=n^2/\sqrt D$ with a constant $\varepsilon>0$, and this will suffice to refute both the $3$SUM and the APSP hypotheses.

Some of the known reductions, such as those from the real-valued problems in Section~\ref{sec:real}, reduce to the counting version of All-Edges Sparse Triangle, which asks for the number of triangles through each edge; we define the lopsided version.

\begin{definition}[Lopsided All-Edges Sparse Triangle Counting, $\#\Lop(n,D)$]\label{def:countlop}
This is the same problem as $\Lop(n,D)$, except that for every edge $(a,b)\in W$ we ask for the number of triangles it lies in, that is, the number of common neighbors of $a$ and $b$ in $M$.
\end{definition}

In matrix language, let $X\in\{0,1\}^{n\times D}$ and $Y\in\{0,1\}^{D\times n}$ be the biadjacency matrices of the edges between $A$ and $M$ and between $M$ and $B$, that is, $X[a,v]=1$ if and only if $a\in A$ and $v\in M$ are adjacent, and $Y[v,b]=1$ if and only if $v\in M$ and $b\in B$ are adjacent (padded with zeros if $M$ has fewer than $D$ vertices). Then $\Lop(n,D)$ asks for the entries of the Boolean product of $X$ and $Y$ at the positions in $W$, and $\#\Lop(n,D)$ asks for the corresponding entries of the product of $X$ and $Y$ over the integers. In the language of Section~\ref{sec:result}, $W$ is the set of wanted positions, and we call its elements the \defn{query pairs}.

The counting version $\#\Lop(n,D)$ is also equivalent, up to polylogarithmic factors, to the problem of Theorem~\ref{thm:main} and Corollary~\ref{cor:zt}, that is, to computing the wanted entries $(XY)[I,J]$, $(I,J)\in W$, of a thin matrix product, for matrices $X\in\Z^{n\times D}$ and $Y\in\Z^{D\times n}$ whose entries have absolute value $n^{O(1)}$.\footnote{An instance of $\#\Lop(n,D)$ asks for the wanted entries of a thin matrix product of two $0/1$ matrices. To reduce to $\#\Lop$, we use a standard idea: split each matrix into its positive and negative parts and write their entries in binary, $X=\sum_i2^iX_i$ and $Y=\sum_j2^jY_j$ with $0/1$ matrices $X_i$ and $Y_j$. Then $XY=\sum_{i,j}2^{i+j}X_iY_j$, so a product with $b$-bit entries costs $O(b^2)$ instances of $\#\Lop(n,D)$ with the same query pairs $W$.} 

Besides asking for the wanted entries of a Boolean thin matrix product, the detection version $\Lop(n,D)$ can also be viewed as offline Set Disjointness (see, e.g.,~\cite{KPP16,VX20a}): the sets are the neighborhoods in $M$ of the vertices of $A$ and $B$, over a universe of at most $D$ elements, and a query pair $(a,b)\in W$ asks whether the sets of $a$ and $b$ are disjoint. The reductions below (mostly from prior work) produce at most $n^2/\sqrt D$ query pairs and sets of at most $\sqrt D$ elements, which is the regime of~\cite{KPP16} and of~\cite[Corollary~3.12]{VX20a}.

With the above discussion in mind, Theorem~\ref{thm:main} gives the following.

\begin{corollary}\label{fact:lopmain}
Let $D\ge4$ be a power of four with $n\ge D^{18}$, and consider an instance of $\#\Lop(n,D)$ or of $\Lop(n,D)$ with $|W|$ query pairs. If $|W|\le n^2/\sqrt D$, then the instance can be solved deterministically in $O(n^2\log^2\!D/D^{1/18})$ time. In general, splitting $W$ into sets of at most $n^2/\sqrt D$ query pairs solves it deterministically in time
\[
O\big((n^2+|W|\sqrt D)\log^2\!D/D^{1/18}\big).
\]
\end{corollary}
\begin{proof}
Apply Theorem~\ref{thm:main} with $N=n$ to the two biadjacency matrices: the entries $(XY)[a,b]$, $(a,b)\in W$, are the numbers of common neighbors, and they are nonzero exactly for the query pairs that lie in a triangle. For larger $W$, apply the theorem to each of at most $\lceil|W|\sqrt D/n^2\rceil+1$ pieces.
\end{proof}

The data structure of Section~\ref{sec:general} improves the saving $D^{1/18}$ to $D^{0.063}$ (Corollary~\ref{cor:zt}), and gives the following.

\begin{corollary}\label{fact:lop}
Let $n\ge D^{18}$, and consider an instance of $\#\Lop(n,D)$ or of $\Lop(n,D)$ with $|W|$ query pairs. It can be solved deterministically in $O(|W|\,D^{0.437}+n^2/D^{0.063})$ time. 
\end{corollary}
\begin{proof}
This is Corollary~\ref{cor:zt} with $N=n$, applied to the two biadjacency matrices as above.
\end{proof}

For $|W|=O(n^2/\sqrt D)$, which is what the reductions below produce, the running time in Corollary~\ref{fact:lop} is $O(n^2/D^{0.063})$.

The reduction of Section~\ref{sec:redproof} is stated for an arbitrary $D$, with the number of oracle calls and the additional time as functions of $D$ and a free parameter $g$, and Section~\ref{sec:zwt} chooses $D$ and $g$.

\subsection{A deterministic reduction from Exact Triangle to \texorpdfstring{{\boldmath$\Lop$}}{Lop}}\label{sec:redproof}

\paragraph{Exact Triangle.}
An instance of \defn{Exact Triangle} (also called Zero-Weight Triangle, and sometimes ``Love Triangle''~\cite{GP18}) consists of a complete tripartite graph on vertex parts $A,B,C$ of $n$ vertices each, and an integer weight $w(e)$ on every edge, with $|w(e)|\le n^\nu$ for some constant $\nu\ge1$. Let $S(a,b,c):=w(a,b)+w(b,c)+w(a,c)$. A \defn{zero triangle} is a triangle $(a,b,c)\in A\times B\times C$ with $S(a,b,c)=0$, and the task is to decide whether one exists. As is standard, an instance on an arbitrary $n$-vertex graph reduces to this form by taking three copies of the vertex set and giving every missing edge the weight $3n^\nu+1$, which lies in no zero triangle (this replaces $\nu$ by $\nu+1$), and the search version, which asks the algorithm to find a zero triangle if one exists, is equivalent to the decision version via a simple self-reduction (the reduction below finds one anyway).

\begin{theorem}[Exact Triangle to $\Lop$, deterministically]\label{thm:red}
Let $16\le D\le n$, and let $1\le g\le\sqrt D$ be an integer. Exact Triangle on $n$ vertices per part with weights of absolute value at most $n^\nu$ reduces deterministically to at most $4ng$ instances of $\Lop(n,D)$, each with at most $n^2/\sqrt D$ query pairs, plus $O(\nu\,n^3\log n/g+n^{\omega+o(1)}D^{3/2}+n^2Dg)$ additional time, where $\omega<2.372$ is the matrix multiplication exponent~\cite{ADVXXZ25}. The reduction is non-adaptive: it produces all the instances before it makes any oracle call, so they can be solved in any order, or all together.
\end{theorem}

The parameter $g$ trades the number of instances against the cost of the search for witnesses in the last step of the proof (the term $\nu n^3\log n/g$); Section~\ref{sec:zwt} balances the two. The oracle returns at most $n^2/\sqrt D$ answers per instance, and the time to read them is counted as part of the oracle calls, rather than as additional time.

\begin{proof}
The proof has three steps: we hash the weights modulo a prime, build the instances, and search for witnesses.

\emph{Hashing modulo a prime.} We reduce the weights modulo a prime $p\in[\sqrt D/2,\sqrt D)$, chosen deterministically. This size of $p$ fits the instances below: a piece of $C$ of up to $\sqrt D$ vertices, paired with the $p$ labels, gives a middle part of at most $D$ vertices, and the $n^2$ pairs $(a,b)$ fall into $p$ residue classes of about $n^2/\sqrt D$ pairs on average, which is the number of query pairs that an instance may have. A \defn{false positive} of $p$ is a triple $(a,b,c)$ with $S(a,b,c)\ne0$ and $p\mid S(a,b,c)$; let $F(p)$ denote the number of false positives of $p$. For every prime $p$ in the range we count the triples with $S(a,b,c)\equiv0\pmod p$. This number is the number of false positives $F(p)$ plus $Z_0$, the number of zero triangles, which does not depend on $p$. Each $F(p)+Z_0$ is computed by a standard trick for small weights~\cite{AGM97,Zwi02}: let $P[a,c]:=x^{w(a,c)\bmod p}$ and $Q[c,b]:=x^{w(b,c)\bmod p}$ be matrices over the ring $\Z[x]/(x^p-1)$, so that the coefficient of $x^r$ in $(PQ)[a,b]$ is the number of $c\in C$ with $w(a,c)+w(b,c)\equiv r\pmod p$; then $F(p)+Z_0$ is the sum over the pairs $(a,b)\in A\times B$ of the coefficient of $x^{-w(a,b)\bmod p}$ in $(PQ)[a,b]$. Computing $PQ$ takes $n^{\omega+o(1)}$ ring operations, or $O(n^{\log_27})$ with Strassen's algorithm~\cite{Str69}, each of them $O(p^2)$ word operations, so $n^{\omega+o(1)}D^{3/2}$ time over the fewer than $\sqrt D$ primes in the range. (Any matrix multiplication algorithm fast enough to make this term negligible suffices. For our choice of $D$, Strassen's algorithm already does, and keeps the reduction potentially practical.)\footnote{Our reduction uses matrix multiplication, so that it is not ``combinatorial.'' One can obtain a combinatorial reduction to $\#\Lop(n,D)$ instead of $\Lop(n,D)$ by counting the false positives with the oracle itself (for instance, by choosing the modulus one prime factor at a time as in~\cite[Lemma~3.2]{FKP24}). We use matrix multiplication because we want to show that even the detection problem $\Lop(n,D)$ solves Exact Triangle deterministically.}

We select the prime with the smallest count, which is also the prime with the fewest false positives, and call it $p$. To bound $F(p)$, note that a triple with $S(a,b,c)\ne0$ is a false positive exactly of the primes in the range that divide $S(a,b,c)$. Since $0<|S(a,b,c)|\le3n^\nu$ and these primes are at least $\sqrt D/2$, there are at most $\log_{\sqrt D/2}(3n^\nu)$ of them. Hence the numbers of false positives of all the primes in the range add up to at most $n^3\log_{\sqrt D/2}(3n^\nu)$. %
By the prime number theorem there are $\Omega(\sqrt D/\log D)$ primes in the range, and $F(p)$ is at most the average over them, so, as in~\cite[Lemma~3.2]{FKP24}, $F(p)=O(n^3\log(3n^\nu)/\sqrt D)=O(\nu\,n^3\log n/\sqrt D)$, since $\log D=O(\log(\sqrt D/2))$ for $D\ge16$.

\emph{The instances.} As in~\cite{CX24}, let $s:=\lfloor\sqrt D\rfloor$ and split $C$ into pieces $C_1,\ldots,C_h$ of at most $\lceil s/g\rceil$ vertices each, so that $h\le\lceil ng/s\rceil$. For $\varrho\in\Z_p$ let $W_\varrho$ be the set of edges $(a,b)\in A\times B$ with $w(a,b)\equiv\varrho\pmod p$, and cut it into chunks of at most $n^2/\sqrt D$ query pairs. There are at most $p+\sqrt D\le2\sqrt D$ chunks in all. For each chunk $\mathcal Q\subseteq W_\varrho$ and each piece $C_k$ form the instance of $\Lop(n,D)$ with query pairs $W:=\mathcal Q$ and middle part $C_k\times\Z_p$, of size at most $sp\le D$, whose middle vertices are pairs (vertex, label) as in~\cite{CX24}, with
\[
a\sim(c,\sigma)\iff\sigma\equiv w(a,c)+\varrho,\qquad(c,\sigma)\sim b\iff\sigma\equiv-w(b,c)\pmod p .
\]
Within the chunk, the condition $S(a,b,c)\equiv0\pmod p$ has become the equality $w(a,c)+\varrho\equiv-w(b,c)$ of a label of $(a,c)$ and a label of $(b,c)$, so a query pair $(a,b)\in\mathcal Q$ has a common neighbor if and only if some $c\in C_k$ has $S(a,b,c)\equiv0\pmod p$. There are at most $2\sqrt D\,h\le2\sqrt D(ng/s+1)\le4ng$ instances, since $s\ge\sqrt D-1\ge3\sqrt D/4$ because $D\ge16$, and $\sqrt D\le\sqrt n\le2n/3$ because $D\le n$ and we may assume $n\ge3$, and they are produced non-adaptively, since $p$ is chosen before any of them. Writing them down costs $O(n^2Dg)$: the two bipartite graphs of an instance have $O(nD)$ entries, and the chunks are computed once and shared by the pieces.

\emph{Witnesses.} For every query pair that the oracle accepts, scan the piece $C_k$ of its instance for a $c$ with $S(a,b,c)=0$, as in the exhaustive search of~\cite{CX24}. A zero triangle is always found, since its $c$ lies in some piece and makes the oracle accept its pair. We stop as soon as a zero triangle is found. A failed scan, of a piece $C_k$ for a pair $(a,b)$, contains a $c\in C_k$ with $S(a,b,c)\equiv0\pmod p$ but $S(a,b,c)\ne0$, that is, a false positive $(a,b,c)$ of $p$, and distinct scans contain distinct false positives, so there are at most $F(p)$ failed scans, and the scans cost $O((F(p)+1)\sqrt D/g)=O(\nu\,n^3\log n/g)$.
\end{proof}

\begin{remark}[Comparison with~\cite{CX24} and~\cite{VX20a}]
We emulate the instances and the witness search of Chan and Xu~\cite[Section~3]{CX24}. The difference is in how the pairs are grouped so that the condition on $S(a,b,c)$ becomes an equality of labels. For exact integer weights, $w(a,b)$ takes too many values to group by, so Chan and Xu group the pairs by a reference witness $k$ with known $S(a,b,k)$, which they find by a recursion on the bits of the weights, and use Fredman's trick. Their labels are exact and there are no false positives. After hashing, $w(a,b)\bmod p$ takes only $p$ values, so we can group by it directly, in a single round of instances instead of a recursion. The labels take only $p$ values, so the middle part is small, but in exchange, this introduces false positives. Vassilevska Williams and Xu~\cite[Theorem~3.4]{VX20a} also hash, but by a random map $w(u,v)\mapsto x\,w(u,v)+y_u-y_v$ into a large prime field, with $x$ and the vertex weights $y_v$ random, followed by a split of the field into intervals. With residues in place of their intervals, the block $C_k\times\{-j\}$ of our instance for $\varrho$ and $C_k$ is their graph $G_{-\varrho-j,j,\varrho}$, and our instance puts the $p$ graphs with the same $\varrho$ side by side. They use randomness to bound the degrees and the false positives through each edge, which their listing step needs~\cite[Claims~3.5 and~3.6, Proposition~3.8]{VX20a}, whereas the scans above need only the total number of false positives, controlled by the choice of $p$.
\end{remark}

\subsection{Exact Triangle in truly subcubic time}\label{sec:zwt}

\begin{theorem}[Exact Triangle]\label{thm:zwt}
Let $\varepsilon':=0.00175$ and $\varepsilon_T:=0.0017$. For every constant $\nu\ge1$, Exact Triangle on $n$ vertices per part with integer weights of absolute value at most $n^\nu$ can be solved by a deterministic algorithm in $O(n^{3-1/648}\log^2n)$ time using Theorem~\ref{thm:main} (via Corollary~\ref{fact:lopmain}), and in $O(n^{3-\varepsilon'}\log n)\leq O(n^{3-\varepsilon_T})$ time using Corollary~\ref{cor:zt} (via Corollary~\ref{fact:lop}).
\end{theorem}
\begin{proof}
Assume $n$ is larger than a constant depending on $\nu$ (smaller instances are solved by brute force). In both cases we apply Theorem~\ref{thm:red}, whose instances have at most $n^2/\sqrt D$ query pairs each, and solve every instance by a corollary of Section~\ref{sec:tri}. The values of $g$ below balance the cost of the instances against that of the scans (Remark~\ref{rem:g}).

\emph{By Theorem~\ref{thm:main}.} Let $D$ be the largest power of four with $D\le n^{1/18}$, so that $D\ge n^{1/18}/4$, and let $g:=\lceil D^{1/36}\rceil$. Solve every instance by Corollary~\ref{fact:lopmain}, which applies because $D^{18}\le n$. The $O(nD^{1/36})$ instances cost $O(n^2\log^2\!D/D^{1/18})$ each, so $O(n^3D^{-1/36}\log^2n)$ in all, the scans cost $O(n^3D^{-1/36}\log n)$, the choice of $p$ costs $O(n^{\log_27}D^{3/2})=O(n^{2.9})$ with Strassen's algorithm, and building the instances costs $O(n^2D^{1.03})$. Finally $D^{-1/36}\le4^{1/36}n^{-1/648}$, so the time is $O(n^{3-1/648}\log^2n)$.

\emph{By Corollary~\ref{cor:zt}.} Let $D:=\lfloor n^{1/18}\rfloor$ and $g:=\lceil D^{0.0315}\rceil$, and solve every instance by Corollary~\ref{fact:lop}, which applies because $D^{18}\le n$. The instances cost $O(nD^{0.0315})\cdot O(n^2/D^{0.063})=O(n^3D^{-0.0315})$, the scans $O(n^3D^{-0.0315}\log n)$, the choice of $p$ costs $O(n^{\log_27}D^{3/2})=O(n^{2.9})$ with Strassen's algorithm, and building the instances costs $O(n^2D^{1.04})$. Finally $D\ge n^{1/18}/2$ gives $D^{-0.0315}\le2\,n^{-0.0315/18}=2\,n^{-0.00175}$, so the time is $O(n^{3-\varepsilon'}\log n)\leq O(n^{3-\varepsilon_T})$.
\end{proof}

\begin{remark}\label{rem:g}
With $g=D^\eta$ and a saving $D^{\gamma}$ per instance, the instances cost $n^3D^{\eta-\gamma}$ and the scans $n^3D^{-\eta}$, up to logarithmic factors. They balance at $\eta=\gamma/2$, that is, at $\eta=1/36$ for the saving $D^{1/18}$ of Theorem~\ref{thm:main} and at $\eta=0.0315$ for the saving $D^{0.063}$ of Corollary~\ref{cor:zt}. Any improvement of these savings improves the final exponent $\eta/18$ directly, and the factor $1/18$ in it comes from the assumption $N\ge D^{18}$ of these theorems. Any other algorithm for $\Lop(n,D)$ can be plugged into Theorem~\ref{thm:red} in the same way. With the straightforward algorithm ($O(n^2/g)$ per instance, since every vertex of $A$ has $O(\sqrt D/g)$ neighbors in the middle part) the reduction recovers the brute-force bound $n^3$ up to a logarithmic factor.
\end{remark}

\subsection{3SUM and APSP reduce to Exact Triangle}\label{sec:3sumapsp}

The reductions here are known and deterministic. We recall them in the form we use, and then apply Theorem~\ref{thm:zwt}, with its bound $O(n^{3-\varepsilon'}\log n)$, $\varepsilon'=0.00175$, when it uses Corollary~\ref{cor:zt}. 

\begin{theorem}[Known reductions]\label{thm:knownred}
For every constant $\nu$:
\begin{enumerate}
\item[(a)] \textup{($3$SUM)} $3$SUM on $n$ integers of absolute value at most $n^\nu$ reduces deterministically, in $n^{3/2+o(1)}$ time, to $n^{1/2+o(1)}$ instances of Exact Triangle on $n^{1/2+o(1)}$ vertices per part with weights of absolute value $n^{O(1)}$~\cite{CH20,VW13}.
\item[(b)] \textup{($(\min,+)$-product and APSP)} If a deterministic algorithm solves Exact Triangle on $s$ vertices per part with weights of absolute value at most $cU$, for a suitable constant $c$, in time $T(s)$ with $T(s)/s$ nondecreasing, then the $(\min,+)$-product of two $n\times n$ integer matrices with entries of absolute value at most $U$ can be computed deterministically in $O(n^2T(n^{1/3})\log^2U)$ time, and APSP on directed $n$-vertex graphs with integer weights of absolute value at most $n^\nu$ and no negative cycles in $O(n^2T(n^{1/3})\log^3n)$ time~\cite{VW10,VW18,VW13}.
\end{enumerate}
\end{theorem}

Using the above known reductions we solve both $3$SUM and APSP polynomially faster:
\begin{theorem}[$3$SUM and APSP]\label{thm:3sumapsp}
Let $\nu$ be a constant. Using Theorem~\ref{thm:main} (through Theorem~\ref{thm:zwt}), deterministic algorithms solve $3$SUM on $n$ integers of absolute value at most $n^\nu$ in $n^{2-1/1296+o(1)}\leq O(n^{1.99923})$ time, and the $(\min,+)$-product of two $n\times n$ integer matrices with entries of absolute value at most $n^\nu$, as well as APSP on directed $n$-vertex graphs with integer weights of absolute value at most $n^\nu$ and no negative cycles, in $\tilde O(n^{3-1/1944})\leq O(n^{2.99949})$ time. Using Corollary~\ref{cor:zt} instead, the times are $n^{2-\varepsilon'/2+o(1)}\leq O(n^{1.9992})$ and $\tilde O(n^{3-\varepsilon'/3})\leq O(n^{2.99942})$.
\end{theorem}
\begin{proof}
Plug Theorem~\ref{thm:zwt} into Theorem~\ref{thm:knownred}. With the bound $O(n^{3-1/648}\log^2n)$ of Theorem~\ref{thm:zwt}, the cost for $3$SUM is
\[
n^{3/2+o(1)}+n^{1/2+o(1)}\cdot O\big((n^{1/2+o(1)})^{3-1/648}\log^2n\big)=n^{2-1/1296+o(1)}=n^{1.999228\ldots+o(1)} ,
\]
and for the $(\min,+)$-product and APSP, with $T(s)=O(s^{3-1/648}\log^2s)$, for which $T(s)/s$ is nondecreasing, it is
\[
\tilde O\big(n^2(n^{1/3})^{3-1/648}\big)=\tilde O\big(n^{3-1/1944}\big),\quad 3-1/1944=2.99948\ldots<2.99949 .
\]
With the bound $O(n^{3-\varepsilon'}\log n)$ instead, the cost for $3$SUM is
\[
n^{3/2+o(1)}+n^{1/2+o(1)}\cdot O\big((n^{1/2+o(1)})^{3-\varepsilon'}\log n\big)=n^{2-\varepsilon'/2+o(1)}=n^{1.999125+o(1)} ,
\]
and for the $(\min,+)$-product and APSP, with $T(s)=O(s^{3-\varepsilon'}\log s)$, it is
\[
\tilde O\big(n^2(n^{1/3})^{3-\varepsilon'}\big)=\tilde O\big(n^{3-\varepsilon'/3}\big),\qquad 3-\varepsilon'/3=2.99941\ldots<2.99942 .
\]
\end{proof}

\begin{remark}[3XOR]\label{rem:3xor}
The same approach as for $3$SUM gives a deterministic truly subquadratic algorithm for 3XOR (given three lists of $n$ vectors in $\mathbb F_2^{O(\log n)}$, find one vector from each list such that the three XOR to zero~\cite{JV16,DSW18}), using linear hash functions over $\mathbb F_2$ whose rows are chosen one at a time from an $\varepsilon$-biased set~\cite{AGHP92} by the method of conditional expectations.
\end{remark}

%% file: 04-general.tex
\section{The matrix theorem in general: a data structure}\label{sec:general}

In this section, we extend our algorithm from Section~\ref{sec:result} in two directions. First, we turn it into a data structure. After preprocessing $X$ and $Y$ in less time than it takes to write down $XY$, the data structure answers a query for any single entry of $XY$, not known in advance, in time polynomially smaller than the $D$ operations that it would take to compute that entry as an inner product. Second, we let the parameters of our construction vary. This will give larger time savings than $D^{1/18}$, and trade-offs between the time savings, how thin the product must be, the query time, and the density of the set $W$ of wanted entries. We assume that the reader is familiar with Section~\ref{sec:result}, whose notation we use throughout, although we will begin with a brief reminder of the relevant notions.

In Section~\ref{sec:results}, we state the main results, Theorems~\ref{thm:single} and~\ref{thm:general}, with explicit parameter examples in Table~\ref{tab:ds}, as well as Corollary~\ref{cor:zt}, the instance of Theorem~\ref{thm:single} that our reductions in Sections~\ref{sec:reduction} and~\ref{sec:other} use. In Section~\ref{sec:boxes}, we will introduce the key new idea, the partial sums that our data structure stores, which we call boxes. In Section~\ref{sec:single}, we will then give the data structure in terms of the parameters (listed in Table~\ref{tab:params}), and finally we will set the parameters in Section~\ref{sec:corproofs}.

\begin{table}[!t]
\centering
\small
\renewcommand{\arraystretch}{1.25}
\begin{tabular}{@{}l>{\raggedright\arraybackslash}p{6.3cm}ll@{}}
symbol & meaning & in Section~\ref{sec:result} & in this section\\\hline
$m$, $D=4^m$ & the number of inner levels, and the inner dimension of the product $XY$ & & \\
$L$, $c=L/m$ & the number of levels of the recursion, and its ratio to $m$ & $L=19m$ & $L=cm$, $c>10$\\
$\Nz=3^{L-m}$ & the number of rows of a row block, and of columns of a column block & & \\
$K=\binom Lm$ & the number of products $X_QY_Q$ of a tile, one for each set $Q$ of $m$ levels & & \\
$\Cm=K\Nz^2$ & the number of output entries of a tile & & \\
$\alpha_d=\binom md9^d$ & the number of leaves of order $d$ contributing to one output entry & & \\
$\beta_d=\binom L{m-d}9^{L-m+d}$ & the number of leaves of order $d$ in a tile & $\beta_d\le2^{-d}\Cm$ & $\beta_d\le\rho^d\Cm$\\
$\rho=9m/(L-m+1)$ & the decay rate of the $\beta_d$: $\beta_d/\beta_{d-1}\le\rho$, see~\eqref{eq:gendecay} & $\rho<\frac12$ & $\rho<1$\\
$t$ & the switching order: a query reads the leaves of order below $t$, and the preprocessing puts the others into boxes; in Section~\ref{sec:result}, the order at which the count of Lemma~\ref{lem:count} splits & $t=m/9$ & $t=\theta m$, $0<\theta<0.9$\\
$\gamma$ & the preprocessing, or the whole computation when $W$ is given in advance, takes $O(N^2\log^2\!D/D^{\gamma})$ time & $\gamma=1/18$ & see Table~\ref{tab:ds}\\
$\varepsilon$ & the results hold for $D\le N^{\varepsilon}$ & $\varepsilon=1/18$ & $\varepsilon<R_c(\gamma)<0.1204\ldots$\\
$q$ & a query takes $O(D^{q}\log D)$ time & no queries & see Table~\ref{tab:ds}\\
$\kappa$ & the set $W$ of wanted entries has $|W|\le N^2/D^{\kappa}$ & $\kappa=\frac12$ & see Table~\ref{tab:ds}
\end{tabular}
\caption{The parameters of our construction. The last two columns give their values or bounds in Section~\ref{sec:result} and in this section. Where these columns are empty, the parameter is given by the formula in the first column in both sections. In this section, $L$ and $t$ are $cm$ and $\theta m$ rounded up to integers, and the bound $R_c(\gamma)$ is defined in Section~\ref{sec:corproofs}.}
\label{tab:params}
\end{table}

\paragraph{Notions from Section~\ref{sec:result}.}
We keep the recursion, the tiling, and the encodings of Sections~\ref{sec:full} and~\ref{sec:encoders}, and we use the terminology of Sections~\ref{sec:whatit} and~\ref{sec:counts} as defined there: the leaves contributing to an output string, its private leaf, the order of a leaf, and the counts $\alpha_d$ and $\beta_d$. For convenience, we briefly recall these notions here.
\begin{itemize}[nosep]
\item A \emph{leaf} $\tau=\tau_1\cdots\tau_L$ is a string of $L$ terms of Sch\"onhage's identity, where $\tau_\ell$ is the term chosen at \emph{level} $\ell$ of the recursion (Section~\ref{sec:recursion}). Throughout this section, we compare levels by their indices: we say that level $\ell$ is \emph{lower} than level $\ell'$ if $\ell<\ell'$, so that the lowest level is the first level of the recursion.
\item An \emph{output string} $w$ with \emph{inner set} $Q$ has the variable $z_0$ at each of the $m$ levels of $Q$, and one of the variables $z_{ij}$ at every other level. It indexes the entry $(X_QY_Q)[w]$ of the product $X_QY_Q$ (Section~\ref{sec:defs}).
\item A leaf \emph{contributes to} $w$ if it chooses $P_{ij}$ at every level where $w$ has $z_{ij}$, while it may choose any of the ten terms at the levels of $Q$, where $w$ has $z_0$ (Section~\ref{sec:whatit}).
\item The \emph{private leaf} of $w$ is the leaf contributing to $w$ that chooses $P_0$ at every level of $Q$ (Section~\ref{sec:counts}).
\item The \emph{order} of a leaf is $m$ minus the number of levels at which it chooses $P_0$. For a leaf contributing to $w$, this is the number of levels of $Q$ at which it chooses a term other than $P_0$ (Section~\ref{sec:counts}).
\item A \emph{tile} consists of a \emph{band} of $K_0=\lfloor\sqrt K\rfloor$ row blocks of $X$ together with a band of $K_0$ column blocks of $Y$ (Section~\ref{sec:tiling}), and the \emph{encoding} of a band is the array of the $10^L$ numbers $\Phi_\tau(a)$, or $\Psi_\tau(b)$, computed from its input array (Section~\ref{sec:encoders}). We will recall both in more detail in Section~\ref{sec:single}.
\end{itemize}

As in Section~\ref{sec:counts}, all output strings in this section have inner sets of exactly $m$ elements, so that they index the $\Cm=K\Nz^2$ entries of the $K$ products $X_QY_Q$ of a tile, and we will use that $\Cm=\beta_0\le10^L$. For a leaf $\tau$ of a tile with input arrays $a$ and $b$, we call $\Phi_\tau(a)\Psi_\tau(b)$ the \defn{product at} $\tau$. We read its two factors from the encodings of the tile (Section~\ref{sec:encoders}), and by~\eqref{eq:multdef} and Lemma~\ref{lem:correct}, $(X_QY_Q)[w]$ is the sum of the products at the leaves contributing to $w$. The one change from Section~\ref{sec:result} is that $L$ is now any integer with $L\ge10m$, and we call $c=L/m$ the \defn{ratio}.

\paragraph{Technique overview.}
The pruned recursion of Section~\ref{sec:sparse} needs to know $W$ in advance, since it visits exactly the leaves that the wanted entries need. When the queries arrive one at a time, we must instead prepare for every possible output entry.

Recall that our preprocessing must take less time than it takes to write down $XY$. This means it stores fewer than one number per entry of $XY$, so each number that it stores must be useful for many different entries. Meanwhile, a query must take far fewer than the $D$ operations that it would take to compute the entry directly as an inner product. In Section~\ref{sec:result}, the running time had three parts: the leaves of small order, the leaves of large order, and the encodings (Lemma~\ref{lem:count} and Section~\ref{sec:proof}). Our data structure will divide these parts between the queries and the preprocessing. Each output entry uses only a few of the leaves of small order, so a query may just read them one by one. By contrast, each output entry may use many leaves of large order, but there are not too many such leaves in total, and each of them is shared by many output entries. Thus, in preprocessing we will add them up in advance, in groups which can be quickly summed during queries. The preprocessing also computes the encodings, once.

The $10^m$ leaves contributing to an output string form a cube: at each level of the inner set $Q$ of the output string, they choose any of the ten terms, and at each other level, they agree with the private leaf of the output string. Every one of these leaves, other than the private leaf, also contributes to other output entries. We will therefore add up parts of every cube in advance, into what we call boxes (which are themselves smaller cubes), so that a query adds up a few box values instead of $10^m$ products. For this we need a rule that cuts every cube into few boxes, with each of its leaves of large order in exactly one of them; we will prove this for our rule in Lemma~\ref{lem:partition}. We will make sure that the definition of a box does not refer to the output entry, so that the same box serves many output entries, which keeps the total number of boxes small, as we will show in Lemma~\ref{lem:allboxes}.

We also need to compute all the boxes of a tile quickly in preprocessing, and here we will use dynamic programming. Since a box is a cube, it has a fixed term at some levels, and its other levels are free, meaning all ten terms are allowed there. Fixing each of the ten terms in turn at the highest free level splits the box into ten smaller boxes, and its value is the sum of their values. Thus, in Lemma~\ref{lem:allboxes}, we will compute the boxes from the smallest up, each from the values of smaller boxes computed before.

The order $t$ at which a query switches from leaves to boxes, which we call the \defn{switching order}, trades off the query time (about $\alpha_t$) against the preprocessing time (about $\rho^t$ boxes per entry of the product, where $\rho<1$ is the rate at which the $\beta_d$ decay). The ratio $c=L/m$ (which we fixed to be $19$ in Section~\ref{sec:result}, but we now allow to vary) trades off the time savings against how thin the product must be. While a larger $c$ makes $\rho$ smaller, it also makes the encodings larger relative to $D$, so that computing them stays within the time bound only when $N$ is a larger power of $D$. With $c=21$ and $t=m/9$, we will get the instance $N\ge D^{18}$ of Corollary~\ref{cor:zt}, and pushing $c$ down toward $10$ will give the upper limit for $D$ of our technique, namely $D\le N^{0.1204}$  (Corollary~\ref{cor:tradeoff}).

\subsection{The results}\label{sec:results}

In this section, we design a data structure with the following guarantee. Let $\varepsilon^*:=\ln4/(5\ln10)=0.1204\ldots$.

\begin{theorem}\label{thm:single}
For every $\varepsilon<\varepsilon^*$, and every $q>0$, there is a $\gamma>0$ such that the following holds. Given as input matrices $X\in\Z^{N\times D}$ and $Y\in\Z^{D\times N}$, where $2\le D\le N^{\varepsilon}$, whose entries are integers of absolute value at most $N^{O(1)}$, we can preprocess them deterministically in $O(N^2\log^2\!D/D^{\gamma})$ time and space, after which any single entry $(XY)[I,J]$ can be computed deterministically in $O(D^{q}\log D)$ time.
\end{theorem}

The constant $\varepsilon^*$ gives the upper limit for $D$ of our technique, and we will explain where it comes from in Section~\ref{sec:corproofs} below (after the proof of Corollary~\ref{cor:tradeoff}). To achieve the general form of the matrix theorem, where we know the positions in advance as in Section~\ref{sec:result}, we simply ask the data structure one query per position. Setting the parameters to balance the preprocessing time and the total time for all the queries gives the following.

\begin{theorem}\label{thm:general}
For every $\varepsilon<\varepsilon^*$ and every $\kappa>0$ there is a $\gamma>0$ such that the following holds. Given as input matrices $X\in\Z^{N\times D}$ and $Y\in\Z^{D\times N}$, where $2\le D\le N^{\varepsilon}$, whose entries are integers of absolute value at most $N^{O(1)}$, as well as a set $W$ of at most $N^2/D^{\kappa}$ positions of an $N\times N$ matrix, the entries $(XY)[I,J]$, $(I,J)\in W$, can be computed deterministically in time $O(N^2\log^2\!D/D^{\gamma})$.
\end{theorem}

The logarithmic factors in both theorems can be removed by halving $\gamma$ and applying Theorem~\ref{thm:single} with $q/2$ in place of $q$; for $D=1$ the bounds are trivial.

\begin{table}[t]
\centering
\small
\setlength{\tabcolsep}{3pt}
\begin{tabular}{c|c|cccccc|ccccc}
 & & \multicolumn{6}{c|}{Theorem~\ref{thm:single}: the data structure} & \multicolumn{5}{c}{Theorem~\ref{thm:general}: the wanted entries}\\
ratio & $D\le N^{\varepsilon}$ & \multicolumn{6}{c|}{$\gamma$ for query time $D^{q}$, $q=$} & \multicolumn{5}{c}{$\gamma$ for $|W|\le N^2/D^{\kappa}$, $\kappa=$}\\
$c$ & $\varepsilon$ & $0.10$ & $0.25$ & $0.43$ & $0.50$ & $0.75$ & $0.90$ & $0.10$ & $0.25$ & $0.50$ & $0.75$ & $1.00$\\\hline
$40$ & $0.029$ & $0.0205$ & $0.0608$ & $0.1175$ & $0.1429$ & $0.2417$ & $0.3095$ & $0.0166$ & $0.0473$ & $0.1060$ & $0.1720$ & $0.2442$\\
$21$ & $0.056$ & $0.0112$ & $0.0331$ & $\mathbf{0.0640}$ & $0.0778$ & $0.1316$ & $0.1685$ & $0.0099$ & $0.0286$ & $0.0653$ & $0.1074$ & $0.1546$\\
$19$ & $0.062$ & $0.0097$ & $0.0287$ & $\mathbf{0.0555}$ & $0.0675$ & $0.1142$ & $0.1463$ & $0.0087$ & $0.0253$ & $0.0578$ & $0.0954$ & $0.1379$\\
$15$ & $0.079$ & $0.0061$ & $0.0183$ & $0.0354$ & $0.0430$ & $0.0728$ & $0.0932$ & $0.0057$ & $0.0168$ & $0.0389$ & $0.0646$ & $0.0941$\\
$12$ & $0.099$ & $0.0028$ & $0.0083$ & $0.0160$ & $0.0195$ & $0.0330$ & $0.0423$ & $0.0027$ & $0.0080$ & $0.0186$ & $0.0312$ & $0.0459$\\
$10.5$ & $0.114$ & $0.0007$ & $0.0022$ & $0.0043$ & $0.0052$ & $0.0089$ & $0.0114$ & $0.0007$ & $0.0022$ & $0.0052$ & $0.0087$ & $0.0129$
\end{tabular}
\caption{Explicit parameters for Theorems~\ref{thm:single} and~\ref{thm:general}. Each row is a choice of the ratio $c=L/m$. The parameter $\varepsilon$ is determined by $c$ and says how thin the product must be, i.e., every entry of the row holds whenever $D\le N^{\varepsilon}$. On the left, the columns correspond to a choice of $q$, and the numerical entry is a value $\gamma$ for which Theorem~\ref{thm:single} holds, meaning that after $O(N^2\log^2\!D/D^{\gamma})$ preprocessing, a query takes $O(D^{q}\log D)$ time. On the right, the columns correspond to a choice of $\kappa$, and the numerical entry is a value $\gamma$ for which Theorem~\ref{thm:general} holds, meaning that any $|W|\le N^2/D^{\kappa}$ wanted entries take $O(N^2\log^2\!D/D^{\gamma})$ time. On the left side, the column $q=0.43$ uses the switching order $t=m/9$ of Section~\ref{sec:result} (more precisely, $q=0.4277\ldots$), and its two bold entries are Theorem~\ref{thm:main} ($c=19$, where $\gamma=1/18=0.0555\ldots$) and Corollary~\ref{cor:zt} ($c=21$). A larger $c$ gives a larger $\gamma$ but a smaller $\varepsilon$, and a slower query or a sparser $W$ allows a larger $\gamma$. We round all values down, and we will explain how we compute them in Section~\ref{sec:corproofs}.}
\label{tab:ds}
\end{table}

In Table~\ref{tab:ds}, we list several parameter settings for Theorems~\ref{thm:single} and~\ref{thm:general}. Theorem~\ref{thm:main} is the entry $c=19$, $q=0.43$, applied to a set $W$ of $N^2/\sqrt D$ wanted entries (that is, $\kappa=\frac12$). 
We next state explicitly the parameter setting that we use in the introduction and in our reductions in Sections~\ref{sec:reduction} and~\ref{sec:other}. It is the entry $c=21$, $q=0.43$ of Table~\ref{tab:ds}, whose switching order is $t=m/9$, as in Section~\ref{sec:result}, with the logarithmic factors absorbed into the exponents. 

\begin{corollary}\label{cor:zt}
Let $N\ge D^{18}$, and let $X\in\Z^{N\times D}$ and $Y\in\Z^{D\times N}$ have entries of absolute value at most $N^{O(1)}$. We can preprocess them deterministically in $O(N^2/D^{0.063})$ time and space, after which any single entry $(XY)[I,J]$ can be computed deterministically in $O(D^{0.437})$ time. Hence, for every set $W$ of positions of an $N\times N$ matrix, the entries $(XY)[I,J]$, $(I,J)\in W$, can be computed deterministically in $O(|W|\,D^{0.437}+N^2/D^{0.063})$ time, which is $O(N^2/D^{0.063})$ whenever $|W|\le N^2/\sqrt D$.
\end{corollary}

\subsection{Boxes}\label{sec:boxes}

In this subsection, we define the boxes and prove  two important facts about them. Intuitively, a box is a group of leaves that is shared by many output strings. In preprocessing we will add up the products at the leaves of each box, and then that value can help to answer many queries. In terms of the switching order $t$, we will ensure that the boxes of an output string partition its leaves of order at least $t$, and that there are only $\alpha_t$ of them. We will thus see, in Lemmas~\ref{lem:partition} and~\ref{lem:boxes}, that every output entry is the sum of a few products and a few values of boxes.

\paragraph{Cubes.}
A \defn{cube} is a string $\pi=\pi_1\cdots\pi_L$ of $L$ symbols, each of which is one of the ten terms of Sch\"onhage's identity or a star~$\ast$. Its \defn{leaves} are the leaves obtained by replacing each star by any of the ten terms, so a cube with $e$ stars has $10^e$ leaves, and a cube without stars is a single leaf. The \defn{value} of a cube $\pi$ is the sum of the products at its leaves,
\[
\val(\pi):=\sum_{\tau\text{ a leaf of }\pi}\Phi_\tau(a)\,\Psi_\tau(b).
\]
Every output entry is the value of a cube. Indeed, recall that the entry $(X_QY_Q)[w]$ of an output string $w$ with inner set $Q$ is the sum of the products at the $10^m$ leaves contributing to $w$. These are the leaves of the \defn{cube of $w$}, the cube $\pi^w$ that has a star at each level where $w$ has $z_0$ (the levels of $Q$) and the term $P_{ij}$ at each level where $w$ has $z_{ij}$, so that $(X_QY_Q)[w]=\val(\pi^w)$. In Figure~\ref{fig:terms}, this cube is all the rows taken together. Figure~\ref{fig:tree} shows the cube of an output string, and the boxes defined next, in the recursion tree of a small example.

We will cut the cube of every output string into parts that are again cubes, in a way where each leaf of the cube lies in exactly one part, so that
\[
(X_QY_Q)[w]=\val(\pi^w)=\sum_{\text{parts }\pi\text{ of }\pi^w}\val(\pi),
\]
and the value of a part will be shared by all the output strings that have it as a part. In Lemma~\ref{lem:boxes} below, we will prove this equation for the parts that we will choose. There will be two kinds of parts. Each leaf of small order will be a part on its own (recall that a single leaf is a cube without stars), whereas the leaves of large order will be grouped together into larger parts, which we call boxes.

\begin{figure}[t]
\centering
\begin{tikzpicture}[x=1cm,y=1cm,
  call/.style={draw,rounded corners=1pt,inner sep=2pt,font=\scriptsize,fill=blue!8,minimum height=4mm,minimum width=7.5mm},
  leaf/.style={draw=black!60,fill=white,inner sep=0pt,minimum size=2.5mm},   %
  leafone/.style={leaf,fill=black!35},   %
  leafzero/.style={leaf,draw=black,fill=black!80},   %
  wire/.style={gray!70,line width=0.35pt},
  box/.style={draw=orange!70!black,fill=orange!25,rounded corners=2pt,line width=0.4pt},
  lab/.style={font=\scriptsize,align=left}]
\def\cw{1.06}   %
\def\rh{0.38}   %
\def\terms{P_0,P_{11},P_{12},P_{13},P_{21},P_{22},P_{23},P_{31},P_{32},P_{33}}
\draw[black!45,line width=0.4pt,rounded corners=3pt] (-0.5,0.3) rectangle (9*\cw+0.5,-9*\rh-0.3);
\foreach \row in {1,...,9} { \draw[box] (-0.38,-\rh*\row-0.155) rectangle (9*\cw+0.38,-\rh*\row+0.155); }
\foreach \col in {1,...,9} { \draw[box] (\col*\cw-0.25,-0.155) rectangle (\col*\cw+0.25,0.155); }
\node[call] (root) at (4.5*\cw,2.33) {root};
\foreach \tm [count=\col from 0] in \terms {
  \node[call] (a\col) at (\col*\cw,1.63) {$\tm$};
  \node[call] (b\col) at (\col*\cw,0.93) {$P_{21}$};
  \draw[wire] (root.south) -- ++(0,-0.1) -| (a\col.north);
  \draw[wire] (a\col) -- (b\col);
  \draw[wire] (b\col.south) -- (\col*\cw,-9*\rh);
  \foreach \row in {0,...,9} {
    \pgfmathtruncatemacro{\ord}{(\col>0)+(\row>0)}
    \ifcase\ord \node[leafzero] at (\col*\cw,-\rh*\row) {};
    \or \node[leafone] at (\col*\cw,-\rh*\row) {};
    \else \node[leaf] at (\col*\cw,-\rh*\row) {}; \fi
  }
}
\draw[densely dashed,line width=0.6pt,rounded corners=3pt] (2*\cw-0.47,1.93) rectangle (2*\cw+0.47,-9*\rh-0.22);
\node[lab,anchor=base] at (2*\cw,-9*\rh-0.65) {one subtree};
\node[lab,anchor=base east] at (9*\cw+0.5,-9*\rh-0.65) {the cube $\pi^w=\ast\,P_{21}\,\ast$};
\begin{scope}[shift={(3.45,-9*\rh-0.65)}]
  \node[lab,anchor=base west] at (0,0) {order:};
  \node[leafzero] at (1.15,0.09) {}; \node[lab,anchor=base west] at (1.28,0) {$0$};
  \node[leafone] at (1.90,0.09) {};  \node[lab,anchor=base west] at (2.03,0) {$1$};
  \node[leaf] at (2.65,0.09) {};     \node[lab,anchor=base west] at (2.78,0) {$2$};
\end{scope}
\node[lab,anchor=east] at (-0.65,1.63) {level $1\in Q$};
\node[lab,anchor=east] at (-0.65,0.93) {level $2\notin Q$};
\foreach \tm [count=\row from 0] in \terms { \node[font=\scriptsize,anchor=east] at (-0.6,-\rh*\row) {$\tm$}; }
\draw[decorate,decoration={brace,amplitude=4pt,mirror},line width=0.35pt] (-1.22,0.15) -- (-1.22,-9*\rh-0.15) node[midway,left=5pt,lab] {level $3\in Q$};
\node[lab,anchor=east] (pl) at (-0.6,0.38) {private leaf};
\draw[-{Stealth[length=3pt]},line width=0.35pt] (pl.east) -- (-0.15,0.15);
\node[lab,anchor=west] at (9*\cw+0.62,0.93) {$P_{21}$ is the term\\ of $w$ at level $2$};
\node[lab,anchor=west] at (9*\cw+0.62,0) {$9$ boxes $P_{ij}\,P_{21}\,P_0$};
\draw[decorate,decoration={brace,amplitude=4pt},line width=0.35pt] (9*\cw+0.62,-\rh+0.155) -- (9*\cw+0.62,-9*\rh-0.155) node[midway,right=5pt,lab] {$9$ boxes\\ $\ast\,P_{21}\,P_{ij}$};
\end{tikzpicture}
\caption{The cube and the boxes of the output string $w=z_0z_{21}z_0$ in the recursion tree, for $L=3$ and switching order $t=1$. Since $w$ has $z_0$ at levels~$1$ and~$3$, its inner set is $Q=\{1,3\}$ (so $m=2$). A query for $w$ reads the private leaf, which has order $0<t$, and the values of the $\alpha_1=18$ boxes of $w$, which contain the other $99$ leaves (of orders $1$ and $2$). A box with a star is not a subtree: it consists of the same leaf of every subtree.}
\label{fig:tree}
\end{figure}
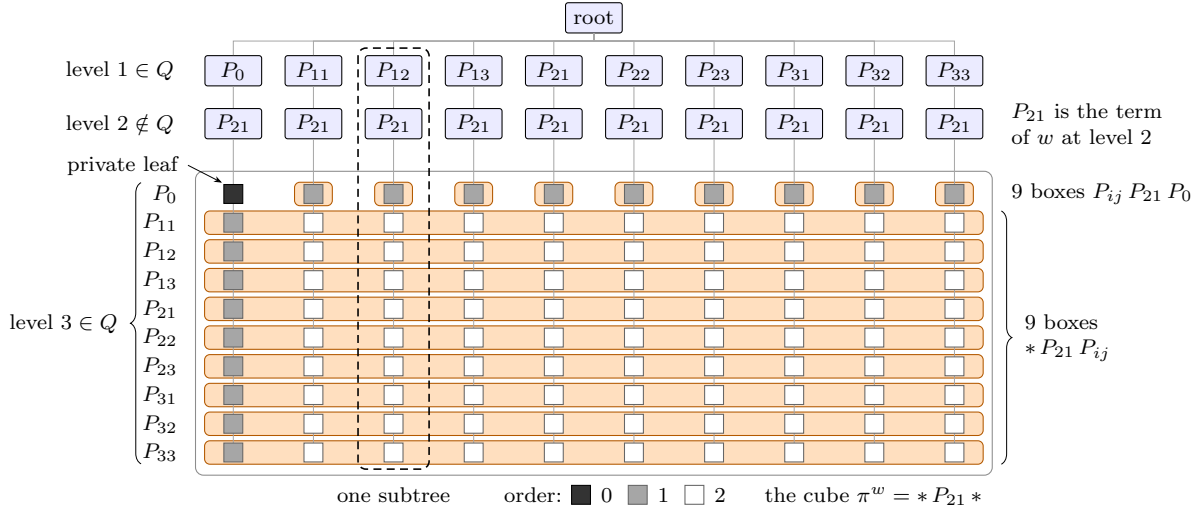

\paragraph{The definition of a box.}
The order of a leaf tells us how widely it is shared between output entries. Recall that the order is $m$ minus the number of levels at which the leaf chooses $P_0$, so that $\alpha_d=\binom md9^d$ leaves of order $d$ contribute to an output string, and that a leaf of larger order contributes to more output strings. Throughout, we fix the switching order $t$, an integer with $0\le t\le m$. The boxes will collect the leaves of order at least $t$, i.e., the leaves that choose $P_0$ at most $m-t$ times, and a query will read the leaves of order below $t$ one by one (Section~\ref{sec:single}).

Before giving the formal definition, we describe the idea informally and give some examples. Consider an output string $w$ with inner set $Q$, and a leaf $\tau$ of order at least $t$ contributing to $w$. To find the box of $\tau$, we look at the levels of $Q$ one at a time, from the highest down, and we stop as soon as we have seen $t$ levels at which $\tau$ chooses a term other than $P_0$. The box of $\tau$ then agrees with $\tau$ everywhere, except that it has a star at every level of $Q$ below the level where we stopped. For instance, let $m=4$ and $t=2$ as in Figure~\ref{fig:boxes}, and suppose that $\tau$ chooses the terms $P_{12},P_0,P_{31},P_{22}$ at the four levels $\ell_1<\ell_2<\ell_3<\ell_4$ of $Q$. We see $P_{22}$ at level $\ell_4$ and $P_{31}$ at level $\ell_3$, and we stop there, so the box of $\tau$ is $\ast\;\ast\;P_{31}\;P_{22}$ (in the first row of the figure). This box contains $\tau$ along with the $99$ other leaves that differ from $\tau$ only at the levels $\ell_1$ and $\ell_2$. Similarly, for a leaf that chooses the terms $P_0,P_{23},P_0,P_{11}$, we stop at level $\ell_2$, so its box is $\ast\;P_{23}\;P_0\;P_{11}$ (in the second row of the figure). Since a box has stars at the levels that we did not reach, it contains many leaves. Moreover, although we described the box of $\tau$ in terms of $w$, the definition of a box below does not refer to an output string, and the same box serves many output strings.

\begin{figure}[t]
\centering
\begin{tikzpicture}[x=1cm,y=1cm,
  cell/.style={draw,minimum width=8mm,minimum height=6mm,inner sep=1pt,font=\small},
  bstar/.style={cell,fill=orange!22},
  bzero/.style={cell,fill=blue!12},
  lab/.style={font=\small,anchor=west,align=left},
  every node/.append style={text height=1.6ex,text depth=0.35ex}]
\newcommand{\bxrow}[5]{%
  \node[anchor=east,font=\small] at (0.3,#1) {#2};
  \foreach \txt/\sty [count=\k] in {#3} { \node[\sty] at (\k*0.85,#1) {$\txt$}; }
  \node[lab] at (4*0.85+0.6,#1) {#4};
  \node[lab] at (4*0.85+5.0,#1) {#5};
}
\node[anchor=east,font=\small] at (0.3,0.75) {$V$};
\foreach \k in {1,...,4} { \node[font=\small] at (\k*0.85,0.75) {$\ell_\k$}; }
\node[lab] at (4*0.85+0.6,0.75) {sets $Z$ of its leaves};
\node[lab] at (4*0.85+5.0,0.75) {leaves per box};
\bxrow{0}{$\{\ell_1,\ell_2\}$}{\ast/bstar,\ast/bstar,P_{ij}/cell,P_{ij}/cell}{$\emptyset,\ \{\ell_1\},\ \{\ell_2\},\ \{\ell_1,\ell_2\}$}{$10^2$}
\bxrow{-0.75}{$\{\ell_1,\ell_3\}$}{\ast/bstar,P_{ij}/cell,P_0/bzero,P_{ij}/cell}{$\{\ell_3\},\ \{\ell_1,\ell_3\}$}{$10$}
\bxrow{-1.5}{$\{\ell_1,\ell_4\}$}{\ast/bstar,P_{ij}/cell,P_{ij}/cell,P_0/bzero}{$\{\ell_4\},\ \{\ell_1,\ell_4\}$}{$10$}
\bxrow{-2.25}{$\{\ell_2,\ell_3\}$}{P_{ij}/cell,P_0/bzero,P_0/bzero,P_{ij}/cell}{$\{\ell_2,\ell_3\}$}{$1$}
\bxrow{-3.0}{$\{\ell_2,\ell_4\}$}{P_{ij}/cell,P_0/bzero,P_{ij}/cell,P_0/bzero}{$\{\ell_2,\ell_4\}$}{$1$}
\bxrow{-3.75}{$\{\ell_3,\ell_4\}$}{P_{ij}/cell,P_{ij}/cell,P_0/bzero,P_0/bzero}{$\{\ell_3,\ell_4\}$}{$1$}
\end{tikzpicture}
\caption{The boxes of an output string $w$ for $m=4$ and $t=2$, drawn at the four levels $\ell_1<\ell_2<\ell_3<\ell_4$ of its inner set $Q$ (at every other level, they have the term of the private leaf of $w$). Here $P_{ij}$ stands for one of the nine terms other than $P_0$, and $\ast$ for all ten terms. Each row is one of the six sets $V$ of $m-t=2$ levels and stands for the $9^t=81$ boxes of $\mathcal B_V$, one for each choice of the two terms $P_{ij}$. Read from right to left, a row stops at its second $P_{ij}$, and the levels to the left of it, those of $F_V$, carry stars. On the right are the sets $Z$ of levels at which the leaves of such a box choose $P_0$, and the number $10^{|F_V|}$ of its leaves. Every set $Z$ of at most two levels appears in exactly one row (Lemma~\ref{lem:partition}), so these $6\cdot81=486=\alpha_2$ boxes contain each of the $\alpha_2+\alpha_3+\alpha_4=9963$ leaves of order at least $2$ contributing to $w$ exactly once.}
\label{fig:boxes}
\end{figure}

A \defn{box} is a cube $\pi=\pi_1\cdots\pi_L$ such that (i)~at most $m-t$ of its symbols are $P_0$ or stars, and (ii)~every star is at a lower level than every $P_0$. That is, the symbols $P_0$ and stars of a box, read in increasing order of the levels, are
\[
\underbrace{\ast\;\ast\;\cdots\;\ast}_{e}\;\;\underbrace{P_0\;P_0\;\cdots\;P_0}_{f-e}\qquad\text{for some }0\le e\le f\le m-t,
\]
and its other $L-f$ symbols are among the nine terms $P_{ij}$. A leaf of $\pi$ chooses $P_0$ at the $f-e$ levels where $\pi$ has $P_0$ and possibly at some of its $e$ star levels, and hence at most $f\le m-t$ times by~(i), so a box only contains leaves of order at least $t$. By~(ii), a box is obtained from a leaf of order at least $t$ by replacing its $e$ lowest symbols $P_0$ (rather than an arbitrary subset of them) by stars. Therefore, each leaf yields at most $m+1$ boxes, one for each value of $e$, which keeps their total number small, as we will show in Lemma~\ref{lem:allboxes}. Boxes of this shape suffice: we will prove in Lemmas~\ref{lem:partition} and~\ref{lem:boxes} that the leaves of order at least $t$ of an output string can be split among $\alpha_t$ of them. (We will see below that the boxes of an output string all have exactly $m-t$ symbols that are $P_0$ or stars. We nonetheless allow fewer in the definition of a box, since boxes with fewer such symbols will arise as intermediate values in the dynamic program of Lemma~\ref{lem:allboxes}.)

It is important here that a star stands for all ten terms, $P_0$ included, so that one box collects leaves of several orders (from $m-f$ to $m-f+e$, for a box with $e$ stars and $f-e$ symbols $P_0$). This is needed for fast queries. If a star stood only for the nine terms $P_{ij}$, then every box would fix the set of levels at which its leaves choose $P_0$, so an output string would need a separate box for each subset of $Q$ of size at most $m-t$. That would be about $2^m=\sqrt D$ boxes when $t$ is small, so a query would take about $\sqrt D$ time, which gives no saving below $N^2$ when there are $N^2/\sqrt D$ queries.

\paragraph{The boxes of an output string.}
We now make precise the informal rule above for finding the box of a leaf, and describe all the boxes of an output string $w$ with inner set $Q$. Let $\tau$ be a leaf of order at least $t$ contributing to $w$, whose box we found above by going through the levels of $Q$ from the highest down (Figure~\ref{fig:boxes}). In symbols, let
\begin{alignat*}{2}
Z&:=\{\ell\in Q:\ \tau_\ell=P_0\}&\qquad&\text{(the levels at which $\tau$ chooses $P_0$),}\\
V&:=Z\cup\{\text{the $m-t-|Z|$ lowest levels of }Q\setminus Z\}&&\text{($Z$ padded from the bottom),}\\
F_V&:=\{\ell\in Q:\ \ell<\min(Q\setminus V)\}&&\text{(the levels that we have not reached).}
\end{alignat*}
Since the order of $\tau$ is $m-|Z|\ge t$, the set $V$ has exactly $m-t$ levels, namely, all the levels of $Q$ except the $t$ at which we saw a term other than $P_0$. We define $F_V$ by the same formula for every $V\subseteq Q$, with $F_V:=Q$ if $V=Q$. It is the longest initial segment of $Q$ contained in $V$. ($Z$ stands for \emph{zero}, since these are the levels where $\tau$ chooses $P_0$, and $F$ stands for \emph{free}, since these are the levels where the box has a star.)

For $V\subseteq Q$ with $|V|=m-t$, let $\mathcal B_V$ be the set of the $9^t$ cubes $\pi$ with
\[
\pi_\ell=\begin{cases}
\ast&\text{if }\ell\in F_V,\\
P_0&\text{if }\ell\in V\setminus F_V,\\
\text{one of the nine terms }P_{ij}&\text{if }\ell\in Q\setminus V,\\
\text{the term }P_{ij}\text{ with }w_\ell=z_{ij}&\text{if }\ell\notin Q.
\end{cases}
\]
Every cube in $\mathcal B_V$ is a box: it satisfies condition~(i) of the definition because it has $|V|=m-t$ symbols $P_0$ or stars, and condition~(ii) because its stars are below its symbols $P_0$, since $F_V$ is an initial segment of $Q$. When $V$ is the set defined above from a leaf $\tau$, the box of $\tau$ is the one in $\mathcal B_V$ with the terms of $\tau$ at the levels of $Q\setminus V$. Since a box of $\mathcal B_V$ has stars at the levels of $F_V$, it contains the leaves with any terms there, of which there are $10^{|F_V|}$. Notice that the set $\mathcal B_V$ depends on $w$ (since it depends on $Q$ and on the terms at the levels outside $Q$), but we omit this from the notation since $w$ will always be clear from context. We call the boxes in these sets, over all $V\subseteq Q$ with $|V|=m-t$, the \defn{boxes of $w$}.

There is another way to describe the boxes of $w$, in terms of the leaves of order exactly $t$ contributing to $w$. For such a leaf, consider the lowest level of $Q$ at which it chooses a term other than $P_0$, and replace its $P_0$ by a star at every lower level of $Q$ (or at every level of $Q$, if $t=0$). The result is a box of $w$, and each box of $w$ arises exactly once in this way. This means that $w$ has exactly $\alpha_t$ boxes.

\paragraph{Every leaf of order at least $t$ lies in exactly one box.}
A leaf contributing to $w$ that chooses $P_0$ at the set $Z$ of levels is a leaf of a box of $\mathcal B_V$ if and only if it chooses $P_0$ at every level of $V\setminus F_V$ and at no level of $Q\setminus V$, that is,
\[
V\setminus F_V\;\subseteq\;Z\;\subseteq\;V.
\]
Furthermore, it is a leaf of exactly one of these boxes, namely, the one with its terms at the levels of $Q\setminus V$. We will see in the following lemma that, for $|Z|\le m-t$, exactly one $V$ satisfies this condition, namely $Z$ padded from the bottom. Thus every leaf of order at least $t$ lands in exactly one box.

\begin{lemma}\label{lem:partition}
Let $Q$ be a set of $m$ levels, and let $0\le t\le m$. For every $Z\subseteq Q$ with $|Z|\le m-t$, there is exactly one set $V\subseteq Q$ with $|V|=m-t$ and
\[
V\setminus F_V\;\subseteq\;Z\;\subseteq\;V,
\]
namely $Z$ together with the $m-t-|Z|$ lowest levels of $Q\setminus Z$.
\end{lemma}

\begin{proof}
Let $V\subseteq Q$ with $|V|=m-t$ and $Z\subseteq V$. We show that $V\setminus F_V\subseteq Z$ holds if and only if $V\setminus Z$ consists of the $|V\setminus Z|=m-t-|Z|$ lowest levels of $Q\setminus Z$, which determines $V$. If $V=Q$, then both conditions hold: $F_V=Q$, so $V\setminus F_V=\emptyset\subseteq Z$, and $V\setminus Z$ is all of $Q\setminus Z$. Otherwise, let $\ell$ be the lowest level of $Q\setminus V$. Then $F_V$ consists of the levels of $Q$ below $\ell$, so $V\setminus F_V\subseteq Z$ says that every level of $V\setminus Z$ is below $\ell$. Since $Q\setminus Z$ is the disjoint union of $V\setminus Z$ and $Q\setminus V$, and every level of $Q\setminus V$ is at least $\ell$, this holds if and only if $V\setminus Z$ consists of the lowest levels of $Q\setminus Z$.
\end{proof}

\begin{lemma}\label{lem:boxes}
Let $w$ be an output string, with inner set $Q$. The leaves of order at least $t$ contributing to $w$ are exactly the leaves of the boxes in $\bigcup_V\mathcal B_V$, the union over all $V\subseteq Q$ with $|V|=m-t$, and each of them is a leaf of exactly one of these boxes. Hence
\[
(X_QY_Q)[w]\;=\sum_{\substack{\tau\text{ contributing to }w\\\text{of order }<\,t}}\Phi_\tau(a)\,\Psi_\tau(b)\;+\;\sum_{\substack{V\subseteq Q\\|V|=m-t}}\;\sum_{\pi\in\mathcal B_V}\val(\pi),
\]
a sum of $\sum_{d=0}^{t-1}\alpha_d$ products and $\alpha_{t}$ values of boxes.
\end{lemma}

\begin{proof}
Every leaf of a box of $\mathcal B_V$ agrees with the private leaf of $w$ outside $Q$, so it contributes to $w$, and it chooses $P_0$ at most $|V|=m-t$ times, so its order is at least $t$. Conversely, a leaf $\tau$ contributing to $w$, of order at least $t$, chooses $P_0$ at a set $Z\subseteq Q$ of at most $m-t$ levels, and as we saw above, it is a leaf of a box of $\mathcal B_V$ if and only if $V\setminus F_V\subseteq Z\subseteq V$, and then of exactly one such box. By Lemma~\ref{lem:partition}, exactly one $V$ satisfies this condition. The desired equation follows, since $(X_QY_Q)[w]$ is the sum of the products at the leaves contributing to $w$. Finally, $\alpha_d$ leaves of order $d$ contribute to $w$, and there are $\binom m{m-t}$ sets $V$ with $9^t$ boxes each, i.e., $\binom m{m-t}9^t=\binom mt9^t=\alpha_t$ boxes in all (as many as the leaves of order exactly $t$ contributing to $w$, although they cover its leaves of every order from $t$ to $m$).
\end{proof}

\subsection{The data structure, in terms of the parameters}\label{sec:single}

In this subsection, we assemble the boxes into our data structure and bound its costs in terms of the parameters $m$, $L$, and $t$. We will then choose settings for the parameters in Section~\ref{sec:corproofs}. We begin by recalling two notions from Section~\ref{sec:result} that the data structure builds on.

The \emph{tiling} of Section~\ref{sec:tiling} cuts the product $XY$ into \emph{tiles}. A row block is $\Nz$ consecutive rows of $X$, a column block is $\Nz$ consecutive columns of $Y$, and a tile is a band of $K_0=\lfloor\sqrt K\rfloor$ consecutive row blocks together with a band of $K_0$ consecutive column blocks. One run of the recursion computes all $K_0^2$ block products of a tile at once, each as the product $X_QY_Q$ for its own subset $Q$. The \emph{encodings} of Section~\ref{sec:encoders} are the arrays of the $10^L$ numbers $\Phi_\tau(a)$, one for each leaf $\tau$, computed from the input array $a$ of a row band, and similarly the numbers $\Psi_\tau(b)$ computed from the input array $b$ of a column band. We compute an encoding once per band and share it among all the tiles that use that band, and the product at a leaf of a tile is the product of its two encoded numbers.

Our data structure consists of three components:
\begin{itemize}[nosep]
\item the list of the $K_0^2$ subsets $Q$ assigned to the block products of a tile (recall that this list is the same for every tile),
\item the encodings of all row bands and all column bands, and
\item for every tile, the values of all its boxes (the boxes are the same strings in every tile, but their values depend on the input arrays of the tile).
\end{itemize}
To answer a query for an entry $(XY)[I,J]$, the data structure performs the following three steps:
\begin{enumerate}[label=(\arabic*),nosep]
\item Find the tile that contains $(I,J)$, the subset $Q$ of its block product, and the output string $w$ of $(I,J)$.
\item Add up the products at the leaves of order below $t$ contributing to $w$, by reading the two factors of each product from the encodings of the tile.
\item Add to this the values of the $\alpha_t$ boxes of $w$, which are stored for the tile.
\end{enumerate}
By Lemma~\ref{lem:boxes}, the result is $(X_QY_Q)[w]=(XY)[I,J]$. We will give the details of the query, as well as the preprocessing, in the proof of Theorem~\ref{thm:singletech} below.

The encodings are indexed by the leaves (Section~\ref{sec:encoders}), and for each tile we store its boxes, with their values, in a standard trie on their strings of $L$ symbols. Thus reading the product at a leaf, or looking up or inserting a box, takes $O(L)$ operations.

We will next show that there are at most $m+1$ times as many boxes as there are leaves of order at least $t$, and that the dynamic program described in the technique overview computes their values (a box with $e$ stars being the sum of ten boxes with $e-1$ stars).

\begin{lemma}\label{lem:allboxes}
There are at most $(m+1)\sum_{d=t}^{m}\beta_d$ boxes. Given the two encodings of a tile, we can compute the values of all these boxes, and store them in the trie for that tile, in $O(L)$ time and space per box.
\end{lemma}

\begin{proof}
\emph{The count.} A box in which $f$ of the symbols are $P_0$ or stars is obtained from a leaf of order $m-f$ (namely the leaf we get by replacing its stars with $P_0$) by turning the $e$ lowest symbols $P_0$ of that leaf into stars, for some $e\le f$. Thus, such a box can be specified by choosing the $f$ levels of these symbols, the number $e$, and a term other than $P_0$ at each of the other $L-f$ levels. There are hence $(f+1)\binom Lf9^{L-f}=(f+1)\,\beta_{m-f}$ such boxes, and summing over $f\le m-t$ gives the upper bound on the number of boxes.

\emph{The values.} We compute the values of the boxes in increasing order of their number of stars. The boxes without stars are the leaves with at most $m-t$ symbols $P_0$, and the value of each of them is the product of its two numbers in the encodings.

For a box $\pi$ with $e\ge1$ stars, let $\ell$ be the highest level at which $\pi$ has a star. The leaves of $\pi$ are the leaves of the ten strings $\pi[\ell\leftarrow\lambda]$ obtained by replacing that star by a term $\lambda$, and each of these strings is again a box, with $e-1$ stars. Indeed, if $\lambda=P_0$, then the number of symbols that are $P_0$ or stars stays the same, and the new $P_0$ at level $\ell$ is above all the remaining stars. If $\lambda$ is any other term, then that number decreases by one, and every star is still below every $P_0$. Thus conditions (i) and (ii) in the definition of a box hold in both cases. (We expand the highest star because replacing a lower star by $P_0$ would leave a star above a $P_0$, which is not a box.) Hence we compute
\[
\val(\pi)=\sum_\lambda\val(\pi[\ell\leftarrow\lambda])
\]
with ten lookups in the trie, in $O(L)$ operations. Generating the boxes with $e$ stars and inserting them into the trie also takes $O(L)$ operations per box, and adds at most $L$ vertices per box.
\end{proof}

See Figure~\ref{fig:count} for an illustration of these counts (in that figure, the switching order $t$ is set to $m/9$). At every order $d$, the pruned recursion of Section~\ref{sec:sparse} visits at most $\min\{|U|\alpha_d,\beta_d\}$ leaves of order $d$ in a tile whose wanted output strings form the set $U$ (see the proof of Lemma~\ref{lem:count}), the smaller of the two bounds at that order. The data structure cannot do this, since $U$ is not known in advance. Instead, a query computes $\alpha_d$ products at every order $d<t$, and looks up $\alpha_t$ boxes at the order $t$. The data structure also stores at most $(m+1)\sum_{d\ge t}\beta_d$ boxes per tile, even if they are not used by a query. The preprocessing time is therefore dominated by the sum $\sum_{d\ge t}\beta_d$, and hence by how fast the $\beta_d$ decay. Let
\[
\rho:=\frac{9m}{L-m+1},
\]
which is less than $1$ since $L\ge10m$. It plays the role of $\frac12$ in~\eqref{eq:decay}: for $1\le d\le m$,
\begin{equation}\label{eq:gendecay}
\frac{\beta_d}{\beta_{d-1}}=\frac{9(m-d+1)}{L-m+d}\le\frac{9m}{L-m+1}=\rho,\qquad\text{so}\qquad\beta_d\le\rho^d\,\Cm .
\end{equation}
The closer the ratio $c=L/m$ is to $10$, the closer $\rho$ is to $1$, and the more boxes the preprocessing has to compute. This is one side of the trade-off mentioned at the beginning of this section between the time savings and how thin the product must be.

We can now bound the costs. The theorem below assumes that $N\ge\sqrt K\,\Nz$. Since a tile spans $K_0\Nz\le\sqrt K\,\Nz$ rows and as many columns of the $N\times N$ product $XY$, this ensures that a tile fits inside $XY$, so padding $N$ to a multiple of $K_0\Nz$ at most doubles it. We checked the same requirement in Section~\ref{sec:proof}, where it was the first of the two consequences of $N\ge D^{18}$. Here it will also follow, in Section~\ref{sec:corproofs}, from the requirement that computing the encodings takes at most $N^2/D^{\gamma}$ time, since $10^L\ge\Cm$.

\begin{theorem}\label{thm:singletech}
Let $m\ge1$, $D=4^m$, $L\ge10m$, $0\le t\le m$, and $N\ge\sqrt K\,\Nz$. Given as input matrices $X\in\Z^{N\times D}$ and $Y\in\Z^{D\times N}$, whose entries are integers of absolute value at most $N^{O(1)}$, we can preprocess them deterministically in time and space
\begin{equation}\label{eq:prep}
O\Big(L\,m\,\frac{\rho^t}{1-\rho}\,N^2\;+\;\frac{N\cdot10^L}{\sqrt K\,\Nz}\Big).
\end{equation}
After this, any single entry $(XY)[I,J]$ can be computed deterministically in $O\big(L\sum_{d=0}^{t}\alpha_d\big)$ time. In particular, for every set $W$ of positions of an $N\times N$ matrix, the entries $(XY)[I,J]$, $(I,J)\in W$, can be computed deterministically in time
\begin{equation}\label{eq:wanted}
O\Big(\underbrace{L\,|W|\sum_{d=0}^{t}\alpha_d}_{\text{\normalfont queries}}\;+\;\underbrace{L\,m\,\frac{\rho^t}{1-\rho}\,N^2}_{\text{\normalfont boxes}}\;+\;\underbrace{\frac{N\cdot10^L}{\sqrt K\,\Nz}}_{\text{\normalfont encodings}}\Big).
\end{equation}
The constants hidden in the $O(\cdot)$ depend only on the exponent in $N^{O(1)}$.
\end{theorem}

\begin{proof}
\emph{Preprocessing.} We tile the product as in Section~\ref{sec:tiling}, padding $N$ to a multiple of $K_0\Nz$, which at most doubles it since $N\ge\sqrt K\,\Nz\ge K_0\Nz$, and we store the $K_0^2$ subsets, in $O(KL)\le O(10^L)$ operations. We then form and encode the input arrays of all row bands and column bands, at most $4N/(\sqrt K\,\Nz)$ of them since $K_0\ge\sqrt K/2$. Each has at most $K\Nz D\le7^L$ nonzero entries, so we form it in $O(L\cdot7^L)$ operations, and its encoding takes $O(10^L)$ operations (Section~\ref{sec:encoders}). This is the last term of~\eqref{eq:prep}. Finally, we compute the values of all the boxes of each of the at most $4N^2/\Cm$ tiles, where $\Cm=K\Nz^2$ is the number of output entries of a tile. By Lemma~\ref{lem:allboxes}, this takes $O(L(m+1)\sum_{d\ge t}\beta_d)$ time and space per tile, and $\sum_{d\ge t}\beta_d\le\Cm\rho^t/(1-\rho)$ by~\eqref{eq:gendecay}, which gives the first term.

\emph{Query.} Given $(I,J)$, we find its tile from the bands of row $I$ and column $J$, the subset $Q$ of its block product, and its output string $w$ (whose variables at the levels outside $Q$ we read off the row and column of $(I,J)$ within the block product), all in $O(L)$ operations (Section~\ref{sec:proof}). We then compute $(X_QY_Q)[w]$ as the sum in Lemma~\ref{lem:boxes}, which has two parts.

The first part is the sum of the products at the leaves of order below $t$ contributing to $w$. Each of these leaves is obtained from the private leaf of $w$, which has $P_0$ at every level of $Q$, by picking fewer than $t$ of those levels and replacing $P_0$ with one of the nine other terms at each of them. We enumerate them, and for each such leaf $\tau$ we look up its two numbers $\Phi_\tau(a)$ and $\Psi_\tau(b)$ in the encodings of the tile and multiply them. The second part is the sum of the values of the boxes of $w$: for every $V\subseteq Q$ with $|V|=m-t$ and every box of $\mathcal B_V$, we look up its value in the trie of the tile. In all, these are $\sum_{d=0}^{t-1}\alpha_d+\alpha_t$ numbers, and we find each of them in $O(L)$ operations, including forming its leaf or its box from the private leaf. Enumerating them also takes $O(L)$ operations per number. By Section~\ref{sec:proof}, the sum is $(XY)[I,J]$. For a set $W$, we ask $|W|$ queries, which gives~\eqref{eq:wanted}.

\emph{Word size.} As in Section~\ref{sec:proof}, every number in an encoding is a sum of entries of the input array with coefficients $0,\pm1$, and every value we compute is a sum of at most $10^m$ products of two such numbers. Since $N\ge\Nz=3^{L-m}$ and $L\ge10m$, we have $10^L\le N^{5/2}$, so, as in Section~\ref{sec:proof}, all the values are $N^{O(1)}$ in absolute value, and every integer we use has $O(\log N)$ bits.
\end{proof}

\subsection{Choosing the parameters}\label{sec:corproofs}

We now choose the parameters in Theorem~\ref{thm:singletech} and prove the results of Section~\ref{sec:results}. In other words, we do the arithmetic and bookkeeping to derive the explicit exponents from Theorem~\ref{thm:singletech}. We first do the calculation in full for a single choice of the parameters, $L=21m$ and $t=\lceil m/9\rceil$, which proves Corollary~\ref{cor:zt}. We then say what changes for other choices, which gives Theorems~\ref{thm:single} and~\ref{thm:general} and the entries of Table~\ref{tab:ds}.

\begin{proof}[Proof of Corollary~\ref{cor:zt}]
\emph{Setting up.} Let $m:=\lceil\log_4D\rceil$, and pad the inner dimension to $4^m<4D$ with zero columns of $X$ and zero rows of $Y$. This changes no entry of $XY$, and it changes $D$ by a factor less than $4$, which only affects the constants. So from now on $D=4^m$, and since the original $D$ was larger than $4^{m-1}$, the assumption $N\ge D^{18}$ now reads $N\ge4^{18(m-1)}$. Let $L:=21m$ and $t:=\lceil m/9\rceil$. We will apply Theorem~\ref{thm:singletech} with these parameters, and we will verify its condition $N\ge\sqrt K\,\Nz$ below, in the step on the encodings.

\emph{Boxes.} We have $\rho=9m/(L-m+1)=9m/(20m+1)<9/20$, so $1/(1-\rho)<2$, and since $t\ge m/9$, also $\rho^t\le(9/20)^{m/9}=D^{-\gamma}$, where $\gamma:=\ln(20/9)/(9\ln4)=0.0640\ldots$. Hence the first term of~\eqref{eq:prep} is $O(m^2N^2/D^{\gamma})$.

\emph{Queries.} A query reads $\sum_{d\le t}\alpha_d$ numbers, and we bound this sum as in the proof of Lemma~\ref{lem:count}. Since $9^d=72^d8^{-d}\le72^t8^{-d}$ for $d\le t$, and $t<m/9+1$,
\[
\sum_{d=0}^{t}\binom md9^d\;\le\;72^t\sum_{d=0}^{m}\binom md8^{-d}=72^t\Big(\frac98\Big)^m<72\cdot\Big(72\cdot\big(\tfrac98\big)^9\Big)^{m/9}=72\,D^{q},
\]
where $q:=\ln(72\cdot(9/8)^9)/(9\ln4)=0.4277\ldots$. Hence a query takes $O(L\,D^{q})=O(m\,D^{q})$ time.

\emph{Encodings.} It remains to show that the last term of~\eqref{eq:prep} is at most $N^2/D^{\gamma}$, and that $N\ge\sqrt K\,\Nz$. Both follow from
\begin{equation}\label{eq:encrange}
\frac{D^{\gamma}\cdot10^L}{\sqrt K\,\Nz}\;\le\;N ,
\end{equation}
the first by rearranging, and the second because $10^L\ge\Cm=K\Nz^2$ and $D^{\gamma}\ge1$. As in Section~\ref{sec:proof}, we compare $m$-th powers by comparing their bases. We have $D^{\gamma}=(20/9)^{m/9}$, $10^L=10^{21m}$, and $\Nz=3^{20m}$, and the standard bound $\binom nk\ge\frac1{n+1}\cdot\frac{n^n}{k^k(n-k)^{n-k}}$ gives $K=\binom{21m}m\ge\frac1{21m+1}\big(21^{21}/20^{20}\big)^m$. Hence the left-hand side of~\eqref{eq:encrange} is at most $\sqrt{21m+1}\cdot\Lambda^m$, where
\[
\Lambda:=\frac{10^{21}\cdot(20/9)^{1/9}}{(21^{21}/20^{20})^{1/2}\cdot3^{20}}=4.198\ldots\cdot10^{10} .
\]
We now use the assumption $N\ge D^{18}$. Its base $4^{18}=6.871\ldots\cdot10^{10}$ is larger than $\Lambda$ by a factor of more than $1.63$. Since $N\ge4^{18(m-1)}$, the inequality~\eqref{eq:encrange} therefore holds whenever $1.63^m\ge4^{18}\sqrt{21m+1}$, which is the case for all $m\ge60$. (For $m<60$, $D$ is bounded by a constant, and the corollary holds trivially.)

\emph{Conclusion.} For $m\ge60$, Theorem~\ref{thm:singletech} thus applies. The preprocessing takes $O(m^2N^2/D^{\gamma})$ time and space, and a query takes $O(m\,D^{q})$ time. Since $m=O(\log D)$ and the exponents $\gamma-0.063$ and $0.437-q$ are positive, we have $m^2=O(D^{\gamma-0.063})$ and $m=O(D^{0.437-q})$, which gives the bounds $O(N^2/D^{0.063})$ and $O(D^{0.437})$ of the corollary. The bound for a set $W$ follows by asking $|W|$ queries.
\end{proof}

\paragraph{Other choices of the parameters.}
The numbers $21$ and $\frac19$ did not play a special role in this proof, and the same calculation works with $L=\lceil cm\rceil$ and $t=\lceil\theta m\rceil$ for any constants $c>10$ and $0<\theta<0.9$. This gives Corollary~\ref{cor:tradeoff} below, which we use for the left half of Table~\ref{tab:ds}, with the $c$ of the row and the $\theta$ that gives the $q$ of the column ($\theta=\frac19$ for $q=0.43$). To state it, let $\Hh(x):=-x\ln x-(1-x)\ln(1-x)$ be the entropy function (with natural logarithms), let $\Lambda:=10^{c}\,e^{-\frac12c\Hh(1/c)}\,3^{-(c-1)}\,4^{\gamma}$ be the general form of the base $\Lambda$ above, and let
\begin{equation}\label{eq:Rc}
R_c(\gamma):=\frac{\ln4}{\ln\Lambda}=\frac{\ln4}{c\ln10-\tfrac12c\Hh(1/c)-(c-1)\ln3+\gamma\ln4} .
\end{equation}

\begin{corollary}\label{cor:tradeoff}
Let $c>10$ and $0<\theta<0.9$, let $\rho_c:=9/(c-1)$, and let
\[
\gamma:=\frac{\theta\ln(1/\rho_c)}{\ln4}>0\qquad\text{and}\qquad q:=\frac{\Hh(\theta)+\theta\ln9}{\ln4} .
\]
Given as input matrices $X\in\Z^{N\times D}$ and $Y\in\Z^{D\times N}$ whose entries are integers of absolute value at most $N^{O(1)}$, where $2\le D\le N^{\varepsilon}$ and $\varepsilon<R_c(\gamma)$, with $R_c$ as in~\eqref{eq:Rc}, we can preprocess them deterministically in $O(N^2\log^2\!D/D^{\gamma})$ time and space, after which any single entry $(XY)[I,J]$ can be computed deterministically in $O(D^{q}\log D)$ time. The constants hidden in the $O(\cdot)$ depend on $c$, $\theta$, and $\varepsilon$.

Moreover, $R_c(0)$ increases to $\varepsilon^*=\ln4/(5\ln10)=0.1204\ldots$ as $c$ decreases to $10$.
\end{corollary}

\begin{proof}
We repeat the proof of Corollary~\ref{cor:zt} with $L:=\lceil cm\rceil$ and $t:=\lceil\theta m\rceil$. After the padding, the assumption $D\le N^{\varepsilon}$ now reads $N\ge4^{(m-1)/\varepsilon}$. For the boxes, $\rho<\rho_c$ gives $\rho^t\le\rho_c^{\theta m}=D^{-\gamma}$. For the queries, the same calculation with $x:=(1-\theta)/\theta$ in place of $8$, where $9x>1$ because $\theta<0.9$, gives $\sum_{d\le t}\alpha_d\le(9x)^t(1+1/x)^m=O(D^{q})$. For the encodings, the same bound on the binomial coefficient, $\binom nk\ge e^{n\Hh(k/n)}/(n+1)$, shows that the left-hand side of~\eqref{eq:encrange} is $O(\sqrt m\cdot\Lambda^m)$. As $\varepsilon<R_c(\gamma)$ says that $\Lambda<4^{1/\varepsilon}$, the inequality~\eqref{eq:encrange} thus holds once $m$ exceeds a constant depending on $c$, $\theta$, and $\varepsilon$, and for smaller $m$ the corollary again holds trivially.

Finally, a direct calculation shows that for $\gamma=0$, $\ln\Lambda$ increases with $c$ and tends to $5\ln10$ as $c$ decreases to $10$, by the identity $\Hh(\frac1{10})+\frac9{10}\ln9=\ln10$. Hence $R_c(0)$ increases to $\ln4/(5\ln10)=\varepsilon^*$.
\end{proof}

\paragraph{Where $\varepsilon^*$ comes from.}
The identity $\Hh(\frac1{10})+\frac9{10}\ln9=\ln10$ says that $\binom{10m}m9^{9m}$, the largest term of the binomial expansion of $(1+9)^{10m}$, is $10^{10m}$ up to polynomial factors. So at $c=10$ a tile has $\Cm\approx10^L$ output entries, as many as the recursion has leaves, and the encodings cost $N\cdot10^L/\sqrt{\Cm}\approx N\cdot10^{L/2}$, which is below $N^2$ exactly when $N\ge10^{L/2}=D^{1/\varepsilon^*}$.

\begin{proof}[Proof of Theorem~\ref{thm:single}]
Let $\varepsilon<\varepsilon^*$ and $q>0$. Choose $c>10$ with $R_c(0)>\varepsilon$ (Corollary~\ref{cor:tradeoff}), and then $\theta\in(0,0.9)$ small enough that $(\Hh(\theta)+\theta\ln9)/\ln4<q$ and $R_c(\theta\ln(1/\rho_c)/\ln4)>\varepsilon$. This is possible since, as $\theta\to0$, the first expression tends to $0$ and the second to $R_c(0)$. Corollary~\ref{cor:tradeoff} with these $c$ and $\theta$ gives the theorem.
\end{proof}

\begin{proof}[Proof of Theorem~\ref{thm:general}]
Apply Theorem~\ref{thm:single} with $q:=\kappa/2$, and ask one query for each position of $W$. As $|W|\le N^2/D^{\kappa}$, this takes $O(N^2\log^2\!D/D^{\gamma}+|W|\,D^{\kappa/2}\log D)=O(N^2\log^2\!D/D^{\gamma'})$ time, where $\gamma':=\min\{\gamma,\kappa/2\}$.
\end{proof}

With Corollary~\ref{cor:tradeoff} in place of Theorem~\ref{thm:single}, the same argument gives explicit exponents, which we use for the right half of Table~\ref{tab:ds}, with the $\theta$ at which $\gamma=\kappa-q$. The $\varepsilon$ of a row of the table is the smallest $R_c(\gamma)$ over its entries.

\begin{corollary}\label{cor:range}
Let $c$, $\theta$, $\gamma$, $q$, $X$, $Y$, $D$, and $\varepsilon$ be as in Corollary~\ref{cor:tradeoff}, and let $\kappa>0$. For every set $W$ of at most $N^2/D^{\kappa}$ positions of an $N\times N$ matrix, the entries $(XY)[I,J]$, $(I,J)\in W$, can be computed deterministically in time
\[
O\big(N^2\log^2\!D\,(D^{-\gamma}+D^{q-\kappa})\big),
\]
a saving of $D^{\min\{\gamma,\,\kappa-q\}}$, up to logarithmic factors, over the size $N^2$ of the product.
\end{corollary}

%% file: 05-other.tex
\section{Further reductions}\label{sec:other}

{This section presents the remaining algorithmic results of the introduction: directed unweighted APSP (Section~\ref{sec:dirapsp}), the real-valued versions of $3$SUM, APSP, and Exact Triangle (Section~\ref{sec:real}), the weighted $k$-Clique problems (Section~\ref{sec:kclique}), and the refutation of the rectangular hinted OMv conjectures of van den Brand, Nanongkai, and Saranurak~\cite{BNS19} (Section~\ref{sec:bns}). Each result but the last composes a known reduction with the algorithms of Sections~\ref{sec:reduction} and~\ref{sec:general}; the last applies the data structure of Section~\ref{sec:general} directly. We cite the reduction we use in each case and modify it only where our parameter regime requires it, saying what changes.}

\paragraph{What we use.}
We recall the three results of Sections~\ref{sec:reduction} and~\ref{sec:general} that the reductions plug into.
\begin{itemize}
\item \emph{Exact Triangle} (Section~\ref{sec:redproof}) asks whether a complete tripartite graph with parts $A,B,C$ of $n$ vertices and integer edge weights of absolute value at most $n^\nu$ has a triangle of weight zero. Theorem~\ref{thm:zwt} solves it deterministically in $O(n^{3-\varepsilon'}\log n)\leq O(n^{3-\varepsilon_T})$ time, where $\varepsilon'=0.00175$ and $\varepsilon_T=0.0017$, and Theorem~\ref{thm:3sumapsp} derives from this a deterministic $\tilde{O}(n^{3-\varepsilon'/3})$-time algorithm for the $(\min,+)$-product of two $n\times n$ matrices with integer entries of absolute value $n^{O(1)}$.
\item $\Lop(n,D)$ and $\#\Lop(n,D)$ (Definitions~\ref{def:lop} and~\ref{def:countlop}) are given a tripartite graph with parts $A$ and $B$ of $n$ vertices and a middle part of at most $D$ vertices, together with a set $W$ of \emph{query pairs} $(a,b)\in A\times B$, and ask for every query pair in $W$ whether $a$ and $b$ have a common neighbor in the middle part, or for the counting version, how many they have. Corollary~\ref{fact:lop} solves both deterministically in $O(|W|D^{0.437}+n^2/D^{0.063})$ time when $n\ge D^{18}$. Theorem~\ref{thm:red} reduces Exact Triangle to $\Lop$ deterministically by hashing the weights modulo a prime.
\item Corollary~\ref{cor:zt} is the data structure behind Corollary~\ref{fact:lop}. Given $X\in\Z^{N\times D}$ and $Y\in\Z^{D\times N}$ with $N\ge D^{18}$ and entries of absolute value $N^{O(1)}$, it preprocesses them deterministically in $O(N^2/D^{0.063})$ time, after which it computes any entry of $XY$ in $O(D^{0.437})$ time. Section~\ref{sec:real} applies it to matrices whose entries are not $0/1$, and Section~\ref{sec:bns} uses its queries online.
\end{itemize}
{The algorithms of Sections~\ref{sec:dirapsp}, \ref{sec:kclique}, and~\ref{sec:bns} are deterministic. Those of Section~\ref{sec:real} are randomized (Las Vegas), because the reductions of~\cite{CVX22} that they use as black boxes are randomized.}

\subsection{Directed APSP with small integer weights}\label{sec:dirapsp}

Let $\omega(1,r,1)$ denote the exponent of multiplying an $n\times n^{r}$ matrix by an $n^{r}\times n$ matrix, and let $\mu$ be the solution of $\omega(1,\mu,1)=1+2\mu$. Then $\mu\ge\frac12$, since $\omega(1,\mu,1)\ge2$, and $\mu<0.5275$ by the bounds of~\cite{ADVXXZ25}. Zwick's algorithm~\cite{Zwi02} solves APSP in directed unweighted graphs in $\tilde O(n^{2+\mu})$ time, and the Directed Unweighted APSP hypothesis~\cite{CVX21,CVX23,Fis26} asserts that no algorithm is polynomially faster. We refute this in Theorem~\ref{thm:dirapsp}.

We revisit Theorem~3.1 of Chan, Vassilevska Williams, and Xu~\cite{CVX21}, which is just a presentation of Zwick's algorithm~\cite{Zwi02} as a reduction. Zwick's algorithm is typically used in its randomized form because the proof is simple, so the reduction in \cite{CVX21} was stated in its randomized form. However, following an argument by Yuster and Zwick~\cite[Section~8]{YZ05b}, the reduction has a deterministic version as well:

\begin{theorem}[Zwick's algorithm as a deterministic reduction~\cite{Zwi02,CVX21}]\label{thm:detred}
If the $(\min,+)$-product of an $n\times n^{\mu}$ matrix by an $n^{\mu}\times n$ matrix, both with integer entries of absolute value at most $\tilde{O}(n^{1-\mu})$, can be computed in $O(n^{2+\mu-\varepsilon_1})$ time for a constant $\varepsilon_1>0$, then APSP in directed $n$-vertex graphs with integer weights of absolute value $\polylog n$ and no negative cycles can be solved in $O(n^{2+\mu-\varepsilon''})$ time for a constant $\varepsilon''>0$ that depends on $\varepsilon_1$ and $\mu$. If the algorithm for the product is deterministic, then so is the algorithm for APSP.
\end{theorem}

The randomized proof of~\cite{CVX21} follows the randomized version of Zwick's algorithm, which takes random samples of vertices to hit the long shortest paths, so the algorithm that it produces is randomized. {Zwick~\cite[Sections~6 and~7]{Zwi02} also gives a deterministic version of his algorithm, which constructs the hitting sets (his bridging sets) from witnesses of the products by a greedy hitting-set computation, within the same running time up to polylogarithmic factors, and Yuster and Zwick~\cite[Section~8]{YZ05b} announce, without details, a faster deterministic construction of bridging sets for their distance oracle. Zwick's construction lists a path with $s$ vertices for each of the $n^2$ pairs, which is more than the reduction can afford when $s$ is large. We therefore give a deterministic version of the reduction, which lists only walks with an endpoint in the previous hitting set, in the spirit of~\cite{YZ05b}. We believe that this is likely what Yuster and Zwick~\cite{YZ05b} had in mind.}

\begin{proof}[Proof of Theorem~\ref{thm:detred}]
The deterministic algorithm follows the randomized reduction, with two changes. It constructs the hitting sets deterministically, by a hierarchy of bridging sets in the manner of Zwick~\cite[Section~6]{Zwi02}, and it computes the distances that this construction needs only to and from small sets of vertices, like the distance oracle of Yuster and Zwick~\cite{YZ05b}. The hypothesized algorithm for the $(\min,+)$-product is used only for the values of its products. (This is convenient but not essential: the witness-finding method of Seidel~\cite{Sei95} and Alon, Galil, Margalit, and Naor~\cite{AGMN92,GM93}, derandomized by Alon and Naor~\cite{AN96} and applied to $(\min,+)$-products by Zwick~\cite[Lemma~3.3]{Zwi02}, uses the product algorithm only on submatrices, so any deterministic algorithm for the product can be made to return witnesses at a polylogarithmic cost.) The price is a factor $n^{\gamma}$ in the running time, for an arbitrarily small constant $\gamma>0$, which the polynomial saving of the hypothesis absorbs. Throughout, $M$ is the largest absolute weight, which is polylogarithmic, and $\tilde O$ hides polylogarithmic factors.

\emph{Bridges.} For vertices $u$ and $v$, let $\delta(u,v)$ be the distance from $u$ to $v$, and let $\eta(u,v)$ be the smallest number of edges on a shortest path from $u$ to $v$. Every subpath of a shortest path with the fewest edges is again a shortest path with the fewest edges. For a threshold $t$ and a constant $C$, a set $B\subseteq V$ is a \emph{$(t,C)$-bridge} if every pair $(u,v)$ with $\eta(u,v)\ge t$ has a shortest \emph{walk} from $u$ to $v$ that passes through a vertex of $B$ and has at most $C\eta(u,v)$ edges. This is Zwick's strong bridging set with a constant-factor slack in the number of edges, and with walks in place of paths because a shortest walk may run around zero-weight cycles. Let $R:=\lceil n^{\gamma}\rceil$ (so $R$ is tiny!) and $t_i:=\min\{R^{i},n\}$ for $i=0,1,\dots,\lceil1/\gamma\rceil$. We shall construct $(t_i,C_i)$-bridges $B_i$ with $|B_i|=\tilde O(n/t_i)$, starting with $B_0=V$ and $C_0=1$, where $C_i$ is a constant that depends only on $i$. The slack $C_i$ grows by a constant factor from each level to the next. This is why the levels are a factor $n^{\gamma}$ apart, so that there are only $O(1/\gamma)$ of them, rather than a factor $3/2$ apart as in Zwick's algorithm.

\emph{The tables $E_i(T,H)$.} The distances to and from small sets of vertices are computed by the following routine. Given the bridges $B_0,\dots,B_i$, a set $T$ of $\tilde O(n/t_i)$ \emph{target} vertices, and a \emph{horizon} $H$ with $t_i\le H\le O(Rt_i)$, the routine returns two tables, $E_i(T,H)[V,T]$ and $E_i(T,H)[T,V]$, with three properties: every finite entry is the weight of a walk that the routine represents explicitly, as described next; every represented walk has $O(H)$ edges, with a constant that depends on $i$; and the entry for $(u,v)$ is $\delta(u,v)$ whenever $\eta(u,v)\le H$. The products in the routine are computed with the deterministic $(\min,+)$-product algorithm with witnesses of Zwick~\cite[Lemma~3.3]{Zwi02}, which finds the witnesses deterministically by the method of Alon and Naor~\cite{AN96}. A witness of an entry $(u,v)$ of a product is an index $k$ at which its minimum is attained, and we store it with the entry. The walk that the entry represents is then the concatenation of the walks represented by the entries $(u,k)$ and $(k,v)$ of the two factors, and it is listed, when needed, by following the stored witnesses recursively down to the entries of the weight matrix, which represent single edges; this takes $O(H)$ time, because the walk is a concatenation of $O(H)$ such entries. For a product in which one dimension is $n$ and the other two are at most $q$, with entries of absolute value at most $W$, that algorithm takes $\tilde O(Wnq^{\omega-1})$ time, by cutting the product into $n/q$ square products.

The routine for level $0$ takes the weight matrix with zeroes on the diagonal and squares it $O(\log H)$ times, with the entries larger than $2HM$ replaced by $\infty$, and keeps the rows and columns of $T$. The routine for level $i\ge1$ first calls the routine for level $i-1$ with the target set $T\cup B_i$ and the horizon $L:=2C_it_i$. This call is okay, because $|T\cup B_i|=\tilde O(n/t_i)\le\tilde O(n/t_{i-1})$ and $L=O(Rt_{i-1})$. Let $F$ be the pair of tables that it returns; they have exact distances for the pairs with $\eta\le L$, and they have the entries of $V\times(T\cup B_i)$ and $(T\cup B_i)\times V$. Let $Q$ be the smallest power of two above $H/t_i$, and set the diagonal of $F[B_i,B_i]$ to zero. The routine computes, by repeated squaring,
\[
E_i(T,H)[V,T]:=\min\bigl\{F[V,T],\;F[V,B_i]\star F[B_i,B_i]^{\star Q}\star F[B_i,T]\bigr\},
\]
where $\star$ is the $(\min,+)$-product and the minimum is entrywise, and it computes $E_i(T,H)[T,V]$ symmetrically.

The entry for a pair $(u,v)$ with $v\in T$ and $\eta(u,v)\le H$ is exact. To see this, take a shortest path from $u$ to $v$ with the fewest edges, and cut it into blocks of $t_i$ edges followed by a remainder of fewer than $t_i$ edges. Each block is a shortest path with the fewest edges between its endpoints, so the bridge property of $B_i$ replaces it by a shortest walk between the same endpoints that has at most $C_it_i$ edges and passes through a vertex of $B_i$; choose one such vertex in each block. The result is a shortest walk from $u$ to $v$ through at most $H/t_i$ chosen vertices of $B_i$, in which $u$, the chosen vertices, and $v$ are consecutively at most $L=2C_it_i$ edges apart. The table $F$ is exact on these consecutive pairs, so the product above finds $\delta(u,v)$. The represented walks have $O(H)$ edges, since those of $F$ have $O(L)$ edges and at most $Q+2=O(H/t_i)$ of them are concatenated. Every product of the routine has one dimension $n$, the two others $\tilde O(n/t_i)$, and entries $O(HM)=\tilde O(Rt_i)$, so it takes $\tilde O(Rt_i\cdot n\cdot(n/t_i)^{\omega-1})\le\tilde O(n^{\omega+\gamma})$ time, and the recursion has $O(1/\gamma)$ levels.

\emph{The stages.} The algorithm maintains a matrix of distance estimates, initially the weight matrix with zeros on the diagonal, and runs one stage for each level $i=0,1,\dots$ until $t_{i+1}=n$. The stage for level $i$ computes the tables $E_i:=E_i(B_i,C_it_{i+1})$, then the $(\min,+)$-product
\[
E_i[V,B_i]\star E_i[B_i,V],
\]
and takes its entrywise minimum with the matrix of estimates. Then, if $t_{i+1}<n$, it constructs $B_{i+1}$ from the walks represented in $E_i$, as described below. The product fixes every pair $(u,v)$ with $t_i\le\eta(u,v)\le t_{i+1}$: the bridge property of $B_i$ gives $b\in B_i$ with $\delta(u,v)=\delta(u,b)+\delta(b,v)$ and $\eta(u,b),\eta(b,v)\le C_i\eta(u,v)\le C_it_{i+1}$, so both $E_i[u,b]$ and $E_i[b,v]$ are exact. The stages thus find all distances, and every estimate is the weight of a walk, so none is too small.

\emph{The cost of the stages.} Let $s:=t_{i+1}$. The product of the stage for level $i$ has middle dimension $|B_i|=\tilde O(n/t_i)=\tilde O(Rn/s)$ and entries $\tilde O(s)$. Cutting $B_i$ into $\tilde O(R)$ groups of $\lceil n/s\rceil$ vertices turns it into $\tilde O(n^{\gamma})$ products of an $n\times\lceil n/s\rceil$ matrix by an $\lceil n/s\rceil\times n$ matrix with entries $\tilde O(s)$. Up to this grouping, these are the products of the randomized reduction, and the analysis of~\cite[Theorem~3.1]{CVX21} applies to them. Fix a small constant $a>0$. When $s\le n^{1-\mu-a}$, the middle dimension $n^{r}:=\lceil n/s\rceil$ has $r\ge\mu+a$, and fast rectangular matrix multiplication computes the product in $\tilde O(s\,n^{\omega(1,r,1)})=O(n^{2+\mu-ca})$ time~\cite[Lemma~3.3]{Zwi02}, for a constant $c>0$, by the convexity of $\omega(1,\cdot,1)$, since $\omega(1,\mu,1)=1+2\mu$ and $\omega<2+\mu$. When $s\ge n^{1-\mu+a}$, the naive algorithm computes it in $O(n^{2}\cdot n/s)=O(n^{2+\mu-a})$ time. In between, the middle dimension is at most $n^{\mu+a}$ and the entries are $\tilde O(n^{1-\mu+a})$; we cut the middle dimension into at most $n^{a}$ pieces of $n^{\mu}$ columns and pad each piece to an instance with $n^{1+O(a)}$ rows, so that its entries fit the range of the hypothesis, and the hypothesized algorithm computes it in $O(n^{2+\mu-\varepsilon_1+O(a)})$ time. With $a$ small enough, the product of every stage takes $O(n^{2+\mu+\gamma-\varepsilon_2})$ time, for a constant $\varepsilon_2>0$ that depends on $\varepsilon_1$ and $\mu$.

\emph{The next bridge.} Let $s:=t_{i+1}<n$. The bridge $B_{i+1}$ is read off the walks represented in $E_i=E_i(B_i,C_is)$, and this is where the small target set pays off. These are $2n|B_i|$ walks with $O(s)$ edges each, so listing them all takes $\tilde O(n|B_i|s)=\tilde O(n^{2}R)\leq \tilde O(n^{2+\gamma})$ time, whereas listing a path for every pair of vertices, as Zwick's construction does, would take $O(n^{2}s)$ time, which exceeds our budget at the large scales. From each listed walk, erase the cycles, and keep the resulting path if it has at least $s/2$ edges. Let $B_{i+1}$ be a greedy hitting set of the vertex sets of the kept paths, each of which has more than $s/2$ vertices. It has $O((n/s)\log n)$ vertices by the analysis of the greedy heuristic~\cite{Lov75,Chv79}, and it is computed in time linear in the total size of the sets, up to a logarithmic factor.

This $B_{i+1}$ is an $(s,C_{i+1})$-bridge for a constant $C_{i+1}$ that depends only on $i$. Consider first a pair $(x,y)$ with $\eta(x,y)=s$. The bridge $B_i$ gives $b$ with $\delta(x,y)=\delta(x,b)+\delta(b,y)$ and $\eta(x,b)+\eta(b,y)\le C_is$, so the entries of $E_i$ for $(x,b)$ and $(b,y)$ are exact, and the two walks that represent them are shortest walks. A cycle on a shortest walk has weight zero, so the two paths obtained by erasing the cycles are still shortest, and their concatenation is a shortest walk from $x$ to $y$ with $O(s)$ edges. This walk has at least $s$ edges, since otherwise erasing its cycles would give a shortest path from $x$ to $y$ with fewer than $\eta(x,y)$ edges. So one of the two paths has at least $s/2$ edges and was kept, and $B_{i+1}$ contains one of its vertices, which lies on the concatenated walk. For a pair with $\eta(x,y)>s$, apply this to the first $s$ edges of a shortest path from $x$ to $y$ with the fewest edges, and keep the rest of that path; the walk obtained has at most $C_{i+1}s+\eta(x,y)-s\le C_{i+1}\eta(x,y)$ edges.

\emph{Running time.} There are $O(1/\gamma)$ stages. Each computes its tables in $\tilde O(n^{\omega+\gamma})$ time, its product in $O(n^{2+\mu+\gamma-\varepsilon_2})$ time, and its bridge in $\tilde O(n^{2+\gamma})$ time. Choosing $\gamma$ smaller than both $\varepsilon_2/2$ and $(2+\mu-\omega)/2$ gives the running time $O(n^{2+\mu-\varepsilon''})$ for a constant $\varepsilon''>0$, and every step is deterministic.
\end{proof}

We compose the known reduction with the deterministic $(\min,+)$-product algorithm of Theorem~\ref{thm:3sumapsp}.

\begin{theorem}[Directed unweighted APSP]\label{thm:dirapsp}
There is a constant $\varepsilon''>0$ such that APSP in directed unweighted $n$-vertex graphs, and more generally in directed $n$-vertex graphs with integer weights of absolute value at most $(\log n)^c$ and no negative cycles, for any constant $c$, can be solved in $O(n^{2+\mu-\varepsilon''})$ time by a deterministic algorithm.
\end{theorem}
\begin{proof}
Cut the $n\times n^{\mu}$ matrix into $n^{1-\mu}$ blocks of $n^{\mu}$ consecutive rows, and the $n^{\mu}\times n$ matrix into $n^{1-\mu}$ blocks of $n^{\mu}$ consecutive columns. The $(\min,+)$-product then consists of $n^{2-2\mu}$ products of two $n^{\mu}\times n^{\mu}$ matrices, one for each pair of blocks. Their entries have absolute value $\tilde{O}(n^{1-\mu})\leq\tilde{O}(n^{\mu})$, because $\mu\ge\frac12$, so Theorem~\ref{thm:3sumapsp} computes each of them in $\tilde{O}((n^{\mu})^{3-\varepsilon'/3})$ time, and all of them in $\tilde{O}(n^{2+\mu-\mu\varepsilon'/3})$ time. This is $O(n^{2+\mu-\varepsilon_1})$ for any constant $\varepsilon_1<\mu\varepsilon'/3$, so~\cite[Theorem~3.1]{CVX21} applies, in the deterministic form of Theorem~\ref{thm:detred}, since the algorithm of Theorem~\ref{thm:3sumapsp} is deterministic.
\end{proof}

\subsection{3SUM, APSP, and Exact Triangle with real inputs}\label{sec:real}

\begin{theorem}[Real inputs]\label{thm:real}
In the real RAM in which the only operations allowed on real numbers are comparisons, additions, and subtractions, Las Vegas algorithms solve the following problems on real inputs: $3$SUM on $n$ numbers in $O(n^{1.998})$ expected time, and the $(\min,+)$-product of two $n\times n$ matrices, APSP with no negative cycles, and Exact Triangle in $O(n^{2.998})$ expected time. The bounds also hold with probability $1-n^{-c}$ for any constant $c$.
\end{theorem}

The reductions are by Chan, Vassilevska Williams, and Xu~\cite{CVX22}, who reduce the real-valued problems to the counting version of All-Edges Sparse Triangle on graphs with $m$ edges (and real APSP also to the Boolean version). Our version of All-Edges Sparse Triangle in Corollary~\ref{fact:lop} needs lopsided sparse graphs.
To deal with this slight change, we follow the reductions of~\cite{CVX22} down to the point where the All-Edges Sparse Triangle oracle is called, and change the call. At that point, the reductions only need the oracle to compute \emph{comparison counts}, which we define next. These come from Fredman's trick~\cite{Fre76}: $a'+b'<a+b$ if and only if $a'-a<b-b'$. This turns a comparison of two sums into a comparison of a number that depends only on the row against a number that depends only on the column.

{We first define comparison counts, then restate the reductions of~\cite{CVX22} as reductions to comparison counts (Lemma~\ref{lem:realred}), and then solve the comparison counts problem: following~\cite{CVX22}, Lemma~\ref{lem:cmp} reduces it to the wanted entries of one thin matrix product using Matou\v sek's technique for the dominance product~\cite{Mat91}, and Corollary~\ref{cor:cmp} applies Corollary~\ref{cor:zt} to that product.} The proof of Theorem~\ref{thm:real} then adds up the running times. The only randomized steps are in the reductions of~\cite{CVX22}, and the algorithms never output a wrong answer, so they are Las Vegas.

\paragraph{Comparison counts.}
We are given $n$ \emph{row lists}: for every $r\in[n]$, a list $\mathcal R_r$ of at most $d$ real numbers $x$, where each $x\in\mathcal R_r$ has a \emph{color} $\chi(x)\in[d]$ associated with it.

Similarly, we are given $n$ \emph{column lists}: for every $c\in[n]$, a list $\mathcal C_c$ of at most $d$ real numbers $y$, each $y\in\mathcal C_c$ again with a color $\chi(y)\in[d]$ associated with it.

Given a set $P\subseteq[n]^2$ of (row, column) pairs, we want, for every $(r,c)\in P$, the number
\[
\gamma(r,c):=\big|\{(x,y):\ x\in\mathcal R_r,\ y\in\mathcal C_c,\ \chi(x)=\chi(y),\ x<y\}\big| .
\]

\paragraph{The reductions.}
We follow the reductions of~\cite{CVX22} down to their two counting problems~\cite[Problems~4.1 and~4.5]{CVX22}, which we define in the proof of Lemma~\ref{lem:realred} below, and turn every call of a counting problem into comparison counts, in place of the triangle-counting instances of~\cite[Lemmas~4.3 and~4.7]{CVX22}. Throughout, $d\in[n]$ is a parameter, which the proof of Theorem~\ref{thm:real} chooses.

\begin{lemma}[after Chan, Vassilevska Williams, and Xu~\cite{CVX22}]\label{lem:realred}
Let $1\le d\le n$.
\begin{enumerate}
\item[(a)] The $(\min,+)$-product of two $n\times n$ real matrices, APSP with real edge weights and no negative cycles on $n$ vertices, and Exact Triangle with real weights on $n$ vertices reduce, with Las Vegas randomization, to $\tilde{O}(n/d)$ \emph{counting calls} and $\tilde{O}(n^3/d)$ further time. A counting call consists of $d$ comparison counts, each with $n$ row lists and $n$ column lists of at most $d$ numbers, with at most one number of each color in a list, and with pair sets $P_1,\dots,P_d$ satisfying $\sum_k|P_k|=n^2$; forming its lists costs $O(nd^2)$ subtractions.
\item[(b)] $3$SUM on $n$ real numbers reduces, with Las Vegas randomization, to polylogarithmically many comparison counts, each with $n$ row lists and $n$ column lists of at most $d$ numbers, all of the same color, and a pair set $P$ of size $O(n^2/d)$, and $\tilde{O}(n^2/d)$ further time; forming the lists of a comparison count costs $O(nd)$ subtractions.
\end{enumerate}
A comparison count may have $\le$ in place of $<$. The only operations on real numbers performed by the reductions are comparisons, additions, and subtractions.
\end{lemma}

\begin{proof}
\emph{(a) $(\min,+)$-product, APSP, and Exact Triangle.} Let $A,B,C\in\mathbb R^{n\times n}$. As in~\cite[Section~3.2]{CVX22}, cut $A$ into $n/d$ blocks $A'\in\mathbb R^{n\times d}$ of $d$ consecutive columns, and $B$ into the corresponding blocks $B'\in\mathbb R^{d\times n}$ of $d$ consecutive rows. For each pair $(A',B')$ we solve~\cite[Problem~3.2]{CVX22}: find, for every $(i,j)$, the predecessor and the successor of $C[i,j]$ among the $d$ sums $A'[i,k]+B'[k,j]$, where a sum equal to $C[i,j]$ counts as its predecessor. For Exact Triangle, $A[i,k]:=w(i,k)$, $B[k,j]:=w(k,j)$, and $C[i,j]:=-w(i,j)$, and there is a triangle of weight zero if and only if, in some block, some $C[i,j]$ is its own predecessor. For the $(\min,+)$-product of $A$ and $B$, we take every $C[i,j]$ smaller than all the sums, for instance $C[i,j]:=\min_k A[i,k]+\min_k B[k,j]-1$. Then the successor of $C[i,j]$ in a block is the smallest sum of the block, and the product is the entrywise minimum over the $n/d$ blocks, which costs $O(n^3/d)$ comparisons. APSP with real weights and no negative cycles can be computed by successive squaring of the weight matrix, performing $\lceil\log_2n\rceil$ such products.

Fix a pair $(A',B')$.  We solve the following \emph{counting problem}~\cite[Problem~4.1]{CVX22}: we are given a \emph{pivot} $k_{ij}\in[d]$ for every $(i,j)\in[n]^2$ and a subset $S\subseteq[d]$, and we must count, for every $(i,j)$, the indices $k'\in S$ with $A'[i,k']+B'[k',j]<A'[i,k_{ij}]+B'[k_{ij},j]$; a call of the problem may also have $\le$ in place of $<$. 

By Lemmas~3.4 and~4.2 of~\cite{CVX22}, the problem for one pair $(A',B')$ reduces, with Las Vegas randomization, to polylogarithmically many calls of the counting problem, and $\tilde{O}(n^2)$ further time: Lemma~3.4 finds the predecessors and successors by a randomized search that repeatedly asks, for every $(i,j)$, for a sum strictly between two given sums, and Lemma~4.2 finds these sums by counting, for each of the two given sums, the sums below it, and by sampling at random. Lemma~4.3 of~\cite{CVX22} turns a call into $O(d)$ triangle-counting instances; we turn it into $d$ comparison counts instead. 

Following~\cite{CVX22}, we group the pairs by their pivot, $P_k:=\{(i,j):k_{ij}=k\}$. For $k\in[d]$, let the row list $\mathcal R_i$ consist of the numbers $A'[i,k']-A'[i,k]$, $k'\in S$, each with color $k'$, and let the column list $\mathcal C_j$ consist of the numbers $B'[k,j]-B'[k',j]$, $k'\in S$, again with color $k'$.

By Fredman's trick, $A'[i,k']+B'[k',j]<A'[i,k]+B'[k,j]$ if and only if $A'[i,k']-A'[i,k]<B'[k,j]-B'[k',j]$, where both numbers in the latter comparison have color $k'$.

Thus the count of a pair $(i,j)\in P_k$ is the comparison count $\gamma(i,j)$ of these lists with $P:=P_k$. Forming the lists costs us $O(nd)$ subtractions for each $k$. Over the $n/d$ pairs $(A',B')$, and the $\lceil\log_2n\rceil$ products in the case of APSP, this gives $\tilde{O}(n/d)$ counting calls and $\tilde{O}(n^3/d)$ additional time, which includes the $O(n^3/d)$ comparisons of the entrywise minimum.

\emph{(b) $3$SUM.} Let $A,B,C$ be sets of $n$ reals, the problem asks whether $a+b=c$ for some $a\in A$, $b\in B$, and $c\in C$. The $3$SUM version with one set and $a+b+c=0$ reduces to this version in the standard way.

Now, sort $A$ and $B$, and cut them into $n/d$ consecutive blocks $A_1,\dots,A_{n/d}$ and $B_1,\dots,B_{n/d}$ of $d$ numbers each, writing $A_i[k]$ for the $k$-th number of $A_i$. Here the counting problem~\cite[Problem~4.5]{CVX22} is the following: we are given a set $\mathcal Q$ of $O(n^2/d)$ quadruples $(i,j,k,\ell)$, and we must count, for every $(i,j,k,\ell)\in\mathcal Q$, the pairs $(k',\ell')\in[d]^2$ with $A_i[k']+B_j[\ell']<A_i[k]+B_j[\ell]$; again, a call may have $\le$ in place of $<$, and it may restrict $k'$ and $\ell'$ to subsets of $[d]$. By Section~3.3 and Lemmas~3.9 and~4.6 of~\cite{CVX22}, $3$SUM reduces, with Las Vegas randomization, to polylogarithmically many calls of this counting problem, and $\tilde{O}(n^2/d)$ additional time. By Fredman's trick, the condition is $A_i[k']-A_i[k]<B_j[\ell]-B_j[\ell']$. So a call is one comparison count: the rows are the $n$ pairs $(i,k)$ and the columns are the $n$ pairs $(j,\ell)$, the row list $\mathcal R_{(i,k)}$ consists of the $d$ numbers $A_i[k']-A_i[k]$, $k'\in[d]$, the column list $\mathcal C_{(j,\ell)}$ of the $d$ numbers $B_j[\ell]-B_j[\ell']$, $\ell'\in[d]$, all numbers have the same color, and $P:=\mathcal Q$. A restriction of $k'$ or of $\ell'$ deletes numbers from the lists. Forming the lists can be done using $O(nd)$ subtractions.
\end{proof}

\paragraph{Solving comparison counts.}
In the comparison counts of Lemma~\ref{lem:realred}(a), every list has at most one number of each color, so that $\gamma(r,c)$ is the number of colors $k$ with $A[r,k]<B[k,c]$, where $A[r,k]$ is the number of color $k$ in $\mathcal R_r$ and $B[k,c]$ the one in $\mathcal C_c$ (a color missing from either list is skipped); this is the dominance product of the $n\times d$ matrix $A$ and the $d\times n$ matrix $B$ at the positions in $P$. In those of Lemma~\ref{lem:realred}(b), all numbers have the same color, and $\gamma(r,c)$ is the number of pairs $x<y$ between the two lists. Merging the two sorted lists of every pair in $P$ takes $O(|P|d)$ time. 
To improve on this, we utilize Matou\v sek's technique for the dominance product~\cite{Mat91}, which counts the pairs that lie in different blocks of the sorted order by one thin matrix product, and the pairs in the same block directly.

\begin{lemma}[after Matou\v sek~\cite{Mat91}]\label{lem:cmp}
Let $s\ge1$ be an integer and $D'':=\lceil2nd/s\rceil+d$. Given row lists $\mathcal R_r$ and column lists $\mathcal C_c$ as above and a set $P\subseteq[n]^2$, we can build matrices $X\in\{0,\dots,d\}^{n\times D''}$ and $Y\in\{0,\dots,d\}^{D''\times n}$ and compute integers $\gamma_2(r,c)$, $(r,c)\in P$, with
\[
\gamma(r,c)=(XY)[r,c]+\gamma_2(r,c)\qquad\text{for every }(r,c)\in P,
\]
deterministically in $O(nD''+(|P|+nds+ds^2)\log n)$ time. The only operations on real numbers are the comparisons made when sorting the numbers of each color. The same holds with $\le$ in place of $<$ in the definition of $\gamma$.
\end{lemma}

\begin{proof}
For each color $k$, sort the numbers of color $k$ of all the lists, and break ties so that, among equal numbers, those from column lists precede those from row lists. Let $\mathrm{pos}(x)$ be the position of $x$ in this order. For $x$ from a row list and $y$ from a column list of the same color, $x<y$ holds if and only if $\mathrm{pos}(x)<\mathrm{pos}(y)$. (For $\le$, we reverse the tie-breaking rule.) Cut the order into \emph{blocks} of $s$ consecutive positions, and let $\mathrm{blk}(x):=\lfloor\mathrm{pos}(x)/s\rfloor$. Then $x<y$ holds if and only if either $\mathrm{blk}(x)<\mathrm{blk}(y)$, or $\mathrm{blk}(x)=\mathrm{blk}(y)$ and $\mathrm{pos}(x)<\mathrm{pos}(y)$. As in Matou\v sek's algorithm, we count the pairs of the first kind, which lie in different blocks, by a matrix product, and the pairs of the second kind, which lie in the same block, by brute force.

\emph{Different blocks.} The lists contain at most $2nd$ numbers in all, so there are at most $2nd/s+d\le D''$ pairs (color, block), and we index the columns of $X$ and the rows of $Y$ by the pairs, padding with zeros if needed. Let $X[r,(k,\beta)]$ be the number of numbers of color $k$ in $\mathcal R_r$ that lie in block $\beta$, and let $Y[(k,\beta),c]$ be the number of numbers of color $k$ in $\mathcal C_c$ that lie in a block after $\beta$. Then $(XY)[r,c]=\sum_{(k,\beta)}X[r,(k,\beta)]\,Y[(k,\beta),c]$ is the number of pairs $(x,y)$ of the same color with $\mathrm{blk}(x)<\mathrm{blk}(y)$, as required, and the entries of $X$ and $Y$ are between $0$ and $d$ because a list has at most $d$ numbers. We use $O(nd\log n)$ comparisons to sort and then fill in $X$ and $Y$ in 
 $O(nD'')$ time.

\emph{Same block.} For every color $k$ and block $\beta$, let $I$ be the numbers $x$ of row lists in the block, each with its row, and let $J$ be the numbers $y$ of column lists in the block, each with its column, so that $|I|+|J|\le s$. For every $(x,y)\in I\times J$ with $\mathrm{pos}(x)<\mathrm{pos}(y)$ whose row and column form a pair $(r,c)\in P$ (this can be tested via a binary search in $P$ after it is sorted), we add $1$ to $\gamma_2(r,c)$. Then $\gamma_2$ counts the pairs of the second kind. Since $|I||J|\le s^2/4$ and there are at most $2nd/s+d$ blocks, we enumerate at most $(2nd/s+d)s^2/4\le nds+ds^2$ pairs $(x,y)$. Together with sorting $P$ and outputting the $|P|$ integers $\gamma_2(r,c)$, this takes $O((|P|+nds+ds^2)\log n)$ time.
\end{proof}

\begin{corollary}\label{cor:cmp}
Let $d\le n^{1/40}$, and let $n$ be larger than a suitable constant. Given row lists $\mathcal R_r$ and column lists $\mathcal C_c$ as above and a set $P\subseteq[n]^2$, the counts $\gamma(r,c)$, $(r,c)\in P$, with $<$ or with $\le$, can be computed deterministically in
\[
O\big(|P|\,n^{0.0229}+n^2/d^{0.126}+n^{2-1/440}\log n\big)
\]
time, by an algorithm whose only operations on real numbers are comparisons.
\end{corollary}
\begin{proof}
Apply Lemma~\ref{lem:cmp} with $s:=\lceil n^{1-1/440}/d\rceil$, so that $D''\le2d^2n^{1/440}+d+1\le3d^2n^{1/440}$, and pad the product to the inner dimension $D^*:=\lceil3d^2n^{1/440}\rceil$ with zero columns of $X$ and zero rows of $Y$. 
Since $d\le n^{1/40}$, we have $D^*\le4n^{23/440}$, so $(D^*)^{18}\le4^{18}n^{0.941}\le n$, and Corollary~\ref{cor:zt} with $N:=n$ and $D:=D^*$ computes $(XY)[r,c]$ for all $(r,c)\in P$ deterministically in $O(|P|(D^*)^{0.437}+n^2/(D^*)^{0.063})$ time. Here $(D^*)^{0.437}\le4^{0.437}n^{0.437\cdot23/440}=O(n^{0.0229})$ and $(D^*)^{0.063}\ge(3d^2)^{0.063}\ge d^{0.126}$. 

Since $nD^*=O(n^{1.06})$, $nds\le2n^{2-1/440}$ and $ds^2\le4n^{2-1/440}$, the time of Lemma~\ref{lem:cmp} itself is $$O(nD^*+(|P|+nds+ds^2)\log n)=O(|P|\,n^{0.0229}+n^{2-1/440}\log n).$$
\end{proof}

\begin{proof}[Proof of Theorem~\ref{thm:real}]
Let $d:=\lfloor n^{1/40}\rfloor$, apply Lemma~\ref{lem:realred} with this $d$, and answer every comparison count by Corollary~\ref{cor:cmp}.

\emph{$(\min,+)$-product, APSP, and Exact Triangle.} By Corollary~\ref{cor:cmp}, the $d$ comparison counts of a counting call cost
\[
\sum_{k\in[d]}O\big(|P_k|\,n^{0.0229}+n^2/d^{0.126}+n^{2-1/440}\log n\big)=O(n^{2.0229}\log n),
\]
since $\sum_k|P_k|=n^2$, $d\cdot n^2/d^{0.126}\le n^{2+0.874/40}\le n^{2.0219}$, and $d\cdot n^{2-1/440}\le n^{2+1/40-1/440}\le n^{2.0228}$; forming the lists costs $O(nd^2)=O(n^{1.05})$. So the $\tilde{O}(n/d)$ counting calls of Lemma~\ref{lem:realred}(a) cost
\[
\tilde{O}\big((n/d)\,n^{2.0229}\big)=\tilde{O}\big(n^{3-1/40+0.0229}\big)\leq O(n^{2.998})
\]
expected time, and its additional time is $\tilde{O}(n^3/d)=\tilde{O}(n^{2.975})$. %

\emph{$3$SUM.} By Corollary~\ref{cor:cmp}, a comparison count of Lemma~\ref{lem:realred}(b) costs
\[
O\Big(\frac{n^2}{d}\,n^{0.0229}+\frac{n^2}{d^{0.126}}+n^{2-1/440}\log n\Big)=O\big(n^{2-1/40+0.0229}\log n\big)=O(n^{1.9979}\log n),
\]
so $3$SUM can also be solved in $\tilde{O}(n^{1.9979})\leq O(n^{1.998})$ expected time.

\emph{Las Vegas, and high probability.} The reductions of~\cite{CVX22} never return a wrong answer and Corollary~\ref{cor:cmp} is deterministic, so the output is always correct, and only the running time is random. The bounds also hold with probability $1-n^{-c}$ for any constant $c$: we stop a run that exceeds twice the bound on its expected time and start again, and $c\log_2n$ runs all fail with probability at most $n^{-c}$.
\end{proof}

\subsection{Zero-Weight, Min-Weight, and Max-Weight \texorpdfstring{$k$}{k}-Clique}\label{sec:kclique}

The classical reduction of Ne\v{s}et\v{r}il and Poljak~\cite{NP85} from $k$-Clique to triangle detection turns the weighted $k$-Clique problems into Exact Triangle instances with $n^{\lfloor k/3\rfloor}$ vertices per part (after fixing $k-3\lfloor k/3\rfloor$ vertices when $3\nmid k$), so Theorem~\ref{thm:zwt} gives:
\begin{corollary}[Weighted $k$-Clique]\label{cor:kclique}
Let $k\ge3$ and $\nu\ge1$ be constants. Given a complete $k$-partite graph with parts of $n$ vertices and integer edge weights of absolute value at most $n^\nu$, deterministic algorithms decide in $O(n^{k-\varepsilon_T\lfloor k/3\rfloor})$ time whether some $k$-clique, with one vertex in each part, has total edge weight zero, and find a $k$-clique of minimum, or of maximum, total edge weight.
\end{corollary}
\begin{proof}
We apply the classical reduction from $k$-Clique to triangle detection of Ne\v{s}et\v{r}il and Poljak~\cite{NP85}. When $k$ is divisible by $3$, the reduction is simple. Split the $k$ parts into three groups of $k/3$ parts each, and build a tripartite graph $H$ whose vertices in the $i$th part (for $i=1,2,3$) are the $k/3$-cliques of the original graph with one vertex in each part of the $i$th group, with an edge between two such cliques when together they form a $2k/3$-clique. Since the original graph is complete $k$-partite, every choice of one vertex from each part of a group is a $k/3$-clique, and any two such cliques from different groups form a $2k/3$-clique, so $H$ is simply the complete tripartite graph with $n^{k/3}$ vertices per part. The triangles of $H$ correspond bijectively to the $k$-cliques of the original graph with one vertex in each part.

When $k$ is not divisible by $3$, let $t:=k-3\lfloor k/3\rfloor\in\{1,2\}$. We enumerate the $n^t$ ways of fixing one vertex in each of the first $t$ parts, and for each of them we build $H$ as above from the remaining $3\lfloor k/3\rfloor$ parts; now the triangles of $H$ correspond to the $k$-cliques through the $t$ fixed vertices. In both cases $H$ has $N:=n^{\lfloor k/3\rfloor}$ vertices per part, and we index its parts by $i\in\mathbb Z_3$.

We give $H$ edge weights so that a triangle and its $k$-clique have the same weight. Let $U$ be a vertex of $H$ in part $i$ and $V$ a vertex in part $i+1$. The weight of the edge $UV$ of $H$ is the total weight of the edges of the original graph between a vertex of $U$ and a vertex of $V$, of the edges inside $U$, of the edges between $U$ and the fixed vertices, and, if $i=1$ and $t=2$, of the edge between the two fixed vertices. The edges inside $V$ and the edges between $V$ and the fixed vertices are not counted here; they are counted in the edge from $V$ to part $i+2$. In this way every edge of a $k$-clique is counted exactly once over the three edges of its triangle, so there is no double-counting and the weights agree. The weight of an edge of $H$ has absolute value at most $\binom k2n^\nu\le N^{2\nu}$ for $n\ge k^2$.

Theorem~\ref{thm:zwt} decides whether $H$ has a triangle of weight zero in $O(N^{3-\varepsilon'}\log N)$ time, with the constants $\varepsilon'>\varepsilon_T$ of that theorem. By~\cite[Theorem~3.3]{VW13}, finding a maximum weight triangle costs $O(\log^2n)$ times as much, and a minimum weight triangle is a maximum weight triangle for the negated weights. Since $\varepsilon'>\varepsilon_T$, the logarithmic factors are absorbed, so each graph $H$ costs $O(N^{3-\varepsilon_T})$ time, and the $n^t$ graphs together cost $O(n^tN^{3-\varepsilon_T})=O(n^{k-\varepsilon_T\lfloor k/3\rfloor})$.
\end{proof}

Consequently, the Zero-Weight, Min-Weight, and Max-Weight $k$-Clique hypotheses~\cite{AL13,LVW18,BDT16} with polynomially bounded weights are refuted, also on general graphs, by the standard reduction to the complete $k$-partite form.

\subsection{Three conjectures of van den Brand, Nanongkai, and Saranurak}\label{sec:bns}

Van den Brand, Nanongkai, and Saranurak~\cite{BNS19} propose three ``hinted'' variants of the Online Matrix--Vector conjecture~\cite{HKNS15}, which they use to prove matching lower bounds for the running times of their dynamic matrix inverse algorithms. All three are over the Boolean semiring, allow polynomial preprocessing time in the first phase, and have parameters $t=n^{\tau}$ and $t_i=n^{\tau_i}$ for constants $\tau,\tau_i\in(0,1)$. We write $\omega(a,b,c)$ for the exponent of multiplying an $n^a\times n^b$ matrix by an $n^b\times n^c$ matrix. Each conjecture asserts that no algorithm {\em simultaneously} beats all of its listed bounds, for any $\varepsilon>0$.
\begin{itemize}
\item \emph{$v$-hinted Mv} (Definition~5.1 and Conjecture~5.2 of~\cite{BNS19}). Phase~1: an $n\times t$ matrix $M$. Phase~2: a $t\times n$ matrix $V$. Phase~3: an index $i$, after which the algorithm outputs $MV_{[n],i}$, the product of $M$ and column $i$ of $V$. Bounds conjectured to be impossible to achieve simultaneously: $n^{\omega(1,1,\tau)-\varepsilon}$ for Phase~2 and $n^{1+\tau-\varepsilon}$ for Phase~3.
\item \emph{Mv-hinted Mv} (Definition~5.6 and Conjecture~5.7  of~\cite{BNS19}). Phase~1: $N\in\{0,1\}^{n\times n}$ and $V\in\{0,1\}^{t\times n}$. Phase~2: a vector $I\in[n]^t$ of column indices. Phase~3: an index $j$, after which the algorithm outputs $N_{[n],I}V_{[t],j}$, where $N_{[n],I}$ is the $n\times t$ matrix whose $k$-th column is column $I_k$ of $N$. Bounds conjectured to be impossible to achieve simultaneously: $n^{\omega(1,\tau,1)-\varepsilon}$ for Phase~2 and $n^{1+\tau-\varepsilon}$ for Phase~3.
\item \emph{uMv-hinted uMv} (Definition~5.11 and Conjecture~5.12  of~\cite{BNS19}). Phase~1: $U\in\{0,1\}^{n\times t_1}$, $N\in\{0,1\}^{n\times n}$, and $V\in\{0,1\}^{t_2\times n}$. Phase~2: $I\in[n]^{t_1}$. Phase~3: $J\in[n]^{t_2}$. Phase~4: indices $i$ and $j$, after which the algorithm outputs $(UN_{I,J}V)_{i,j}$, where $N_{I,J}$ is the $t_1\times t_2$ submatrix of $N$ with the rows $I$ and the columns $J$. Bounds conjectured to be impossible to achieve simultaneously: $n^{\omega(1,\tau_1,1)-\varepsilon}$ for Phase~2, $n^{\omega(\tau_2,\tau_1,1)-\varepsilon}$ for Phase~3, and $n^{\tau_1+\tau_2-\varepsilon}$ for Phase~4.
\end{itemize}
The two obvious algorithms are to multiply everything out with fast matrix multiplication as soon as the hint arrives, or to wait and compute one matrix--vector product at the end, and the conjectures say that no algorithm is polynomially faster than both of these. The main lower bounds of~\cite{BNS19}, $n^{1.528}$ for dynamic matrix inverse with column updates and row queries and $n^{1.407}$ with element updates and element queries, use the conjectures at $\tau\approx0.53$ and at $\tau_1+\tau_2\approx1.41$, where they match the algorithms of~\cite{BNS19}. For thin hints, the conjectures are false, because the data structure of Corollary~\ref{cor:zt} is faster than both obvious algorithms.

\begin{corollary}[Hinted OMv]\label{cor:bns}
Conjectures~5.2 and~5.7 of~\cite{BNS19} fail for every $0<\tau<1/18$: Phase~2 takes $O(n^{2-0.063\tau})$ time and Phase~3 takes $O(n^{1+0.437\tau})$ time. Conjecture~5.12 fails for every $0<\tau_1<\tau_2/18$: Phase~3 takes $O(n^{1+\tau_2-0.063\tau_1})$ time and Phase~4 takes $O(n^{\tau_2+0.437\tau_1})$ time, with polynomial Phase~1 and a Phase~2 that only stores $I$.

More generally, Conjectures~5.2 and~5.7 fail for every $0<\tau<\varepsilon^*=0.1204\ldots$, and Conjecture~5.12 fails for every $0<\tau_1<\varepsilon^*\tau_2$. In each case the phases beat the conjectured bounds whatever the value of $\omega$, since $\omega(1,1,\tau)=\omega(1,\tau,1)\ge2$ and $\omega(\tau_2,\tau_1,1)\ge1+\tau_2$ because of the input and output sizes.
\end{corollary}
\begin{proof}
We compute over $\Z$ with $0/1$ matrices; a Boolean entry is 1 exactly when the corresponding integer entry is positive. Let $D:=t=n^\tau$ be the inner dimension. For $\tau<1/18$ we have $D^{18}=n^{18\tau}\le n$, so Corollary~\ref{cor:zt} applies with $N=n$.

\emph{Conjecture~5.2 of~\cite{BNS19}.} In Phase~2, we preprocess $X:=M$ and $Y:=V$ by Corollary~\ref{cor:zt}, in $O(n^2/D^{0.063})=O(n^{2-0.063\tau})$ time. In Phase~3, the $n$ entries of column $i$ of $MV$ are $n$ queries, which take $O(n\,D^{0.437})=O(n^{1+0.437\tau})$ time.

\emph{Conjecture~5.7 of~\cite{BNS19}.} The proof is the same, with $X:=N_{[n],I}$ and $Y:=V$, which are both known in Phase~2 ($X$ may have repeated columns, which do not affect the argument). Phase~1 only reads its input.

\emph{Conjecture~5.12 of~\cite{BNS19}.} Let $X:=U$, an $n\times t_1$ matrix, and $Y:=N_{I,J}$, a $t_1\times t_2$ matrix, which are both known in Phase~3, with the inner dimension $D:=t_1=n^{\tau_1}$. Cut the $n$ rows of $U$ into $n/t_2$ blocks of $t_2$ rows, and preprocess each of the $n/t_2$ products of a block with $Y$, which are products of a $t_2\times t_1$ matrix by a $t_1\times t_2$ matrix, by Corollary~\ref{cor:zt}. It applies because $t_2=n^{\tau_2}\ge n^{18\tau_1}=D^{18}$ when $\tau_1<\tau_2/18$. So Phase~3 costs $(n/t_2)\cdot O(t_2^2/D^{0.063})=O(n^{1+\tau_2-0.063\tau_1})$ time. In Phase~4, $(UN_{I,J}V)_{i,j}=\sum_{\ell=1}^{t_2}(UN_{I,J})_{i,\ell}V_{\ell,j}$ is a sum of at most $t_2$ entries of $XY$, which are $t_2$ queries and take $O(t_2D^{0.437})=O(n^{\tau_2+0.437\tau_1})$ time. Phase~2 only stores $I$.

\emph{General $\tau$.} Let $0<\tau<\varepsilon^*$, and fix $\varepsilon$ with $\tau<\varepsilon<\varepsilon^*$. By Corollary~\ref{cor:tradeoff}, with $c$ and $\theta$ chosen as in the proof of Theorem~\ref{thm:single} for $q:=\frac12$, there is, after absorbing the logarithmic factors, a $\gamma>0$ such that, since $D=n^\tau\le n^\varepsilon$, the arguments above go through with Phase~2 in $O(n^2/D^\gamma)=O(n^{2-\gamma\tau})$ time and Phase~3 in $O(n\,D^{1/2})=O(n^{1+\tau/2})$ time. For Conjecture~5.12, we use the same blocks, which requires $n^{\tau_1}\le n^{\varepsilon\tau_2}$, that is, $\tau_1\le\varepsilon\tau_2$ for some $\varepsilon<\varepsilon^*$.
\end{proof}

We conclude this section by discussing which results of~\cite{BNS19} rely on the conjectures in the refuted regime. From Conjectures~5.2 and~5.7,~\cite{BNS19} derive trade-offs between the worst-case update time $u$ and query time $q$ of dynamic algorithms, for every $0<\tau<1$: either $u=\Omega(n^{\omega(1,1,\tau)-\tau-\varepsilon})$ or $q=\Omega(n^{1+\tau-\varepsilon})$. These trade-offs cover the dynamic matrix product, inverse, and adjoint problems with column updates and row queries (Theorem~5.3 and Corollary~5.5 in~\cite{BNS19}) or with element updates and row queries (Theorem~5.8 and Corollary~5.9 in~\cite{BNS19}), and dynamic transitive closure and DAG path counting with vertex or edge updates and source queries (Corollaries~5.4 and~5.10 in~\cite{BNS19}). For element or pair queries, the condition becomes $q=\Omega(n^{\tau-\varepsilon})$ (Corollaries~5.9 and~5.10 in~\cite{BNS19}).

The main bounds, such as $u+q=\Omega(n^{1.528})$ for column updates and row queries, balance the two terms at $\tau\approx0.53$. For this setting, the current best upper bound on $\omega(1,1,\tau)$ is strictly greater than $2$, so the products are far from thin enough for our construction, and these bounds are still supported by the unrefuted forms of the conjectures. The same holds for the bounds from Conjecture~5.12 in~\cite{BNS19}, such as $\Omega(n^{1.407})$ for dynamic determinant, rank, and bipartite perfect matching (Corollaries~5.15 and~5.16 in~\cite{BNS19}), since Corollary~\ref{cor:bns} refutes Conjecture~5.12 in~\cite{BNS19} only for $\tau_1<0.1204\,\tau_2$, far from the $\tau_1+\tau_2\approx1.41$ at which these bounds are proved.

For $\tau<0.1204$, however, where $\omega(1,1,\tau)=2$, the trade-offs say that an update requires $n^{2-\tau-o(1)}$ time unless a row or source query takes $n^{1+\tau-o(1)}$ time, or, in the versions with element or pair queries, unless such a query takes $n^{\tau-o(1)}$ time; Corollary~\ref{cor:bns} refutes the forms of the conjectures that would imply these trade-offs. This is the end of the trade-offs with the fastest queries, where they match the algorithms of~\cite{San04,BNS19}, such as element updates in $O(n^{2-\tau})$ time with element queries in $O(n^{\tau})$ time~\cite[Theorem~3]{San04}. For small $\tau$, these algorithms are no longer known to be conditionally optimal.

The same goes for results that rely on this end of the trade-offs. Haeupler, Long, and Saranurak~\cite{HLS24} cite Corollary~5.10 of~\cite{BNS19} as showing that incremental and decremental distance oracles with approximation factor below $5/3$ cannot have worst-case update time $n^{2-\Omega(1)}$ and query time $n^{o(1)}$. With queries this fast, ruling out update time $n^{2-\delta}$ needs Conjecture~5.7 of~\cite{BNS19} for some $\tau<\delta$, so new evidence of hardness would be needed for $\delta\le0.1204$, although the remaining cases of the conjecture still rule out update time $O(n^{1.87})$. Van den Brand, Song, and Zhou~\cite{BSZ24} show that their data structure for dynamic attention, with amortized update time $n^{\omega(1,1,\tau)-\tau}$, is conditionally optimal for every $0<\tau\le1$ under a variant of Conjecture~5.7 of~\cite{BNS19} in which Phase~2 gives a matrix with at most $n^{\tau}$ nonzero entries instead of the vector $I$. Our proof applies to this variant as well, since after Phase~2 the product again has inner dimension at most $n^{\tau}$, so this lower bound needs a new hypothesis for $\tau<0.1204$.

%% file: 06-conclusion.tex
\section{Conclusion}\label{sec:conclusion}

We believe our results open exciting research directions in both algorithm design and complexity theory, and we highlight a few here.

\paragraph{Fine-grained complexity.}
An important message from our paper is that the fine-grained reductions that FGC has built up are incredibly valuable, and that it is even more important to design more reductions to enable further algorithm design and improved conditional hardness.

{ Which hypotheses should take the place of $3$SUM and APSP?} As discussed in the introduction, one possibility is to restrict these hypotheses to combinatorial algorithms~\cite{VW10,VW18,AV14}, motivated by the search for more practical algorithms or practical hardness.
 Another possibility is the hypothesis that the balanced case of All-Edges Sparse Triangle, with $m$ edges, needs $m^{4/3-o(1)}$ time. Our algorithms do not speed up this balanced case, and since it is the target of several lower-bound reductions from $3$SUM and Exact Triangle~\cite{Pat10,KPP16,VX20a,CX24}, the hypothesis still gives many problems evidence of hardness. Is it possible that the lopsided version has a faster algorithm but the balanced version does not? 

More broadly, the complexity of the many problems that are only known to be 3SUM-hard or APSP-hard, rather than equivalent to $3$SUM or APSP, is now open (the gray boxes of Figure~\ref{fig:web}); we get no faster algorithms for them, since the reductions go the other way. For instance, can one decide in truly subquadratic time whether $n$ points in the plane contain three on a line~\cite{GO95}? As far as we know, the fastest known running time is still $O(n^2)$, without even any polylogarithmic improvements.  A similar question arises for the hypotheses that our results say nothing about, such as $k$-SUM for $k\ge4$, 3SUM-Indexing, OMv (without hints), and Orthogonal Vectors. Do ideas like ours help to speed up algorithms for any of them?

Our new technique gives a way to deal with sparsity, at least in lopsided instances. Much work has gone into algorithms for sparse matrix multiplication~\cite{YZ05,AP09,PS14,ABFK24,BGGW25,BFN26,Gra26}. However, our techniques do not seem to give new algorithms for sparse matrix multiplication. Is improving upon the known algorithms for this problem hard, or is there another technique that could help here? This question is related to the balanced version of All-Edges Sparse Triangle.

What are the true exponents of $3$SUM, APSP, and the many other problems of Figure~\ref{fig:web} that now have faster algorithms? Our exponents can certainly be improved, and one place to gain is the reductions. Each of our algorithms composes known reductions with the matrix theorem, and many reductions lose a constant fraction of the saving in the exponent. For instance, the reduction from Exact Triangle to Lopsided All-Edges Sparse Triangle keeps only half of the saving in the exponent, and the reductions from $3$SUM and APSP to Exact Triangle keep a half and a third (Section~\ref{sec:reduction}). These reductions, like most fine-grained reductions, were designed to prove hardness, where it only matters that they keep some polynomial saving. Now that they are being used as algorithms, how much they keep is more important. Which of them can be improved? Are there more efficient reductions from all these different problems to Lopsided All-Edges Sparse Triangle, or to thin matrix products themselves, that skip the intermediate problems?\footnote{\label{footnote:improvements}For one example, the exponents for APSP and $3$SUM presented here can already be improved slightly without too much work, by not going through the decision version of Exact Triangle. For APSP, the proof of Theorem~\ref{thm:red} also solves the all-edges version of Exact Triangle (scanning each pair only until its zero triangle is found), to which the $(\min,+)$-product reduces on the same $n$ vertices~\cite{VW10,VW18,VW13}, so that APSP keeps all of the saving of Exact Triangle rather than a third. For $3$SUM, the direct reduction of~\cite{KPP16} to offline Set Disjointness keeps all of the saving rather than a quarter, but it is randomized (its derandomization~\cite{FKP24} keeps far less of the saving). Substantially better improvements than these would be more exciting. We focused in this paper on simplifying the presentation rather than giving slight further improvements to the exponents.} 

Some of these losses cannot be removed by reductions alone. Sheffield, Vassilevska Williams, and Xi~\cite{SVX27} show that the loss of a third in the reduction from the all-edges version of a triangle problem to its detection version~\cite{VW10,VW18} is optimal for black-box reductions, those that work for every relation on the three edge weights at once, so a tighter reduction must use the structure of the particular problem. The loss can still be avoided on the algorithmic side, as in the footnote: our algorithm for Exact Triangle solves the all-edges version at no extra cost, as all known algorithms for triangle problems do~\cite{SVX27}. Our reductions are not black-box in this sense either: the reduction from Exact Triangle to Lopsided All-Edges Sparse Triangle hashes the weights modulo a prime, which uses the additive structure of the problem, so the barriers of~\cite{SVX27} say nothing against improving the half that it loses. The same paper also gives new black-box reductions among triangle, listing, and matrix problems, and such reductions are now of algorithmic interest as well.

Beyond the reductions, all of our algorithms use the same matrix theorem, and a different technique might give improvements such as better exponents, combinatorial or practical algorithms, or algorithms for All-Edges Sparse Triangle with a middle part larger than $n^{0.12}$, such as the $n^{1/2}$ that arises in an approach to girth approximation~\cite{RV12}.

\paragraph{Algebraic algorithm design.}
Any improvement to the matrix theorem carries over to all the other speedups we achieve here. We do not know the true complexity of computing a set $W$ of entries of a thin product $XY$, as a function of $D$ and $|W|$. In particular, to handle $D>N^{0.12}$, one could look for new base identities that compute longer inner products per multiplication than Sch\"onhage's, perhaps by computer search~\cite{FBH+22,KM23,Nov25}, or ask whether other known identities, such as border rank expressions for the Coppersmith--Winograd tensor~\cite{CW90},  which underlie the current bounds on $\omega$ and $\alpha$~\cite{VXXZ24,ADVXXZ25,DEK+26}, can be used in the same way. It is not immediately clear how to make use of other tensors, since our analysis uses specific properties of Sch\"onhage's identity, such as its sparsity and the way the leaves of its recursion are shared among the output entries (Section~\ref{sec:sparse}). If other tensors cannot be used in this way, it would be interesting to understand why. For instance, even if the rank or border rank of a tensor is known, it might be that we need to find different rank or border rank expressions to use, similar to work on improving leading constants of matrix multiplication algorithms (see, e.g.,~\cite{BB15,KS20,BS19,AY25,MS25}). %

%% file: ZT_simplified_appendixreductions.bbl
\newcommand{\etalchar}[1]{$^{#1}$}
\begin{thebibliography}{CMWW19}

\bibitem[ABF23]{ABF23}
Amir Abboud, Karl Bringmann, and Nick Fischer.
\newblock Stronger {3-SUM} lower bounds for approximate distance oracles via
  additive combinatorics.
\newblock In {\em Proc. 55th ACM Symposium on Theory of Computing (STOC 2023)},
  pages 391--404, 2023.

\bibitem[ABFK24]{ABFK24}
Amir Abboud, Karl Bringmann, Nick Fischer, and Marvin K{\"u}nnemann.
\newblock The time complexity of fully sparse matrix multiplication.
\newblock In {\em Proc. 35th ACM-SIAM Symposium on Discrete Algorithms (SODA
  2024)}, pages 4670--4703, 2024.

\bibitem[ABKZ22]{ABKZ22}
Amir Abboud, Karl Bringmann, Seri Khoury, and Or~Zamir.
\newblock Hardness of approximation in {P} via short cycle removal: Cycle
  detection, distance oracles, and beyond.
\newblock In {\em Proc. 54th ACM Symposium on Theory of Computing (STOC 2022)},
  pages 1487--1500, 2022.
\newblock Full version: arXiv:2204.10465.

\bibitem[AC05]{AC05}
Nir Ailon and Bernard Chazelle.
\newblock Lower bounds for linear degeneracy testing.
\newblock {\em Journal of the ACM}, 52(2):157--171, 2005.

\bibitem[ACJ{\etalchar{+}}22]{ACJRW22}
Shyan Akmal, Lijie Chen, Ce~Jin, Malvika Raj, and Ryan Williams.
\newblock Improved {M}erlin--{A}rthur protocols for central problems in
  fine-grained complexity.
\newblock In {\em Proc. 13th Innovations in Theoretical Computer Science
  Conference (ITCS 2022)}, volume 215 of {\em LIPIcs}, pages 3:1--3:25, 2022.

\bibitem[ACLL14]{ACLL14}
Amihood Amir, Timothy~M. Chan, Moshe Lewenstein, and Noa Lewenstein.
\newblock On hardness of jumbled indexing.
\newblock In {\em Proc. 41st International Colloquium on Automata, Languages,
  and Programming (ICALP 2014)}, volume 8572 of {\em LNCS}, pages 114--125,
  2014.

\bibitem[ACW16]{ACW16}
Josh Alman, Timothy~M. Chan, and Ryan Williams.
\newblock Polynomial representations of threshold functions and algorithmic
  applications.
\newblock In {\em Proc. 57th IEEE Symposium on Foundations of Computer Science
  (FOCS 2016)}, pages 467--476, 2016.

\bibitem[ACW20]{ACW20}
Josh Alman, Timothy~M. Chan, and R.~Ryan Williams.
\newblock Faster deterministic and {Las Vegas} algorithms for offline
  approximate nearest neighbors in high dimensions.
\newblock In {\em Proc. 31st ACM-SIAM Symposium on Discrete Algorithms (SODA
  2020)}, pages 637--649, 2020.

\bibitem[ADKF70]{ADKF70}
Vladimir~L. Arlazarov, Efim~A. Dinic, Mikhail~A. Kronrod, and Igor~A. Farad{\v
  z}ev.
\newblock On economical construction of the transitive closure of an oriented
  graph.
\newblock {\em Soviet Mathematics Doklady}, 11:1209--1210, 1970.

\bibitem[ADV{\etalchar{+}}25]{ADVXXZ25}
Josh Alman, Ran Duan, Virginia {Vassilevska Williams}, Yinzhan Xu, Zixuan Xu,
  and Renfei Zhou.
\newblock More asymmetry yields faster matrix multiplication.
\newblock In {\em Proc. 36th ACM-SIAM Symposium on Discrete Algorithms (SODA
  2025)}, pages 2005--2039, 2025.

\bibitem[AFK{\etalchar{+}}24]{AFK+24}
Amir Abboud, Nick Fischer, Zander Kelley, Shachar Lovett, and Raghu Meka.
\newblock New graph decompositions and combinatorial {B}oolean matrix
  multiplication algorithms.
\newblock In {\em Proc. 56th ACM Symposium on Theory of Computing (STOC 2024)},
  pages 935--943, 2024.

\bibitem[AFW20]{paritypaper}
Amir Abboud, Shon Feller, and Oren Weimann.
\newblock On the fine-grained complexity of parity problems.
\newblock In {\em Proc. 47th International Colloquium on Automata, Languages,
  and Programming (ICALP 2020)}, volume 168 of {\em LIPIcs}, pages 5:1--5:19,
  2020.

\bibitem[AGHP92]{AGHP92}
Noga Alon, Oded Goldreich, Johan H{\aa}stad, and Ren{\'e} Peralta.
\newblock Simple constructions of almost $k$-wise independent random variables.
\newblock {\em Random Structures \& Algorithms}, 3(3):289--304, 1992.

\bibitem[AGM97]{AGM97}
Noga Alon, Zvi Galil, and Oded Margalit.
\newblock On the exponent of the all pairs shortest path problem.
\newblock {\em Journal of Computer and System Sciences}, 54(2):255--262, 1997.

\bibitem[AGMN92]{AGMN92}
Noga Alon, Zvi Galil, Oded Margalit, and Moni Naor.
\newblock Witnesses for {B}oolean matrix multiplication and for shortest paths.
\newblock In {\em Proc. 33rd IEEE Symposium on Foundations of Computer Science
  (FOCS 1992)}, pages 417--426, 1992.

\bibitem[AGV15]{AGV15}
Amir Abboud, Fabrizio Grandoni, and Virginia {Vassilevska Williams}.
\newblock Subcubic equivalences between graph centrality problems, {APSP} and
  diameter.
\newblock In {\em Proc. 26th ACM-SIAM Symposium on Discrete Algorithms (SODA
  2015)}, pages 1681--1697, 2015.

\bibitem[AGV23]{AGV23}
Amir Abboud, Fabrizio Grandoni, and Virginia {Vassilevska Williams}.
\newblock Subcubic equivalences between graph centrality problems, {APSP}, and
  diameter.
\newblock {\em ACM Transactions on Algorithms}, 19(1):3:1--3:30, 2023.

\bibitem[AHP08]{AH08}
Boris Aronov and Sariel Har-Peled.
\newblock On approximating the depth and related problems.
\newblock {\em SIAM Journal on Computing}, 38(3):899--921, 2008.

\bibitem[AHU74]{AHU74}
Alfred~V. Aho, John~E. Hopcroft, and Jeffrey~D. Ullman.
\newblock {\em The Design and Analysis of Computer Algorithms}.
\newblock Addison-Wesley, 1974.

\bibitem[AHY25]{AHY25}
Josh Alman, Yizhi Huang, and Kevin Yeo.
\newblock Fine-grained complexity in a world without cryptography.
\newblock In {\em Advances in Cryptology -- EUROCRYPT 2025, Part VII}, volume
  15607 of {\em LNCS}, pages 375--405, 2025.

\bibitem[AL13]{AL13}
Amir Abboud and Kevin Lewi.
\newblock Exact weight subgraphs and the {$k$}-sum conjecture.
\newblock In {\em Proc. 40th International Colloquium on Automata, Languages,
  and Programming (ICALP 2013)}, volume 7965 of {\em LNCS}, pages 1--12, 2013.

\bibitem[AL25]{AL25}
Josh Alman and Baitian Li.
\newblock Kronecker powers, orthogonal vectors, and the asymptotic spectrum.
\newblock In {\em Proc. 66th IEEE Symposium on Foundations of Computer Science
  (FOCS 2025)}, pages 1411--1441, 2025.
\newblock arXiv:2509.14489.

\bibitem[ALW14]{ALW14}
Amir Abboud, Kevin Lewi, and Ryan Williams.
\newblock Losing weight by gaining edges.
\newblock In {\em Proc. 22nd European Symposium on Algorithms (ESA 2014)},
  volume 8737 of {\em LNCS}, pages 1--12, 2014.

\bibitem[AN96]{AN96}
Noga Alon and Moni Naor.
\newblock Derandomization, witnesses for {Boolean} matrix multiplication and
  construction of perfect hash functions.
\newblock {\em Algorithmica}, 16(4--5):434--449, 1996.

\bibitem[AP09]{AP09}
Rasmus~Resen Amossen and Rasmus Pagh.
\newblock Faster join-projects and sparse matrix multiplications.
\newblock In {\em Proc. 12th International Conference on Database Theory (ICDT
  2009)}, pages 121--126, 2009.

\bibitem[AV14]{AV14}
Amir Abboud and Virginia {Vassilevska Williams}.
\newblock Popular conjectures imply strong lower bounds for dynamic problems.
\newblock In {\em Proc. 55th IEEE Symposium on Foundations of Computer Science
  (FOCS 2014)}, pages 434--443, 2014.

\bibitem[AVW14]{AVWW14}
Amir Abboud, Virginia {Vassilevska Williams}, and Oren Weimann.
\newblock Consequences of faster alignment of sequences.
\newblock In {\em Proc. 41st International Colloquium on Automata, Languages,
  and Programming (ICALP 2014)}, volume 8572 of {\em LNCS}, pages 39--51, 2014.

\bibitem[AVY15]{AVY15}
Amir Abboud, Virginia {Vassilevska Williams}, and Huacheng Yu.
\newblock Matching triangles and basing hardness on an extremely popular
  conjecture.
\newblock In {\em Proc. 47th ACM Symposium on Theory of Computing (STOC 2015)},
  pages 41--50, 2015.

\bibitem[AW15]{AW15}
Josh Alman and Ryan Williams.
\newblock Probabilistic polynomials and {H}amming nearest neighbors.
\newblock In {\em Proc. 56th IEEE Symposium on Foundations of Computer Science
  (FOCS 2015)}, pages 136--150, 2015.

\bibitem[AWY15]{AWY15}
Amir Abboud, Ryan Williams, and Huacheng Yu.
\newblock More applications of the polynomial method to algorithm design.
\newblock In {\em Proc. 26th ACM-SIAM Symposium on Discrete Algorithms (SODA
  2015)}, pages 218--230, 2015.

\bibitem[AY25]{AY25}
Josh Alman and Hantao Yu.
\newblock Improving the leading constant of matrix multiplication.
\newblock In {\em Proc. 36th ACM-SIAM Symposium on Discrete Algorithms (SODA
  2025)}, pages 1933--1971, 2025.

\bibitem[AYZ97]{AYZ97}
Noga Alon, Raphael Yuster, and Uri Zwick.
\newblock Finding and counting given length cycles.
\newblock {\em Algorithmica}, 17(3):209--223, 1997.

\bibitem[BB15]{BB15}
Austin~R. Benson and Grey Ballard.
\newblock A framework for practical parallel fast matrix multiplication.
\newblock In {\em Proc. 20th ACM SIGPLAN Symposium on Principles and Practice
  of Parallel Programming (PPoPP 2015)}, pages 42--53, 2015.

\bibitem[BCD{\etalchar{+}}14]{BCD14}
David Bremner, Timothy~M. Chan, Erik~D. Demaine, Jeff Erickson, Ferran Hurtado,
  John Iacono, Stefan Langerman, Mihai P{\u a}tra{\c s}cu, and Perouz
  Taslakian.
\newblock Necklaces, convolutions, and {$X+Y$}.
\newblock {\em Algorithmica}, 69(2):294--314, 2014.

\bibitem[BDP08]{BDP08}
Ilya Baran, Erik~D. Demaine, and Mihai P{\u a}tra{\c s}cu.
\newblock Subquadratic algorithms for {3SUM}.
\newblock {\em Algorithmica}, 50(4):584--596, 2008.

\bibitem[BDT16]{BDT16}
Arturs Backurs, Nishanth Dikkala, and Christos Tzamos.
\newblock Tight hardness results for maximum weight rectangles.
\newblock In {\em Proc. 43rd International Colloquium on Automata, Languages,
  and Programming (ICALP 2016)}, volume~55 of {\em LIPIcs}, pages 81:1--81:13,
  2016.

\bibitem[BFN26]{BFN26}
Karl Bringmann, Nick Fischer, and Vasileios Nakos.
\newblock Robustifying sparse matrix multiplication.
\newblock In {\em Proc. 34th European Symposium on Algorithms (ESA 2026)},
  volume 388 of {\em LIPIcs}, pages 157:1--157:25, 2026.
\newblock Full version: arXiv:2607.01427.

\bibitem[BGGW26]{BGGW25}
Huck Bennett, Karthik Gajulapalli, Alexander Golovnev, and Evelyn Warton.
\newblock Output-sparse matrix multiplication using compressed sensing.
\newblock In {\em Approximation, Randomization, and Combinatorial Optimization.
  Algorithms and Techniques (APPROX/RANDOM 2026)}, volume 392 of {\em LIPIcs},
  pages 40:1--40:20, 2026.
\newblock Full version: arXiv:2508.10250.

\bibitem[BGMW20]{BGMW20}
Karl Bringmann, Pawe{\l} Gawrychowski, Shay Mozes, and Oren Weimann.
\newblock Tree edit distance cannot be computed in strongly subcubic time
  (unless {APSP} can).
\newblock {\em ACM Transactions on Algorithms}, 16(4):48:1--48:22, 2020.

\bibitem[BHKK10]{BHKK10}
Andreas Bj{\"o}rklund, Thore Husfeldt, Petteri Kaski, and Mikko Koivisto.
\newblock Trimmed {M}oebius inversion and graphs of bounded degree.
\newblock {\em Theory of Computing Systems}, 47(3):637--654, 2010.

\bibitem[BHP01]{BH01}
Gill Barequet and Sariel Har-Peled.
\newblock Polygon containment and translational min-{H}ausdorff-distance
  between segment sets are {3SUM}-hard.
\newblock {\em International Journal of Computational Geometry \&
  Applications}, 11(4):465--474, 2001.

\bibitem[Bin80]{Bin80}
Dario Bini.
\newblock Relations between exact and approximate bilinear algorithms.
  {A}pplications.
\newblock {\em Calcolo}, 17:87--97, 1980.

\bibitem[BIS17]{BIS17}
Arturs Backurs, Piotr Indyk, and Ludwig Schmidt.
\newblock Better approximations for tree sparsity in nearly-linear time.
\newblock In {\em Proc. 28th ACM-SIAM Symposium on Discrete Algorithms (SODA
  2017)}, pages 2215--2229, 2017.

\bibitem[BN21]{BN21}
Karl Bringmann and Vasileios Nakos.
\newblock A fine-grained perspective on approximating subset sum and partition.
\newblock In {\em Proc. 32nd ACM-SIAM Symposium on Discrete Algorithms (SODA
  2021)}, pages 1797--1815, 2021.

\bibitem[Bri19]{Bri19}
Karl Bringmann.
\newblock Fine-grained complexity theory (tutorial).
\newblock In {\em Proc. 36th International Symposium on Theoretical Aspects of
  Computer Science (STACS 2019)}, volume 126 of {\em LIPIcs}, pages 4:1--4:7,
  2019.

\bibitem[Bri24]{Bri24}
Karl Bringmann.
\newblock Knapsack with small items in near-quadratic time.
\newblock In {\em Proc. 56th ACM Symposium on Theory of Computing (STOC 2024)},
  pages 259--270, 2024.

\bibitem[BS19]{BS19}
Gal Beniamini and Oded Schwartz.
\newblock Faster matrix multiplication via sparse decomposition.
\newblock In {\em Proc. 31st ACM Symposium on Parallelism in Algorithms and
  Architectures (SPAA 2019)}, pages 11--22, 2019.

\bibitem[BW09]{BW09}
Nikhil Bansal and Ryan Williams.
\newblock Regularity lemmas and combinatorial algorithms.
\newblock In {\em Proc. 50th IEEE Symposium on Foundations of Computer Science
  (FOCS 2009)}, pages 745--754, 2009.

\bibitem[CDL{\etalchar{+}}16]{CDL+16}
Marek Cygan, Holger Dell, Daniel Lokshtanov, D{\'a}niel Marx, Jesper Nederlof,
  Yoshio Okamoto, Ramamohan Paturi, Saket Saurabh, and Magnus Wahlstr{\"o}m.
\newblock On problems as hard as {CNF-SAT}.
\newblock {\em ACM Transactions on Algorithms}, 12(3):41:1--41:24, 2016.

\bibitem[CGI{\etalchar{+}}16]{CGI+16}
Marco~L. Carmosino, Jiawei Gao, Russell Impagliazzo, Ivan Mihajlin, Ramamohan
  Paturi, and Stefan Schneider.
\newblock Nondeterministic extensions of the strong exponential time hypothesis
  and consequences for non-reducibility.
\newblock In {\em Proc. 7th ACM Conference on Innovations in Theoretical
  Computer Science (ITCS 2016)}, pages 261--270, 2016.

\bibitem[CH20]{CH20}
Timothy~M. Chan and Qizheng He.
\newblock Reducing {3SUM} to convolution-{3SUM}.
\newblock In {\em Proc. 3rd SIAM Symposium on Simplicity in Algorithms (SOSA
  2020)}, pages 1--7, 2020.

\bibitem[Cha08]{Cha08}
Timothy~M. Chan.
\newblock All-pairs shortest paths with real weights in {$O(n^3/\log n)$} time.
\newblock {\em Algorithmica}, 50(2):236--243, 2008.
\newblock Preliminary version in WADS 2005.

\bibitem[Cha10]{Cha10}
Timothy~M. Chan.
\newblock More algorithms for all-pairs shortest paths in weighted graphs.
\newblock {\em SIAM Journal on Computing}, 39(5):2075--2089, 2010.

\bibitem[Cha15]{Cha15}
Timothy~M. Chan.
\newblock Speeding up the four {R}ussians algorithm by about one more
  logarithmic factor.
\newblock In {\em Proc. 26th ACM-SIAM Symposium on Discrete Algorithms (SODA
  2015)}, pages 212--217, 2015.

\bibitem[Cha20]{Cha20}
Timothy~M. Chan.
\newblock More logarithmic-factor speedups for {3SUM},
  {(median,+)}-convolution, and some geometric {3SUM}-hard problems.
\newblock {\em ACM Transactions on Algorithms}, 16(1):7:1--7:23, 2020.

\bibitem[Chv79]{Chv79}
Va{\v s}ek Chv{\'a}tal.
\newblock A greedy heuristic for the set-covering problem.
\newblock {\em Mathematics of Operations Research}, 4(3):233--235, 1979.

\bibitem[CIO16]{CIO16}
Jean Cardinal, John Iacono, and Aur{\'e}lien Ooms.
\newblock Solving $k$-{SUM} using few linear queries.
\newblock In {\em Proc. 24th European Symposium on Algorithms (ESA 2016)},
  volume~57 of {\em LIPIcs}, pages 25:1--25:17, 2016.

\bibitem[CIP10]{CIP10}
Chris Calabro, Russell Impagliazzo, and Ramamohan Paturi.
\newblock On the exact complexity of evaluating quantified {$k$}-{CNF}.
\newblock In {\em Proc. 5th International Symposium on Parameterized and Exact
  Computation (IPEC 2010)}, volume 6478 of {\em LNCS}, pages 50--59, 2010.
\newblock Journal version: Algorithmica 65(4):817--827, 2013.

\bibitem[CMWW19]{CMWW19}
Marek Cygan, Marcin Mucha, Karol W{\fontencoding{T1}\selectfont\k{e}}grzycki,
  and Micha{\l} W{\l}odarczyk.
\newblock On problems equivalent to {(min,+)}-convolution.
\newblock {\em ACM Transactions on Algorithms}, 15(1):14:1--14:25, 2019.

\bibitem[CN85]{CN85}
Norishige Chiba and Takao Nishizeki.
\newblock Arboricity and subgraph listing algorithms.
\newblock {\em SIAM Journal on Computing}, 14(1):210--223, 1985.

\bibitem[Coo71]{Coo71}
Stephen~A. Cook.
\newblock The complexity of theorem-proving procedures.
\newblock In {\em Proc. 3rd ACM Symposium on Theory of Computing (STOC 1971)},
  pages 151--158, 1971.

\bibitem[Cop82]{Cop82}
Don Coppersmith.
\newblock Rapid multiplication of rectangular matrices.
\newblock {\em SIAM Journal on Computing}, 11(3):467--471, 1982.

\bibitem[Cop97]{Cop97}
Don Coppersmith.
\newblock Rectangular matrix multiplication revisited.
\newblock {\em Journal of Complexity}, 13(1):42--49, 1997.

\bibitem[CP10]{CP10}
Hagai Cohen and Ely Porat.
\newblock Fast set intersection and two-patterns matching.
\newblock {\em Theoretical Computer Science}, 411(40--42):3795--3800, 2010.

\bibitem[CT65]{CT65}
James~W. Cooley and John~W. Tukey.
\newblock An algorithm for the machine calculation of complex {F}ourier series.
\newblock {\em Mathematics of Computation}, 19(90):297--301, 1965.

\bibitem[CVX21]{CVX21}
Timothy~M. Chan, Virginia {Vassilevska Williams}, and Yinzhan Xu.
\newblock Algorithms, reductions and equivalences for small weight variants of
  all-pairs shortest paths.
\newblock In {\em Proc. 48th International Colloquium on Automata, Languages,
  and Programming (ICALP 2021)}, volume 198 of {\em LIPIcs}, pages 47:1--47:21,
  2021.
\newblock Full version: arXiv:2102.06181, whose numbering we follow.

\bibitem[CVX22]{CVX22}
Timothy~M. Chan, Virginia {Vassilevska Williams}, and Yinzhan Xu.
\newblock Hardness for triangle problems under even more believable hypotheses:
  Reductions from real {APSP}, real {3SUM}, and {OV}.
\newblock In {\em Proc. 54th ACM Symposium on Theory of Computing (STOC 2022)},
  pages 1501--1514, 2022.
\newblock Full version: arXiv:2203.08356, whose numbering we follow.

\bibitem[CVX23]{CVX23}
Timothy~M. Chan, Virginia {Vassilevska Williams}, and Yinzhan Xu.
\newblock Fredman's trick meets dominance product: Fine-grained complexity of
  unweighted {APSP}, {3SUM} counting, and more.
\newblock In {\em Proc. 55th ACM Symposium on Theory of Computing (STOC 2023)},
  pages 419--432, 2023.

\bibitem[CW90]{CW90}
Don Coppersmith and Shmuel Winograd.
\newblock Matrix multiplication via arithmetic progressions.
\newblock {\em Journal of Symbolic Computation}, 9(3):251--280, 1990.

\bibitem[CW21]{CW21}
Timothy~M. Chan and R.~Ryan Williams.
\newblock Deterministic {APSP}, orthogonal vectors, and more: Quickly
  derandomizing {R}azborov--{S}molensky.
\newblock {\em ACM Transactions on Algorithms}, 17(1):2:1--2:14, 2021.

\bibitem[CX24]{CX24}
Timothy~M. Chan and Yinzhan Xu.
\newblock Simpler reductions from exact triangle.
\newblock In {\em Proc. 7th SIAM Symposium on Simplicity in Algorithms (SOSA
  2024)}, pages 28--38, 2024.

\bibitem[DEK{\etalchar{+}}26]{DEK+26}
Emilien Dupont, Marvin Eisenberger, Borislav Kozlovskii, Abbas Mehrabian,
  Francisco J.~R. Ruiz, Abigail See, Renfei Zhou, Josh Alman, Virginia
  {Vassilevska Williams}, and Matej Balog.
\newblock Improving the matrix multiplication exponent with modern optimization
  and {AlphaEvolve}.
\newblock arXiv:2608.16884, 2026.

\bibitem[DGS20]{DGS20}
Bart{\l}omiej Dudek, Pawe{\l} Gawrychowski, and Tatiana Starikovskaya.
\newblock All non-trivial variants of {3-LDT} are equivalent.
\newblock In {\em Proc. 52nd ACM Symposium on Theory of Computing (STOC 2020)},
  pages 974--981, 2020.

\bibitem[DKPV20]{DKPW20}
Lech Duraj, Krzysztof Kleiner, Adam Polak, and Virginia {Vassilevska Williams}.
\newblock Equivalences between triangle and range query problems.
\newblock In {\em Proc. 31st ACM-SIAM Symposium on Discrete Algorithms (SODA
  2020)}, pages 30--47, 2020.

\bibitem[DMRW09]{DMRW09}
Erik~D. Demaine, Shay Mozes, Benjamin Rossman, and Oren Weimann.
\newblock An optimal decomposition algorithm for tree edit distance.
\newblock {\em ACM Transactions on Algorithms}, 6(1):2:1--2:19, 2009.

\bibitem[Dob90]{Dob90}
W{\l}odzimierz Dobosiewicz.
\newblock A more efficient algorithm for the min-plus multiplication.
\newblock {\em International Journal of Computer Mathematics}, 32(1--2):49--60,
  1990.

\bibitem[DSW18]{DSW18}
Martin Dietzfelbinger, Philipp Schlag, and Stefan Walzer.
\newblock A subquadratic algorithm for {3XOR}.
\newblock In {\em Proc. 43rd International Symposium on Mathematical
  Foundations of Computer Science (MFCS 2018)}, volume 117 of {\em LIPIcs},
  pages 59:1--59:15, 2018.

\bibitem[D{\"u}r23]{Dur23}
Anita D{\"u}rr.
\newblock Improved bounds for rectangular monotone {Min-Plus} product and
  applications.
\newblock {\em Information Processing Letters}, 181:106358, 2023.

\bibitem[Eri99]{Eri99}
Jeff Erickson.
\newblock Lower bounds for linear satisfiability problems.
\newblock {\em Chicago Journal of Theoretical Computer Science}, 1999(8), 1999.

\bibitem[ES17]{ES17}
Esther Ezra and Micha Sharir.
\newblock A nearly quadratic bound for the decision tree complexity of
  $k$-{SUM}.
\newblock In {\em Proc. 33rd International Symposium on Computational Geometry
  (SoCG 2017)}, volume~77 of {\em LIPIcs}, pages 41:1--41:15, 2017.

\bibitem[FBH{\etalchar{+}}22]{FBH+22}
Alhussein Fawzi, Matej Balog, Aja Huang, Thomas Hubert, Bernardino
  Romera-Paredes, Mohammadamin Barekatain, Alexander Novikov, Francisco J.~R.
  Ruiz, Julian Schrittwieser, Grzegorz Swirszcz, David Silver, Demis Hassabis,
  and Pushmeet Kohli.
\newblock Discovering faster matrix multiplication algorithms with
  reinforcement learning.
\newblock {\em Nature}, 610(7930):47--53, 2022.

\bibitem[Fis26]{Fis26}
Nick Fischer.
\newblock Universe reduction for {APSP}: Equivalence of three fine-grained
  hypotheses.
\newblock In {\em Proc. 58th ACM Symposium on Theory of Computing (STOC 2026)},
  pages 922--932, 2026.

\bibitem[FJX25]{FJX25}
Nick Fischer, Ce~Jin, and Yinzhan Xu.
\newblock New applications of {3SUM}-counting in fine-grained complexity and
  pattern matching.
\newblock In {\em Proc. 36th ACM-SIAM Symposium on Discrete Algorithms (SODA
  2025)}, pages 4547--4595, 2025.

\bibitem[FKP24]{FKP24}
Nick Fischer, Piotr Kaliciak, and Adam Polak.
\newblock Deterministic {3SUM}-hardness.
\newblock In {\em Proc. 15th Innovations in Theoretical Computer Science
  Conference (ITCS 2024)}, volume 287 of {\em LIPIcs}, pages 49:1--49:24, 2024.
\newblock Full version: arXiv:2310.12913, whose numbering we follow.

\bibitem[Flo62]{Flo62}
Robert~W. Floyd.
\newblock Algorithm 97: Shortest path.
\newblock {\em Communications of the ACM}, 5(6):345, 1962.

\bibitem[FM71]{FM71}
Michael~J. Fischer and Albert~R. Meyer.
\newblock Boolean matrix multiplication and transitive closure.
\newblock In {\em Proc. 12th Symposium on Switching and Automata Theory (SWAT
  1971)}, pages 129--131, 1971.

\bibitem[Fre76]{Fre76}
Michael~L. Fredman.
\newblock New bounds on the complexity of the shortest path problem.
\newblock {\em SIAM Journal on Computing}, 5(1):83--89, 1976.

\bibitem[GGH{\etalchar{+}}20]{GGHPV20}
Alexander Golovnev, Siyao Guo, Thibaut Horel, Sunoo Park, and Vinod
  Vaikuntanathan.
\newblock Data structures meet cryptography: {3SUM} with preprocessing.
\newblock In {\em Proc. 52nd ACM Symposium on Theory of Computing (STOC 2020)},
  pages 294--307, 2020.

\bibitem[GKLP17]{GKLP17b}
Isaac Goldstein, Tsvi Kopelowitz, Moshe Lewenstein, and Ely Porat.
\newblock Conditional lower bounds for space/time tradeoffs.
\newblock In {\em Proc. 15th Algorithms and Data Structures Symposium (WADS
  2017)}, volume 10389 of {\em LNCS}, pages 421--436, 2017.

\bibitem[GLP19]{GLP19}
Isaac Goldstein, Moshe Lewenstein, and Ely Porat.
\newblock On the hardness of set disjointness and set intersection with bounded
  universe.
\newblock In {\em Proc. 30th International Symposium on Algorithms and
  Computation (ISAAC 2019)}, volume 149 of {\em LIPIcs}, pages 7:1--7:22, 2019.

\bibitem[GM93]{GM93}
Zvi Galil and Oded Margalit.
\newblock Witnesses for {B}oolean matrix multiplication and for transitive
  closure.
\newblock {\em Journal of Complexity}, 9(2):201--221, 1993.

\bibitem[GO95]{GO95}
Anka Gajentaan and Mark~H. Overmars.
\newblock On a class of {$O(n^2)$} problems in computational geometry.
\newblock {\em Computational Geometry: Theory and Applications}, 5(3):165--185,
  1995.

\bibitem[Goo58]{Goo58}
Irving~J. Good.
\newblock The interaction algorithm and practical {Fourier} analysis.
\newblock {\em Journal of the Royal Statistical Society, Series B},
  20(2):361--372, 1958.

\bibitem[GP18]{GP18}
Allan Gr{\o}nlund and Seth Pettie.
\newblock Threesomes, degenerates, and love triangles.
\newblock {\em Journal of the ACM}, 65(4):22:1--22:25, 2018.

\bibitem[Gra26]{Gra26}
Omar Graia.
\newblock Optimal deterministic fully sparse matrix multiplication.
\newblock arXiv:2608.18496, 2026.

\bibitem[GS17]{GS17}
Omer Gold and Micha Sharir.
\newblock Improved bounds for {3SUM}, {$k$-SUM}, and linear degeneracy.
\newblock In {\em Proc. 25th European Symposium on Algorithms (ESA 2017)},
  volume~87 of {\em LIPIcs}, pages 42:1--42:13, 2017.

\bibitem[Han04]{Han04}
Yijie Han.
\newblock Improved algorithm for all pairs shortest paths.
\newblock {\em Information Processing Letters}, 91(5):245--250, 2004.

\bibitem[Han08]{Han08}
Yijie Han.
\newblock An {$O(n^3(\log\log n/\log n)^{5/4})$} time algorithm for all pairs
  shortest path.
\newblock {\em Algorithmica}, 51(4):428--434, 2008.
\newblock Preliminary version in ESA 2006.

\bibitem[Her14]{Her14}
Timon Hertli.
\newblock {3-SAT} faster and simpler---unique-{SAT} bounds for {PPSZ} hold in
  general.
\newblock {\em SIAM Journal on Computing}, 43(2):718--729, 2014.

\bibitem[HKNS15]{HKNS15}
Monika Henzinger, Sebastian Krinninger, Danupon Nanongkai, and Thatchaphol
  Saranurak.
\newblock Unifying and strengthening hardness for dynamic problems via the
  online matrix-vector multiplication conjecture.
\newblock In {\em Proc. 47th ACM Symposium on Theory of Computing (STOC 2015)},
  pages 21--30, 2015.

\bibitem[HLS24]{HLS24}
Bernhard Haeupler, Yaowei Long, and Thatchaphol Saranurak.
\newblock Dynamic deterministic constant-approximate distance oracles with
  $n^{\epsilon}$ worst-case update time.
\newblock In {\em Proc. 65th IEEE Symposium on Foundations of Computer Science
  (FOCS 2024)}, pages 2033--2044, 2024.

\bibitem[HT16]{HT16}
Yijie Han and Tadao Takaoka.
\newblock An {$O(n^3\log\log n/\log^2 n)$} time algorithm for all pairs
  shortest paths.
\newblock {\em Journal of Discrete Algorithms}, 38--41:9--19, 2016.
\newblock Preliminary version in SWAT 2012.

\bibitem[IP01]{IP01}
Russell Impagliazzo and Ramamohan Paturi.
\newblock On the complexity of {$k$}-{SAT}.
\newblock {\em Journal of Computer and System Sciences}, 62(2):367--375, 2001.

\bibitem[IPZ01]{IPZ01}
Russell Impagliazzo, Ramamohan Paturi, and Francis Zane.
\newblock Which problems have strongly exponential complexity?
\newblock {\em Journal of Computer and System Sciences}, 63(4):512--530, 2001.

\bibitem[IR78]{IR78}
Alon Itai and Michael Rodeh.
\newblock Finding a minimum circuit in a graph.
\newblock {\em SIAM Journal on Computing}, 7(4):413--423, 1978.

\bibitem[Jin24]{Jin24}
Ce~Jin.
\newblock 0-1 knapsack in nearly quadratic time.
\newblock In {\em Proc. 56th ACM Symposium on Theory of Computing (STOC 2024)},
  pages 271--282, 2024.

\bibitem[Joh77]{Joh77}
Donald~B. Johnson.
\newblock Efficient algorithms for shortest paths in sparse networks.
\newblock {\em Journal of the ACM}, 24(1):1--13, 1977.

\bibitem[JV16]{JV16}
Zahra Jafargholi and Emanuele Viola.
\newblock {3SUM}, {3XOR}, triangles.
\newblock {\em Algorithmica}, 74(1):326--343, 2016.

\bibitem[JX23]{JX23}
Ce~Jin and Yinzhan Xu.
\newblock Removing additive structure in {3SUM}-based reductions.
\newblock In {\em Proc. 55th ACM Symposium on Theory of Computing (STOC 2023)},
  pages 405--418, 2023.

\bibitem[Kar72]{Kar72}
Richard~M. Karp.
\newblock Reducibility among combinatorial problems.
\newblock In {\em Complexity of Computer Computations}, pages 85--103. Plenum
  Press, 1972.

\bibitem[Ker70]{Ker70}
Leslie~Robert Kerr.
\newblock {\em The Effect of Algebraic Structure on the Computational
  Complexity of Matrix Multiplication}.
\newblock PhD thesis, Cornell University, 1970.

\bibitem[KKP93]{KKP93}
David~R. Karger, Daphne Koller, and Steven~J. Phillips.
\newblock Finding the hidden path: Time bounds for all-pairs shortest paths.
\newblock {\em SIAM Journal on Computing}, 22(6):1199--1217, 1993.

\bibitem[KLM19]{KLM19}
Daniel~M. Kane, Shachar Lovett, and Shay Moran.
\newblock Near-optimal linear decision trees for {$k$}-{SUM} and related
  problems.
\newblock {\em Journal of the ACM}, 66(3):16:1--16:18, 2019.

\bibitem[KM23]{KM23}
Manuel Kauers and Jakob Moosbauer.
\newblock Flip graphs for matrix multiplication.
\newblock In {\em Proc. 48th International Symposium on Symbolic and Algebraic
  Computation (ISSAC 2023)}, pages 381--388, 2023.

\bibitem[KP19]{KP19}
Tsvi Kopelowitz and Ely Porat.
\newblock The strong {3SUM-INDEXING} conjecture is false.
\newblock arXiv:1907.11206, 2019.

\bibitem[KPP16]{KPP16}
Tsvi Kopelowitz, Seth Pettie, and Ely Porat.
\newblock Higher lower bounds from the {3SUM} conjecture.
\newblock In {\em Proc. 27th ACM-SIAM Symposium on Discrete Algorithms (SODA
  2016)}, pages 1272--1287, 2016.

\bibitem[KPS17]{KPS17}
Marvin K{\"u}nnemann, Ramamohan Paturi, and Stefan Schneider.
\newblock On the fine-grained complexity of one-dimensional dynamic
  programming.
\newblock In {\em Proc. 44th International Colloquium on Automata, Languages,
  and Programming (ICALP 2017)}, volume~80 of {\em LIPIcs}, pages 21:1--21:15,
  2017.

\bibitem[KS20]{KS20}
Elaye Karstadt and Oded Schwartz.
\newblock Matrix multiplication, a little faster.
\newblock {\em Journal of the ACM}, 67(1):1:1--1:31, 2020.
\newblock Preliminary version in SPAA 2017.

\bibitem[KV20]{KV20}
Tsvi Kopelowitz and Virginia {Vassilevska Williams}.
\newblock Towards optimal set-disjointness and set-intersection data
  structures.
\newblock In {\em Proc. 47th International Colloquium on Automata, Languages,
  and Programming (ICALP 2020)}, volume 168 of {\em LIPIcs}, pages 74:1--74:16,
  2020.

\bibitem[LG12]{LeG12}
Fran{\c c}ois Le~Gall.
\newblock Faster algorithms for rectangular matrix multiplication.
\newblock In {\em Proc. 53rd IEEE Symposium on Foundations of Computer Science
  (FOCS 2012)}, pages 514--523, 2012.

\bibitem[LLV19]{LLV19}
Rio LaVigne, Andrea Lincoln, and Virginia {Vassilevska Williams}.
\newblock Public-key cryptography in the fine-grained setting.
\newblock In {\em Advances in Cryptology -- CRYPTO 2019, Part III}, volume
  11694 of {\em LNCS}, pages 605--635, 2019.

\bibitem[Lov75]{Lov75}
L{\'a}szl{\'o} Lov{\'a}sz.
\newblock On the ratio of optimal integral and fractional covers.
\newblock {\em Discrete Mathematics}, 13(4):383--390, 1975.

\bibitem[LPV20]{LPV20}
Andrea Lincoln, Adam Polak, and Virginia {Vassilevska Williams}.
\newblock Monochromatic triangles, intermediate matrix products, and
  convolutions.
\newblock In {\em Proc. 11th Innovations in Theoretical Computer Science
  Conference (ITCS 2020)}, volume 151 of {\em LIPIcs}, pages 53:1--53:18, 2020.
\newblock Full version: arXiv:2009.14479, whose statement of Theorem 8 we
  follow.

\bibitem[LU18]{LU18}
Fran{\c c}ois {Le Gall} and Florent Urrutia.
\newblock Improved rectangular matrix multiplication using powers of the
  {C}oppersmith--{W}inograd tensor.
\newblock In {\em Proc. 29th ACM-SIAM Symposium on Discrete Algorithms (SODA
  2018)}, pages 1029--1046, 2018.

\bibitem[LVW18]{LVW18}
Andrea Lincoln, Virginia {Vassilevska Williams}, and R.~Ryan Williams.
\newblock Tight hardness for shortest cycles and paths in sparse graphs.
\newblock In {\em Proc. 29th ACM-SIAM Symposium on Discrete Algorithms (SODA
  2018)}, pages 1236--1252, 2018.

\bibitem[LVWW16]{LVWW16}
Andrea Lincoln, Virginia {Vassilevska Williams}, Joshua~R. Wang, and R.~Ryan
  Williams.
\newblock Deterministic time-space trade-offs for $k$-{SUM}.
\newblock In {\em Proc. 43rd International Colloquium on Automata, Languages,
  and Programming (ICALP 2016)}, volume~55 of {\em LIPIcs}, pages 58:1--58:14,
  2016.

\bibitem[Mao21]{Mao21}
Xiao Mao.
\newblock Breaking the cubic barrier for (unweighted) tree edit distance.
\newblock In {\em Proc. 62nd IEEE Symposium on Foundations of Computer Science
  (FOCS 2021)}, pages 792--803, 2021.
\newblock Journal version: SIAM J. Comput., pp. FOCS21-195--FOCS21-223, 2023.

\bibitem[Mar71]{Mar71}
John~D. Markel.
\newblock {FFT} pruning.
\newblock {\em IEEE Transactions on Audio and Electroacoustics},
  19(4):305--311, 1971.

\bibitem[Mat91]{Mat91}
Ji{\v r}{\'\i} Matou{\v s}ek.
\newblock Computing dominances in {$E^n$}.
\newblock {\em Information Processing Letters}, 38(5):277--278, 1991.

\bibitem[{Mey}84]{MadH84}
Friedhelm {Meyer auf der Heide}.
\newblock A polynomial linear search algorithm for the $n$-dimensional knapsack
  problem.
\newblock {\em Journal of the ACM}, 31(3):668--676, 1984.

\bibitem[MS25]{MS25}
Erik M{\aa}rtensson and Paul {Stankovski Wagner}.
\newblock The number of the beast: Reducing additions in fast matrix
  multiplication algorithms for dimensions up to 666.
\newblock In {\em Proc. SIAM Conference on Applied and Computational Discrete
  Algorithms (ACDA 2025)}, pages 47--60, 2025.

\bibitem[NP85]{NP85}
Jaroslav Ne{\v s}et{\v r}il and Svatopluk Poljak.
\newblock On the complexity of the subgraph problem.
\newblock {\em Commentationes Mathematicae Universitatis Carolinae},
  26(2):415--419, 1985.

\bibitem[NPS{\etalchar{+}}25]{NPSVXY25}
Jakob Nogler, Adam Polak, Barna Saha, Virginia {Vassilevska Williams}, Yinzhan
  Xu, and Christopher Ye.
\newblock Faster weighted and unweighted tree edit distance and {APSP}
  equivalence.
\newblock In {\em Proc. 57th ACM Symposium on Theory of Computing (STOC 2025)},
  pages 2167--2178, 2025.

\bibitem[NVE{\etalchar{+}}25]{Nov25}
Alexander Novikov, Ng{\^a}n V{\~u}, Marvin Eisenberger, Emilien Dupont, Po-Sen
  Huang, Adam~Zsolt Wagner, Sergey Shirobokov, Borislav Kozlovskii, Francisco
  J.~R. Ruiz, Abbas Mehrabian, M.~Pawan Kumar, Abigail See, Swarat Chaudhuri,
  George Holland, Alex Davies, Sebastian Nowozin, Pushmeet Kohli, and Matej
  Balog.
\newblock {AlphaEvolve}: A coding agent for scientific and algorithmic
  discovery.
\newblock arXiv:2506.13131, 2025.

\bibitem[P{\u a}t10]{Pat10}
Mihai P{\u a}tra{\c s}cu.
\newblock Towards polynomial lower bounds for dynamic problems.
\newblock In {\em Proc. 42nd ACM Symposium on Theory of Computing (STOC 2010)},
  pages 603--610, 2010.

\bibitem[PPSZ05]{PPSZ05}
Ramamohan Paturi, Pavel Pudl{\'a}k, Michael~E. Saks, and Francis Zane.
\newblock An improved exponential-time algorithm for {$k$}-{SAT}.
\newblock {\em Journal of the ACM}, 52(3):337--364, 2005.

\bibitem[PS14]{PS14}
Rasmus Pagh and Morten St{\"o}ckel.
\newblock The input/output complexity of sparse matrix multiplication.
\newblock In {\em Proc. 22nd European Symposium on Algorithms (ESA 2014)},
  volume 8737 of {\em LNCS}, pages 750--761, 2014.

\bibitem[RV11]{RV11}
Liam Roditty and Virginia {Vassilevska Williams}.
\newblock Minimum weight cycles and triangles: Equivalences and algorithms.
\newblock In {\em Proc. 52nd IEEE Symposium on Foundations of Computer Science
  (FOCS 2011)}, pages 180--189, 2011.

\bibitem[RV12]{RV12}
Liam Roditty and Virginia {Vassilevska Williams}.
\newblock Subquadratic time approximation algorithms for the girth.
\newblock In {\em Proc. 23rd ACM-SIAM Symposium on Discrete Algorithms (SODA
  2012)}, pages 833--845, 2012.

\bibitem[RZ04]{RZ04}
Liam Roditty and Uri Zwick.
\newblock On dynamic shortest paths problems.
\newblock In {\em Proc. 12th European Symposium on Algorithms (ESA 2004)},
  volume 3221 of {\em LNCS}, pages 580--591, 2004.
\newblock Journal version: Algorithmica 61(2):389--401, 2011.

\bibitem[San04]{San04}
Piotr Sankowski.
\newblock Dynamic transitive closure via dynamic matrix inverse.
\newblock In {\em Proc. 45th IEEE Symposium on Foundations of Computer Science
  (FOCS 2004)}, pages 509--517, 2004.

\bibitem[SB93]{SB93}
Henrik~V. Sorensen and C.~Sidney Burrus.
\newblock Efficient computation of the {DFT} with only a subset of input or
  output points.
\newblock {\em IEEE Transactions on Signal Processing}, 41(3):1184--1200, 1993.

\bibitem[Sch81]{Sch81}
Arnold Sch{\"o}nhage.
\newblock Partial and total matrix multiplication.
\newblock {\em SIAM Journal on Computing}, 10(3):434--455, 1981.

\bibitem[Sch99]{Sch99}
Uwe Sch{\"o}ning.
\newblock A probabilistic algorithm for {$k$}-{SAT} and constraint satisfaction
  problems.
\newblock In {\em Proc. 40th IEEE Symposium on Foundations of Computer Science
  (FOCS 1999)}, pages 410--414, 1999.

\bibitem[Sei95]{Sei95}
Raimund Seidel.
\newblock On the all-pairs-shortest-path problem in unweighted undirected
  graphs.
\newblock {\em Journal of Computer and System Sciences}, 51(3):400--403, 1995.

\bibitem[SEO03]{SEO03}
Michael Soss, Jeff Erickson, and Mark~H. Overmars.
\newblock Preprocessing chains for fast dihedral rotations is hard or even
  impossible.
\newblock {\em Computational Geometry: Theory and Applications},
  26(3):235--246, 2003.

\bibitem[Str69]{Str69}
Volker Strassen.
\newblock Gaussian elimination is not optimal.
\newblock {\em Numerische Mathematik}, 13:354--356, 1969.

\bibitem[Str73]{Str73}
Volker Strassen.
\newblock Vermeidung von {D}ivisionen.
\newblock {\em Journal f{\"u}r die reine und angewandte Mathematik},
  264:184--202, 1973.

\bibitem[SVX27]{SVX27}
Nathan Sheffield, Virginia {Vassilevska Williams}, and Zoe Xi.
\newblock The limits of black-box reductions for all-pairs triangle detection.
\newblock In {\em Proc. 38th ACM-SIAM Symposium on Discrete Algorithms (SODA
  2027)}, 2027.
\newblock To appear. Full version: arXiv:2608.19092.

\bibitem[Tak92]{Tak92}
Tadao Takaoka.
\newblock A new upper bound on the complexity of the all pairs shortest path
  problem.
\newblock {\em Information Processing Letters}, 43(4):195--199, 1992.

\bibitem[Tak04]{Tak04}
Tadao Takaoka.
\newblock A faster algorithm for the all-pairs shortest path problem and its
  application.
\newblock In {\em Proc. 10th International Computing and Combinatorics
  Conference (COCOON 2004)}, volume 3106 of {\em LNCS}, pages 278--289, 2004.

\bibitem[Tak05]{Tak05}
Tadao Takaoka.
\newblock An {$O(n^3\log\log n/\log n)$} time algorithm for the all-pairs
  shortest path problem.
\newblock {\em Information Processing Letters}, 96(5):155--161, 2005.

\bibitem[TT00]{TT00}
Hisao Tamaki and Takeshi Tokuyama.
\newblock Algorithms for the maximum subarray problem based on matrix
  multiplication.
\newblock {\em Interdisciplinary Information Sciences}, 6(2):99--104, 2000.
\newblock Preliminary version in SODA 1998.

\bibitem[{Vas}18]{Vas18}
Virginia {Vassilevska Williams}.
\newblock On some fine-grained questions in algorithms and complexity.
\newblock In {\em Proceedings of the International Congress of Mathematicians
  (ICM 2018)}, pages 3447--3487. World Scientific, 2018.

\bibitem[vdBNS19]{BNS19}
Jan van~den Brand, Danupon Nanongkai, and Thatchaphol Saranurak.
\newblock Dynamic matrix inverse: Improved algorithms and matching conditional
  lower bounds.
\newblock In {\em Proc. 60th IEEE Symposium on Foundations of Computer Science
  (FOCS 2019)}, pages 456--480, 2019.
\newblock Full version: arXiv:1905.05067.

\bibitem[vdBSZ24]{BSZ24}
Jan van~den Brand, Zhao Song, and Tianyi Zhou.
\newblock Algorithm and hardness for dynamic attention maintenance in large
  language models.
\newblock In {\em Proc. 41st International Conference on Machine Learning (ICML
  2024)}, volume 235 of {\em PMLR}, pages 49008--49028, 2024.

\bibitem[VW10]{VW10}
Virginia {Vassilevska Williams} and Ryan Williams.
\newblock Subcubic equivalences between path, matrix and triangle problems.
\newblock In {\em Proc. 51st IEEE Symposium on Foundations of Computer Science
  (FOCS 2010)}, pages 645--654, 2010.

\bibitem[VW13]{VW13}
Virginia {Vassilevska Williams} and Ryan Williams.
\newblock Finding, minimizing, and counting weighted subgraphs.
\newblock {\em SIAM Journal on Computing}, 42(3):831--854, 2013.

\bibitem[VW18]{VW18}
Virginia {Vassilevska Williams} and R.~Ryan Williams.
\newblock Subcubic equivalences between path, matrix, and triangle problems.
\newblock {\em Journal of the ACM}, 65(5):27:1--27:38, 2018.

\bibitem[VX20]{VX20a}
Virginia {Vassilevska Williams} and Yinzhan Xu.
\newblock Monochromatic triangles, triangle listing and {APSP}.
\newblock In {\em Proc. 61st IEEE Symposium on Foundations of Computer Science
  (FOCS 2020)}, pages 786--797, 2020.
\newblock Full version: arXiv:2007.09318, whose numbering we follow.

\bibitem[VXXZ24]{VXXZ24}
Virginia {Vassilevska Williams}, Yinzhan Xu, Zixuan Xu, and Renfei Zhou.
\newblock New bounds for matrix multiplication: From alpha to omega.
\newblock In {\em Proc. 35th ACM-SIAM Symposium on Discrete Algorithms (SODA
  2024)}, pages 3792--3835, 2024.

\bibitem[War62]{War62}
Stephen Warshall.
\newblock A theorem on {B}oolean matrices.
\newblock {\em Journal of the ACM}, 9(1):11--12, 1962.

\bibitem[Wil05]{Wil05}
Ryan Williams.
\newblock A new algorithm for optimal 2-constraint satisfaction and its
  implications.
\newblock {\em Theoretical Computer Science}, 348(2--3):357--365, 2005.

\bibitem[Wil14]{Wil14}
R.~Ryan Williams.
\newblock The polynomial method in circuit complexity applied to algorithm
  design (invited talk).
\newblock In {\em Proc. 34th International Conference on Foundation of Software
  Technology and Theoretical Computer Science (FSTTCS 2014)}, volume~29 of {\em
  LIPIcs}, pages 47--60, 2014.

\bibitem[Wil16]{Wil16}
R.~Ryan Williams.
\newblock Strong {ETH} breaks with {M}erlin and {A}rthur: Short non-interactive
  proofs of batch evaluation.
\newblock In {\em Proc. 31st Computational Complexity Conference (CCC 2016)},
  volume~50 of {\em LIPIcs}, pages 2:1--2:17, 2016.

\bibitem[Wil18]{Wil18}
R.~Ryan Williams.
\newblock Faster all-pairs shortest paths via circuit complexity.
\newblock {\em SIAM Journal on Computing}, 47(5):1965--1985, 2018.

\bibitem[Wil24]{Wil24}
R.~Ryan Williams.
\newblock The orthogonal vectors conjecture and non-uniform circuit lower
  bounds.
\newblock In {\em Proc. 65th IEEE Symposium on Foundations of Computer Science
  (FOCS 2024)}, pages 1372--1387, 2024.
\newblock ECCC TR24-142.

\bibitem[Yat37]{Yat37}
Frank Yates.
\newblock The design and analysis of factorial experiments.
\newblock Technical Communication~35, Imperial Bureau of Soil Science,
  Harpenden, 1937.

\bibitem[Yu18]{Yu18}
Huacheng Yu.
\newblock An improved combinatorial algorithm for {B}oolean matrix
  multiplication.
\newblock {\em Information and Computation}, 261:240--247, 2018.

\bibitem[YZ05a]{YZ05b}
Raphael Yuster and Uri Zwick.
\newblock Answering distance queries in directed graphs using fast matrix
  multiplication.
\newblock In {\em Proc. 46th IEEE Symposium on Foundations of Computer Science
  (FOCS 2005)}, pages 389--396, 2005.

\bibitem[YZ05b]{YZ05}
Raphael Yuster and Uri Zwick.
\newblock Fast sparse matrix multiplication.
\newblock {\em ACM Transactions on Algorithms}, 1(1):2--13, 2005.

\bibitem[Zwi02]{Zwi02}
Uri Zwick.
\newblock All pairs shortest paths using bridging sets and rectangular matrix
  multiplication.
\newblock {\em Journal of the ACM}, 49(3):289--317, 2002.

\bibitem[Zwi06]{Zwi06}
Uri Zwick.
\newblock A slightly improved sub-cubic algorithm for the all pairs shortest
  paths problem with real edge lengths.
\newblock {\em Algorithmica}, 46(2):181--192, 2006.
\newblock Preliminary version in ISAAC 2004.

\end{thebibliography}
